\documentclass[11pt,letterpaper]{article}

\def\showauthornotes{0}
\def\showtableofcontents{1}
\def\showkeys{0}
\def\showdraftbox{0}
\def\showcolorlinks{1}
\def\usemicrotype{1}
\def\showfixme{0}
\def\writemode{1}
\def\arxivmode{1}

\usepackage{etex}

\usepackage[l2tabu, orthodox]{nag}

\usepackage{xspace,enumerate}

\usepackage[dvipsnames]{xcolor}

\usepackage[T1]{fontenc}
\usepackage[full]{textcomp}

\usepackage[american]{babel}

\usepackage{mathtools}

\usepackage{amsthm}

\usepackage{algorithm}
\usepackage{algorithmic}
\newtheorem{theorem}{Theorem}[section]
\newtheorem*{theorem*}{Theorem}

\newtheorem{proposition}[theorem]{Proposition}
\newtheorem*{proposition*}{Proposition}
\newtheorem{lemma}[theorem]{Lemma}
\newtheorem*{lemma*}{Lemma}
\newtheorem{corollary}[theorem]{Corollary}
\newtheorem*{conjecture*}{Conjecture}

\newtheorem*{fact*}{Fact}

\newtheorem*{hypothesis*}{Hypothesis}

\theoremstyle{definition}
\newtheorem{definition}[theorem]{Definition}

\theoremstyle{remark}

\newtheorem*{claim*}{Claim}
\newtheorem{remark}[theorem]{Remark}
\newtheorem*{remark*}{Remark}

\newtheorem*{observation*}{Observation}

\ifnum\writemode=1
\usepackage[
letterpaper,
top=1.2in,
bottom=1.2in,
left=1in,
right=1in]{geometry}

\fi

\ifnum\writemode=0
\usepackage[
letterpaper,
top=0.7in,
bottom=0.9in,
left=1in,
right=1in]{geometry}
\fi

\ifnum\arxivmode=0
\usepackage{newpxtext} %
\usepackage{textcomp} %
\usepackage[varg,bigdelims]{newpxmath}
\usepackage[scr=rsfso]{mathalfa}%
\usepackage{bm} %
\let\mathbb\varmathbb
\fi

\ifnum\arxivmode=1
\usepackage[varg]{pxfonts} %
\usepackage{textcomp} %
\usepackage[scr=rsfso]{mathalfa}%
\usepackage{bm} %
\fi

\ifnum\showkeys=1
\usepackage[color]{showkeys}
\fi

\ifnum\showcolorlinks=1
\usepackage[
pagebackref,
colorlinks=true,
urlcolor=blue,
linkcolor=blue,
citecolor=OliveGreen,
]{hyperref}
\fi

\ifnum\showcolorlinks=0
\usepackage[
pagebackref,
colorlinks=false,
pdfborder={0 0 0}
]{hyperref}
\fi

\usepackage{prettyref}

\newcommand{\savehyperref}[2]{\texorpdfstring{\hyperref[#1]{#2}}{#2}}

\newrefformat{eq}{\savehyperref{#1}{\textup{(\ref*{#1})}}}
\newrefformat{eqn}{\savehyperref{#1}{\textup{(\ref*{#1})}}}
\newrefformat{lem}{\savehyperref{#1}{Lemma~\ref*{#1}}}
\newrefformat{def}{\savehyperref{#1}{Definition~\ref*{#1}}}
\newrefformat{thm}{\savehyperref{#1}{Theorem~\ref*{#1}}}
\newrefformat{cor}{\savehyperref{#1}{Corollary~\ref*{#1}}}
\newrefformat{cha}{\savehyperref{#1}{Chapter~\ref*{#1}}}
\newrefformat{sec}{\savehyperref{#1}{Section~\ref*{#1}}}
\newrefformat{app}{\savehyperref{#1}{Appendix~\ref*{#1}}}
\newrefformat{tab}{\savehyperref{#1}{Table~\ref*{#1}}}
\newrefformat{fig}{\savehyperref{#1}{Figure~\ref*{#1}}}
\newrefformat{hyp}{\savehyperref{#1}{Hypothesis~\ref*{#1}}}
\newrefformat{alg}{\savehyperref{#1}{Algorithm~\ref*{#1}}}
\newrefformat{rem}{\savehyperref{#1}{Remark~\ref*{#1}}}
\newrefformat{item}{\savehyperref{#1}{Item~\ref*{#1}}}
\newrefformat{step}{\savehyperref{#1}{step~\ref*{#1}}}
\newrefformat{conj}{\savehyperref{#1}{Conjecture~\ref*{#1}}}
\newrefformat{fact}{\savehyperref{#1}{Fact~\ref*{#1}}}
\newrefformat{prop}{\savehyperref{#1}{Proposition~\ref*{#1}}}
\newrefformat{prob}{\savehyperref{#1}{Problem~\ref*{#1}}}
\newrefformat{claim}{\savehyperref{#1}{Claim~\ref*{#1}}}
\newrefformat{relax}{\savehyperref{#1}{Relaxation~\ref*{#1}}}
\newrefformat{red}{\savehyperref{#1}{Reduction~\ref*{#1}}}
\newrefformat{part}{\savehyperref{#1}{Part~\ref*{#1}}}

\newcommand{\Sref}[1]{\hyperref[#1]{\S\ref*{#1}}}

\usepackage{nicefrac}

\ifnum\usemicrotype=1
\usepackage{microtype}
\fi

\ifnum\showauthornotes=1
\newcommand{\Authornote}[2]{{\sffamily\small\color{red}{[#1: #2]}}}
\newcommand{\Authornotecolored}[3]{{\sffamily\small\color{#1}{[#2: #3]}}}
\newcommand{\Authorcomment}[2]{{\sffamily\small\color{gray}{[#1: #2]}}}
\newcommand{\Authorstartcomment}[1]{\sffamily\small\color{gray}[#1: }

\newcommand{\Authorfnote}[2]{\footnote{\color{red}{#1: #2}}}
\newcommand{\Authorfixme}[1]{\Authornote{#1}{\textbf{??}}}
\newcommand{\Authormarginmark}[1]{\marginpar{\textcolor{red}{\fbox{\Large #1:!}}}}
\else
\newcommand{\Authornote}[2]{}
\newcommand{\Authornotecolored}[3]{}
\newcommand{\Authorcomment}[2]{}
\newcommand{\Authorstartcomment}[1]{}

\newcommand{\Authorfnote}[2]{}
\newcommand{\Authorfixme}[1]{}
\newcommand{\Authormarginmark}[1]{}
\fi

\newcommand{\Dnote}{\Authornote{D}}
\newcommand{\TMnote}{\Authornotecolored{blue}{T}}

\newcommand{\Dcomment}{\Authorcomment{D}}

\definecolor{forestgreen(traditional)}{rgb}{0.0, 0.27, 0.13}
\definecolor{brightgreen}{rgb}{0.0, 0.87, 0.13}
\newcommand{\Jnote}{\Authornotecolored{brightgreen}{J}}

\ifnum\showfixme=0

\fi

\usepackage{boxedminipage}

\newcommand{\paren}[1]{(#1)}
\newcommand{\Paren}[1]{\left(#1\right)}

\newcommand{\Brac}[1]{\left[#1\right]}

\newcommand{\abs}[1]{\lvert#1\rvert}

\newcommand{\card}[1]{\lvert#1\rvert}

\newcommand{\Set}[1]{\left\{#1\right\}}

\newcommand{\norm}[1]{\lVert#1\rVert}
\newcommand{\Norm}[1]{\left\lVert#1\right\rVert}

\newcommand{\Bignorm}[1]{\Big\lVert#1\Big\rVert}

\newcommand{\iprod}[1]{\langle#1\rangle}

\newcommand{\Esymb}{\mathbb{E}}
\newcommand{\Psymb}{\mathbb{P}}

\DeclareMathOperator*{\E}{\Esymb}

\DeclareMathOperator*{\ProbOp}{\Psymb}

\renewcommand{\Pr}{\ProbOp}

\newcommand{\textparen}[1]{\text{(#1)}}

\ifx\because\undefined
\newcommand{\because}[1]{\textparen{because #1}}
\else
\renewcommand{\because}[1]{\textparen{because #1}}
\fi

\newcommand{\super}[2]{#1^{\paren{#2}}}

\newcommand{\vbig}{\vphantom{\bigoplus}}

\newcommand{\defeq}{\stackrel{\mathrm{def}}=}

\newcommand{\seteq}{\mathrel{\mathop:}=}

\newcommand{\from}{\colon}

\newcommand{\Mid}{\;\middle\vert\;}

\newcommand{\mper}{\,.}
\newcommand{\mcom}{\,,}

\newcommand\bdot\bullet

\ifx\mathds\undefined %

\else
}
\fi

\DeclareMathOperator{\poly}{poly}

\DeclareMathOperator{\supp}{supp}
\DeclareMathOperator{\dist}{dist}

\newcommand{\N}{\mathbb N}
\newcommand{\R}{\mathbb R}

\newcommand{\cA}{\mathcal A}
\newcommand{\cB}{\mathcal B}
\newcommand{\cC}{\mathcal C}

\newcommand{\cH}{\mathcal H}

\newcommand{\cM}{\mathcal M}
\newcommand{\cN}{\mathcal N}

\newcommand{\cR}{\mathcal R}

\renewcommand{\le}{\leqslant}

\renewcommand{\ge}{\geqslant}

\ifnum\showdraftbox=1
\newcommand{\draftbox}{\begin{center}
  \fbox{%
    \begin{minipage}{2in}%
      \begin{center}%
          \Large\textsc{Working Draft}\\%
        Please do not distribute%
      \end{center}%
    \end{minipage}%
  }%
\end{center}
\vspace{0.2cm}}
\else
\newcommand{\draftbox}{}
\fi

\let\epsilon=\varepsilon

\numberwithin{equation}{section}

\newcommand\MYcurrentlabel{xxx}

\newcommand{\MYstore}[2]{%
  \global\expandafter \def \csname MYMEMORY #1 \endcsname{#2}%
}

\newcommand{\MYload}[1]{%
  \csname MYMEMORY #1 \endcsname%
}

\newcommand{\MYnewlabel}[1]{%
  \renewcommand\MYcurrentlabel{#1}%
  \MYoldlabel{#1}%
}

\newcommand{\MYdummylabel}[1]{}

\newcommand{\torestate}[1]{%
  \let\MYoldlabel\label%
  \let\label\MYnewlabel%
  #1%
  \MYstore{\MYcurrentlabel}{#1}%
  \let\label\MYoldlabel%
}

\newcommand{\restatetheorem}[1]{%
  \let\MYoldlabel\label
  \let\label\MYdummylabel
  \begin{theorem*}[Restatement of \prettyref{#1}]
    \MYload{#1}
  \end{theorem*}
  \let\label\MYoldlabel
}

\newcommand{\restatelemma}[1]{%
  \let\MYoldlabel\label
  \let\label\MYdummylabel
  \begin{lemma*}[Restatement of \prettyref{#1}]
    \MYload{#1}
  \end{lemma*}
  \let\label\MYoldlabel
}

\newcommand{\restateprop}[1]{%
  \let\MYoldlabel\label
  \let\label\MYdummylabel
  \begin{proposition*}[Restatement of \prettyref{#1}]
    \MYload{#1}
  \end{proposition*}
  \let\label\MYoldlabel
}

\newcommand{\restatefact}[1]{%
  \let\MYoldlabel\label
  \let\label\MYdummylabel
  \begin{fact*}[Restatement of \prettyref{#1}]
    \MYload{#1}
  \end{fact*}
  \let\label\MYoldlabel
}

\newcommand{\restate}[1]{%
  \let\MYoldlabel\label
  \let\label\MYdummylabel
  \MYload{#1}
  \let\label\MYoldlabel
}

\newcommand{\addreferencesection}{
  \phantomsection
  \addcontentsline{toc}{section}{References}
}

\newcommand{\e}{\epsilon}
\newcommand{\eps}{\epsilon}

\let\origparagraph\paragraph
\renewcommand{\paragraph}[1]{\origparagraph{#1.}}

\allowdisplaybreaks

\usepackage{paralist}

\usepackage{comment}
\usepackage{relsize}
\usepackage[font=footnotesize]{caption}
\usepackage{letltxmacro}

\DeclareMathOperator{\Id}{\mathrm{Id}}

\DeclareUrlCommand\email{}
\newcommand{\whp}{\text{w.h.p.}\xspace}

\newcommand{\inner}[1]{\langle #1 \rangle}

\DeclareMathOperator{\tE}{\tilde{\mathbb E}}
\DeclareMathOperator{\tO}{\tilde{O}}

\DeclareMathOperator*{\pE}{\tilde{\mathbb E}}

\newcommand{\ot}{\otimes}

\let\pref=\prettyref

\newcommand{\trans}{\mkern-1mu\mathsf{T}}
\newcommand*{\transpose}[1]{{#1}\mathclose{\vphantom{#1}}^{\trans}}

\newcommand*{\dyad}[1]{#1#1\mathclose{\vphantom{#1}}^{\trans}}

\DeclareSymbolFont{mysymbols}{OMS}{cmsy}{b}{n}
\DeclareMathSymbol{\proves}{\mathrel}{mysymbols}{96}
\let\vdash\proves

\DeclareMathSymbol{\bettertop}{\mathrel}{mysymbols}{62}

\title{Polynomial-time tensor decompositions with sum-of-squares}

\author{%
Tengyu Ma\thanks{Princeton University. \protect\email{tengyu@cs.princeton.edu}.
Supported by Simons Award in Theoretical Computer Science, IBM PhD Fellowship, Dodds Fellowship and Siebel Scholarship.}
\and
Jonathan Shi\thanks{Cornell University, \protect\email{jshi@cs.cornell.edu}.
Supported by David Steurer's NSF CAREER award.}
\and
David Steurer\thanks{Cornell University, \protect\email{dsteurer@cs.cornell.edu}.
Supported by a Microsoft Research Fellowship, a Alfred P. Sloan Fellowship, an NSF CAREER award, and the Simons Collaboration for Algorithms and Geometry.}}

\begin{document}

\maketitle
\draftbox
\thispagestyle{empty}

\begin{abstract}
We give new algorithms based on the sum-of-squares method for tensor decomposition.
Our results improve the best known running times from quasi-polynomial to polynomial for several problems, including decomposing random overcomplete 3-tensors and learning overcomplete dictionaries with constant relative sparsity.
We also give the first robust analysis for decomposing overcomplete $4$-tensors in the smoothed analysis model.

A key ingredient of our analysis is to establish small spectral gaps in moment matrices derived from solutions to sum-of-squares relaxations.
To enable this analysis we augment sum-of-squares relaxations with spectral analogs of  maximum entropy constraints.
\end{abstract}

\clearpage

\ifnum\showtableofcontents=1
{
\tableofcontents
\thispagestyle{empty}
 }
\fi

\clearpage

\setcounter{page}{1}

\section{Introduction}

Tensors are arrays of (real) numbers with multiple indices---generalizing matrices (two indices) and vectors (one index) in a natural way.
They arise in many different contexts, e.g., moments of multivariate distributions, higher-order derivatives of multivariable functions, and coefficients of multivariate polynomials.
An important ongoing research effort aims to extend algorithmic techniques for vectors and matrices to more general tensors.
A key challenge is that many tractable matrix computations (like rank and spectral norm) become NP-hard in the tensor setting (even for just three indices) \cite{H90,HL13}.
However, recent work gives evidence that it is possible to avoid this computational intractability and develop provably efficient algorithms, especially for low-rank tensor decompositions,
by making suitable assumptions about the input and allowing for approximations
\cite{DBLP:conf/colt/AnandkumarGJ15, %
  DBLP:journals/corr/AnandkumarGJ14b, %
  DBLP:conf/approx/GeM15,
  DBLP:conf/colt/HopkinsSS15,
  HopkinsSSS16}.
These algorithms lead to the best known provable guarantees for a wide range of unsupervised learning problems
\cite{DBLP:journals/jmlr/AnandkumarGHKT14, %
  DBLP:conf/stoc/BhaskaraCMV14, %
  DBLP:conf/stoc/GoyalVX14, %
  DBLP:journals/jmlr/AnandkumarGHK14%
},
including learning mixtures of Gaussians \cite{DBLP:conf/stoc/GeHK15},
Latent Dirichlet topic modeling \cite{DBLP:journals/algorithmica/AnandkumarFHKL15},
and dictionary learning~\cite{DBLP:conf/stoc/BarakKS15}.
Low-rank tensor decompositions are useful for these learning problems because they are often unique up to permuting the factors---in contrast, low-rank matrix factorizations are unique only up to unitary transformation.
In fact, as far as we are aware, in all natural situations where finding low-rank tensor decompositions is tractable, the decompositions are also unique.

We consider the following (symmetric) version of the tensor decomposition problem:
Let $a_1,\ldots,a_n\in \R^d$ be $d$-dimensional unit vectors.
We are given (approximate) access to the first $k$ moments $\cM_1,\ldots,\cM_k$ of the uniform distribution over $a_1,\ldots,a_n$, that is,
\begin{equation}
  \label{eq:moments}
  \cM_t = \tfrac 1n \sum_{i=1}^n a_i^{\otimes t}\quad \text{for $t\in \{1,\ldots,k\}$}\mper
\end{equation}
The goal is to approximately recover the vectors $a_1,\ldots,a_n$.
What conditions on the vectors $a_1,\ldots,a_n$ and the number of moments $k$ allow us to efficiently and robustly solve this problem?

A classical algorithm based on (simultaneous) matrix diagonalization \cite[attributed to Jennrich]{harshman1970foundations,zbMATH00446923} shows that whenever the vectors $a_1,\ldots,a_n$ are linearly independent, $k=3$ moments suffice to recover the vectors in polynomial time.
(This algorithm is also robust against polynomially small errors in the input moment tensors \cite{DBLP:conf/alt/AnandkumarGHKT15, DBLP:conf/stoc/GoyalVX14, DBLP:conf/stoc/BhaskaraCMV14}.)
Therefore an important remaining algorithmic challenge for tensor decomposition is the \emph{overcomplete} case, when the number of vectors (significantly) exceeds their dimension.
Several recent works studied this case with different assumptions on the vectors and the number of moments.
In this work, we give a unified algorithmic framework for overcomplete tensor decomposition that achieves---and in many cases surpasses---the previous best guarantees for polynomial-time algorithms.

In particular, some decompositions that previously required quasi-polynomial time to find are reduced to polynomial time in our framework, including the case of general tensors with order logarithmically large in its overcompleteness $n/d$ \cite{DBLP:conf/stoc/BarakKS15} and random order-3 tensors with rank $n \le d^{3/2}/\log^{O(1)}(d)$ \cite{DBLP:conf/approx/GeM15}.
Iterative methods may also achieve fast local convergence guarantees for incoherent order-3 tensors with rank $o(d^{3/2})$, which become global convergence guarantees under no more than constant overcompleteness \cite{DBLP:journals/jmlr/AnandkumarGHKT14}.
In the smoothed analysis model, where each vector of the desired decomposition is assumed to have been randomly perturbed by an inverse polynomial amount, polynomial-time decomposition was achieved for order-5 tensors of rank up to $d^2/2$ \cite{DBLP:conf/stoc/BhaskaraCMV14}.
Our framework extends this result to order-4 tensors, for which the corresponding analysis was previously unknown for any superconstant overcompleteness.

The starting point of our work is a new analysis of the aforementioned matrix diagonalization algorithm that works for the case when $a_1,\ldots,a_n$ are linearly independent.
A key ingredient of our analysis is a powerful and by now standard concentration bound for Gaussian matrix series \cite{zbMATH05946839,DBLP:journals/focm/Tropp12}.
An important feature of our analysis is that it is captured by the sum-of-squares (SoS) proof system in a robust way.
This fact allows us to use Jennrich's algorithm as a rounding procedure for sum-of-squares relaxations of tensor decomposition, which is the key idea behind improving previous quasi-polynomial time algorithms based on these relaxations \cite{DBLP:conf/stoc/BarakKS15,DBLP:conf/approx/GeM15}.

The main advantage that sum-of-squares relaxations afford for tensor decomposition is that they allow us to efficiently hallucinate faithful \emph{higher-degree moments} for a distribution given only its lower-degree moments.
We can now run classical tensor decomposition algorithms like Jennrich's on these hallucinated higher-degree moments (akin to \emph{rounding}).
The goal is to show that those algorithms work as well as they would on the true higher moments.
What is challenging about it is that the analysis of Jennrich's algorithm relies on small spectral gaps that are difficult to reason about in the sum-of-squares setting.
(Previous sum-of-squares based methods for tensor decomposition also followed this outline but used simpler, more robust rounding algorithms which required quasi-polynomial time.)

To this end, we view solutions to sum-of-squares relaxations as \emph{pseudo-distributions}, which generalize classical probability distributions in a way that takes computational efficiency into account.\footnote{In particular, the set of constant-degree moments of $n$-variate pseudo-distributions admits an $n^{O(1)}$-time separation oracle based on computing eigenvectors.}
More concretely, pseudo-distributions are indistinguishable from actual distributions with respect to tests captured by a restricted system of proofs, called \emph{sum-of-squares proofs}.

An interesting feature of how we use pseudo-distributions is that our relaxations search for  pseudo-distributions of large \emph{entropy} (via an appropriate surrogate).
This objective is surprising, because when we consider convex relaxations of NP-hard search problems, the intended solutions typically correspond to atomic distributions which have entropy $0$.
Here, high entropy in the pseudo-distribution allows us to ensure that rounding results in a useful solution.
This appears to be related to the way in which many randomized rounding procedures use maximum-entropy distributions~\cite{gharan2014new}, but differs in that the aforementioned rounding procedures focus on the entropy of the rounding process rather than the entropy (surrogate) of the solution to the convex relaxation.
A measure of ``entropy'' has also been directly ascribed to pseudo-distributions previously \cite{Lee:2015:LBS:2746539.2746599}, and the principle of maximum entropy has been applied to pseudo-distributions as well \cite{BHKKMP16}, but these have previously occurred separately, and our application is the first to encode a surrogate notion of entropy directly into the sum-of-squares proof system.

\medskip

Our work also takes inspiration from a recent work that uses sum-of-squares techniques to design fast spectral algorithms for a range of problems including tensor decomposition \cite{HopkinsSSS16}.
Their algorithm also proceeds by constructing surrogates for higher moments and applying a classical tensor decomposition algorithm on these surrogates.
The difference is that the surrogates in \cite{HopkinsSSS16} are explicitly constructed as low-degree polynomial of the input tensor, whereas our surrogates are computed by sum-of-squares relaxations.
The explicit surrogates of \cite{HopkinsSSS16} allow for a direct (but involved) analysis through concentration bounds for matrix polynomials.
In our case, a direct analysis is not possible because we have very little control over the surrogates computed by sum-of-squares relaxations.
Therefore, the challenge for us is to understand to what extent classical tensor decomposition algorithms are compatible with the sum-of-squares proof system.
Our analysis ends up being less technically involved compared to \cite{HopkinsSSS16} (using the language of pseudo-distributions and sum-of-squares proofs).

\subsection{Results for tensor decomposition}
\label{sec:results}

Let $\{a_1,\ldots,a_n\}\subseteq\R^d$ be a set of unit vectors.
We study the task of approximately recovering this set of vectors given (noisy) access to its first $k$ moments \pref{eq:moments}.
We organize this overview of our results based on different kinds of assumptions imposed on the set $\{a_1,\ldots,a_n\}$ and the order of tensor/moments that we have access to.
All of our algorithms are randomized and may fail with some small probability over their internal randomness, say probability at most $0.01$.
(Standard arguments allow us to amplify this probability at the cost of a small increase in running time.)

\paragraph{Orthogonal vectors}

This scenario often captures the case of general linearly independent vectors because knowledge of the second moments of $a_1,\ldots,a_n$ allows us to orthonormalize the vectors (this process is sometimes called ``whitening'').
Many efficient algorithms are known in this case.
Our contribution here is in improving the error tolerance.
For a symmetric $3$-tensor $E\in(\R^{d})^{\otimes 3}$, we use $\norm{E}_{\{1\},\{2,3\}}$ to denote the spectral norm of $E$ as a $d$-by-$d^2$ matrix (using the first mode of $E$ to index rows and the last two modes of $E$ to index the columns).
This norm is at most $\sqrt d$ times the injective norm $\norm{E}_{\{1\},\{2\},\{3\}}$ (the maximum of $\langle  E,x\otimes y \otimes z \rangle$ over all unit vectors $x,y,z\in\R^d$).
The previous best error tolerance for this problem required the error tensor $E=T-\sum_{i=1}^n a_i^{\otimes 3}$ to have injective norm $\norm{E}_{\{1\},\{2\},\{3\}}\ll 1/d$.
Our algorithm requires only $\norm{E}_{\{1\},\{2,3\}}\ll 1$, which is satisfied in particular when $\norm{E}_{\{1\},\{2\},\{3\}}\ll 1/\sqrt d$.

\begin{theorem}
\torestate
{\label{thm:orthogonal-vectors}
  There exists a polynomial-time algorithm that given a symmetric $3$-tensor $T\in (\R^d)^{\otimes 3}$ outputs a set of vectors $\{a_1',\ldots,a_{n'}'\}\subseteq \R^d$ such that for every orthonormal set $\{a_1,\ldots,a_n\}\subseteq \R^d$, the Hausdorff distance\footnote{The Hausdorff distance $\dist_H(X,Y)$ between two finite sets $X$ and $Y$ measures the length of the largest gap between the two sets. Formally, $\dist_H(X,Y)$ is the maximum of $\max_{x\in X} \min_{y\in Y} \norm{x-y}$ and $\max_{y\in Y} \min_{x\in X} \norm{x-y}$.} between the two sets is at most
  \begin{equation}
    \dist_H\Paren{\Set{a_1,\ldots,a_n},\Set{a'_1,\ldots,a'_{n'}}}^2
    \le O(1)\cdot \Norm{T - \sum\nolimits_{i=1}^n a_i^{\otimes 3}}_{\{1\},\{2,3\}}\mper \label{eqn:48}
  \end{equation}}
\end{theorem}

Under the additional assumption $\norm{T - \sum\nolimits_{i=1}^n a_i^{\otimes 3}}_{\{1\},\{2,3\}}\le 1/\log d$, the running time of the algorithm can be improved to $O(d^{1+\omega})\le d^{3.33}$ using fast matrix multiplication,
where $\omega$ is the number such that two $n \times n$ matrices can be multiplied
together in time $n^{\omega}$ (See Theorem~\ref{thm:orthogonal-alg}).

It is also possible to replace the spectral norm $\norm{\cdot}_{\{1\},\{2,3\}}$ in the above theorem statement by constant-degree sum-of-squares relaxations of the injective norm of $3$-tensors. (See Remark~\ref{remark:sosnorm} for details. )
If the error $E$ has Gaussian distribution $\cN(0,\sigma^2\cdot \Id^{\otimes 3}_d)$, then this norm is \whp bounded by $\sigma\cdot d^{3/4}(\log d)^{O(1)}$ \cite{DBLP:conf/colt/HopkinsSS15}, whereas the norm $\norm{\cdot}_{\{1\},\{2,3\}}$ has magnitude $\Omega(\sigma\cdot d)$. We prove 
~\pref{thm:orthogonal-vectors} in Section~\ref{subsec:orthogonal}. 

\paragraph{Random vectors}
We consider the case that $a_1,\ldots,a_n$ are chosen independently at random from the unit sphere of $\R^d$.
For $n\le d$, this case is roughly equivalent to the case of orthonormal vectors.
Thus, we are interested in the ``overcomplete'' case $n \gg d$, when the rank is larger than the dimension.
Previous work found the decomposition in quasi-polynomial time when $n \le d^{3/2}/\log^{O(1)}d$~\cite{DBLP:conf/approx/GeM15}, or in time subquadratic
in the input size when $n \le d^{4/3}/\log^{O(1)}d$~\cite{HopkinsSSS16}.
Our polynomial-time algorithm therefore is an improvement when $n$ is between
$d^{4/3}$ and $d^{3/2}$ (up to logarithmic factors).

\begin{theorem}
\torestate{
\label{thm:random-vectors}
  There exists a polynomial-time algorithm $A$ such that with probability $1-d^{-\omega(1)}$ over the choice of random unit vectors $a_1,\ldots,a_n\in\R^d$, every symmetric $3$-tensor $T\in (\R^d)^{\otimes 3}$ satisfies
  \begin{equation}
    \label{eq:13}
    \dist_H\Paren{\vbig A(T),\Set{a_1,\ldots,a_{n}}}^2
    \le O\Paren{\Paren{\frac n{d^{1.5}}}^{\Omega(1)} + \Norm{T - \sum\nolimits_{i=1}^n a_i^{\otimes 3}}_{\{1\},\{2,3\}}}
    \mper
  \end{equation}
}
\end{theorem}
Again it is possible to replace the spectral norm $\norm{\cdot}_{\{1\},\{2,3\}}$ in the above theorem statement by constant-degree sum-of-squares relaxations of the injective norm of $3$-tensors, which as mentioned before give better bounds for Gaussian error tensors. We prove~\pref{thm:random-vectors} in Section~\ref{sec:random_tensor}. 

\paragraph{Smoothed vectors}
Next, we consider a more general setup where the vectors $a_1,\ldots,a_n\in \R^d$ are smoothed, i.e., randomly perturbed.
This scenario is significantly more general than  random vectors.
Again we are interested in the overcomplete case $n\gg d$.
The previous best work \cite{DBLP:conf/stoc/BhaskaraCMV14} showed that the fifth moment of smoothed vectors $a_1,\ldots,a_n$ with $n\le d^2/2$ is enough to approximately recover the vectors even in the presence of a polynomial amount of error.
For fourth moments of smoothed vectors, no such result was known even for lower overcompleteness, say $n=d^{1.01}$.

We give an interpretation of the $4$-tensor decomposition algorithm FOOBI\footnote{The FOOBI algorithm is known to work for overcomplete $4$-tensors when there is no error in the input.
Researchers \cite{DBLP:conf/stoc/BhaskaraCMV14} asked if this algorithm tolerates a polynomial amount of error.
Our work answers this question affirmatively for a variant of FOOBI (based on sum-of-squares).} \cite{DBLP:journals/tsp/LathauwerCC07} as a special case of a sum-of-squares based decomposition algorithm.
We show that the sum-of-squares based algorithm works in the smoothed setting even in the presence of a polynomial amount of error.
We define a condition number $\kappa(\cdot )$ for sets of vectors $a_1,\ldots,a_n\in\R^d$ (a polynomial in the condition number of two matrices, one with columns $\{a_{i}^{\otimes 2}\mid i\in [n]\}$ and one with columns $\{a_i\otimes (a_i\otimes a_j-a_j\otimes a_i)\otimes a_j\mid i\neq j\in [n]\}$).
First, we show that the algorithm can tolerate error $\ll 1/\kappa$ which could be independent of the dimension.
Concretely, our algorithm will output a set of vectors $\hat{a}_1,\dots, \hat{a}_n$ which will be close to $\{a_1,\dots, a_n\}$ up to permutations and sign flip with a relative error that scales linearly in the relative error of the input and the condition number $\kappa$.
Second, we show that for smoothed vectors this condition number is at least inverse polynomial with probability exponentially close to $1$.

\begin{theorem}
\label{thm:smoothed-vectors}
  There exists a polynomial-time algorithm such that for every symmetric $4$-tensor $T\in (\R^d)^{\otimes 4}$ and every set $\{a_1,\ldots,a_n\}\subseteq \R^d$ of vectors not necessarily unit length, there exists a permutation $\pi:[n] \to [n]$ so that the output $\{a_1',\ldots,a_n'\}$ of the algorithm on input $T$ satisfies
  \begin{equation}
    \max_{i\in [n]} \frac{\Norm{a_i-a_{\pi(i)}'}}{\Norm{a_i}}
    \le O(1)\cdot
    \frac{\Norm{T-\sum\nolimits_{i=1}^n a_i^{\otimes 4}}_{\{1,2\},\{3,4\}}}
    {\sigma_n\left(\sum\nolimits_{i=1}^n \dyad{(a_i^{\otimes 2})}\right)}
    \cdot \kappa(a_1,\ldots,a_n)
    \mcom
  \end{equation}
  where $\sigma_n(A)$ refers to the $n$th singular value of the matrix $A$, here the smallest non-zero singular value.
\end{theorem}

We say that a distribution over vectors $a_1,\ldots,a_n\in \R^d$ is $\gamma$-smoothed if $a_i=a^0_i+ \gamma\cdot g_i$, where $a_1^0,\ldots,a_n^0$ are fixed vectors and $g_1,\ldots,g_n$ are independent Gaussian vectors from $\cN(0,\tfrac 1 d \Id_d)$.

\begin{theorem}
\label{thm:smoothed-vectors-2}
  Let $\e>0$ and $n,d\in \N$ with $n\le d^2/10$.
  Then, for any $\gamma$-smoothed distribution over vectors $a_1,\ldots,a_n$ in $\R^d$,
  \begin{displaymath}
    \Pr\Set{\vbig \kappa(a_1,\ldots,a_n) \le \poly(d,\gamma)} \ge 1- \exp(-d^{\Omega(1)})\mper
  \end{displaymath}
\end{theorem}

The above theorems together imply a polynomial-time algorithm for approximately decomposing overcomplete smoothed $4$-tensors even if the input error is polynomially large.
The error probability of the algorithm is exponentially small over the choice of the smoothing.
It is an interesting open problem to extend this result to overcomplete smoothed $3$-tensors, even for lower overcompleteness $n=d^{1.01}$. ~\pref{thm:smoothed-vectors} and ~\pref{thm:smoothed-vectors-2} are proved in Section~\ref{sec:foobi}. 

\paragraph{Separated unit vectors}

In the scenario, when inner products among the vectors $a_1,\ldots,a_n\in \R^d$ are bounded by $\rho<1$ in absolute value, the previous best decomposition algorithm shows that moments of order $(\log n)/\log \rho$ suffice \cite{DBLP:journals/corr/SchrammW15}.
Our algorithm requires moments of higher order (by a factor logarithmic in the desired accuracy) but in return tolerates up to constant spectral error.
This increased error tolerance also allows us to apply this result for dictionary learning with up to constant sparsity (see \pref{sec:applications}).

\begin{theorem}
\torestate{
\label{thm:separated-unit-vectors}
  There exists an algorithm $A$ with polynomial running time (in the size of its input) such that for all $\eta,\rho \in (0,1)$ and $\sigma\ge 1$, for every set of unit vectors $\{a_1,\ldots,a_n\}\subseteq \R^d$ with $\norm{\sum_{i=1}^n \dyad {a_i}}\le \sigma$ and $\max_{i\neq j}|\iprod{a_i,a_j}|\le \rho$, when the algorithm is given a symmetric $k$-tensor $T\in (\R^d)^{\otimes k}$ with $k\ge O\left(\frac{1+\log \sigma}{\log \rho}\right) \cdot \log(1/\eta)$, then its output $A(T)$ is a set of vectors $\{a_1',\ldots,a_{n'}'\}\subseteq \R^d$ such that
  \begin{equation}
    \dist_H\Paren{\vbig \{a_1'^{\otimes 2},\ldots,a_n'^{\otimes 2}\}, \{a_1^{\otimes 2},\ldots,a_n^{\otimes 2}\}}^2 \le O\Paren{\eta + \Norm{T - \sum\nolimits_{i=1}^n a_i^{\otimes k}}_{\{1,\ldots,\lfloor k/2 \rfloor\},\{\lfloor k/2\rfloor+1,\ldots,k\}}}\mper
  \end{equation}
}
\end{theorem}

We also show that a simple spectral algorithm with running time close to $d^{k}$ (the size of the input) achieves similar guarantees 
(see Remark~\ref{remark:extension}).
However, the error tolerance of this algorithm is in terms of an \emph{unbalanced} spectral norm: $\norm{T-\sum_{i=1}^n a_i^{\otimes k}}_{\{1,\ldots,k/3\},\{k/3+1,\ldots,k\}}$ (the spectral norm of the tensor viewed as a $d^{k/3}$-by-$d^{2k/3}$ matrix).
This norm is always larger than the balanced spectral norm in the theorem statement.
In particular, for dictionary learning applications, this norm is larger than $1$, which renders the guarantee of the simpler spectral algorithm vacuous in this case. We prove~\pref{thm:separated-unit-vectors} in Section~\ref{subsec:separated}. 

\paragraph{General unit vectors}
In this scenario, the number of moments that our algorithm requires is constant as long as $\sum_i \dyad {a_i}$ has constant spectral norm and the desired accuracy is constant.

\begin{theorem}
	\torestate{\label{thm:general-tensor-decomposition-general-components}
	There exists an algorithm $A$ (see \pref{alg:tensor-general-general-components}) with polynomial running time (in the size of its input) such that for all $\epsilon\in (0,1), \sigma \ge 1$, for every set of unit vectors $\{a_1,\ldots,a_n\}\subseteq \R^d$ with $\norm{\sum_{i=1}^n \dyad {a_i}}\le \sigma$ and every symmetric $2k$-tensor $T\in (\R^d)^{\otimes 2k}$ with $k\ge (1/\epsilon)^{O(1)}\cdot\log(\sigma)$ and $\Norm{T -\sum_i a_i^{\otimes 2k}}_{\{1,\ldots,k\},\{k+1,\ldots,2k\}}\le 1/3$, we have
	\begin{displaymath}
		\dist_H\Paren{\vbig A(T), \{a_1^{\otimes 2},\ldots,a_n^{\otimes 2}\}}^2
		\le O\Paren{\e}\mper
	\end{displaymath}}
\end{theorem}

The previous best algorithm for this problem required tensors of order $(\log \sigma)/\epsilon$ and had running time $d^{O((\log \sigma)/\epsilon^{O(1)}+\log n)}$ \cite[Theorem 4.3]{DBLP:conf/stoc/BarakKS15}.
We require the same order of the tensor and the runtime is improved to be polynomial in the size of the inputs (that is, $d^{\poly((\log \sigma)/\epsilon)})$. 

We also remark that a bit surprisingly we can handle $1/3$ error in spectral norm, and this is possible partly due to the choice of working with high order tensors. As a sanity check, we note that information-theoretically the components are identifiable: under the assumptions, the only vectors $u$ that satisfy $\inner{T, u^{\otimes 2k}} \ge 1/3$ are those vectors close to one of the $a_i$'s. 
We also note that the rounding algorithm of the sum-of-squares relaxation of this simple inefficient test requires a bit new idea beyond what we used previously. Here the difficulty is to make the runtime $d^{\poly((\log \sigma)/\epsilon)}$ instead of $d^{\poly(\sigma/\epsilon)}$. See Section~\ref{sec:rounding_general_components} for details. 
\paragraph{Spectral algorithms without sum-of-squares}
Finally, using a similar rounding technique directly on a orthogonal tensor (without using sum-of-squares and the pseudo-moment), we also obtain a fast and robust algorithm for orthogonal tensor decomposition. See Section~\ref{sec:nosos-orthogonal} for details. 

\subsection{Applications of tensor decomposition}
\label{sec:applications}

Tensor decomposition has a wide range of applications.
We focus here on learning sparse dictionaries, which is an example of the more general phenomenon of using tensor decomposition to learn latent variable models.
Here, we obtain the first polynomial-time algorithms that work in the overcomplete regime up to constant sparsity.

Dictionary learning is an important problem in multiple areas, ranging from computational neuroscience~\cite{OlshausenF97,OlshausenF96,OlshausenF96b}, machine learning~\cite{ArgyriouEP06,RanzatoBL2007}, to computer vision and image processing~\cite{EladA2006,MairalLBHP2008,YangWHY2008}.
The general goal is to find a good basis for given data.
More formally, in the dictionary learning problem, also known as sparse coding,
we are given samples of a random vector $y \in \R^n$, of the form $y= Ax$
where $A$ is some unknown matrix in $\R^{n \times m}$, called \emph{dictionary}, and
$x$ is sampled from an unknown distribution over sparse vectors.
The goal is to approximately recover the dictionary $A$.

We consider the same class of distributions over sparse vectors $\{x\}$ as \cite{DBLP:conf/stoc/BarakKS15}, which as discussed in \cite{DBLP:conf/stoc/BarakKS15} admits a wide-range of non-product distributions over sparse vectors.
(The case of product distributions reduces to the significantly easier problem of independent component analysis.)
We say that $\{x\}$ is \emph{$(k,\tau)$-nice} if $\E x_i^k = 1$ for every $i\in [m]$, $\E x_i ^{k/2}x_j^{k/2}\le \tau$ for all $i\neq j\in [m]$, and $\E x^\alpha=0$ for every non-square degree-$k$ monomial $x^\alpha$.
Here, $\tau$ is a measure of the relative sparsity of the vectors $\{x\}$.

We give an algorithm that for nice distributions solves the dictionary learning problem in polynomial time when the desired accuracy is constant, the overcompleteness of the dictionary is constant (measured by the spectral norm $\norm{A}$), and the sparsity parameter $\tau$ is a sufficiently small constant (depending only on the desired accuracy and $\norm{A}$).
The previous best algorithm \cite{DBLP:conf/stoc/BarakKS15} requires quasi-polynomial time in this setup (but works in polynomial-time for polynomial sparsity $\tau\le n^{-\Omega(1)}$).

\begin{theorem}
  There exists an algorithm $\cR$ parameterized by $\sigma \ge 1, \eta \in(0,1)$, such that for every dictionary $A\in \R^{n\times m}$ with $\Norm{A}\le \sigma$ and every $(k,\tau)$-nice distribution $\{x\}$ over $\R^m$ with $k \ge k(\eta, \sigma) = O((\log \sigma) /\eta)$ and $\tau \le \tau(k) = k^{-O(k)}$, the algorithm given $n^{O(k)}$ samples from $\{y=Ax\}$ outputs in time $n^{O(k)}$ vectors $a_1',\ldots,a'_m$ that are $O(\eta)^{1/2}$-close to the columns of $A$.
\end{theorem}

Since previous work~\cite{DBLP:conf/stoc/BarakKS15} provides a black box reduction from dictionary learning to tensor decomposition, the theorem above follows from \pref{thm:general-tensor-decomposition-general-components}.
Our \pref{thm:separated-unit-vectors} implies a dictionary learning algorithm with better parameters for the case that the columns of $A$ are separated.

\subsection{Polynomial optimization with few global optima}

Underlying our algorithms for the tensor decomposition is an algorithm for solving general systems of polynomial constraints with the property that the total number of different solutions is small and that there exists a short certificate for that fact in form of a sum-of-squares proof.

Let $\cA$ be a system of polynomial constraints over real variables $x=(x_1,\ldots,x_d)$ and let $P\from \R^d \to \R^{d^\ell}$ be a polynomial map of degree at most $\ell$---for example, $P(x)=x^{\otimes \ell}$.
We say that solutions $a_1,\ldots,a_n\in \R^d$ to $\cA$ are \emph{unique} under the map $P$ if the vectors  $P(a_1),\ldots,P(a_n)$ are orthonormal up to error $0.01$ (in spectral norm) and every solution $a$ to $\cA$ satisfies $P(a) \approx P(a_i)$ for some $i\in [n]$.
We encode this property algebraically by requiring that the constraints in $\cA$ imply the constraint $\sum_{i=1}^n \langle  P(a_i) , P(x) \rangle^4 \ge 0.99\cdot \norm{P(x)}^4$.
We say that the solutions $a_1,\ldots,a_n$ are \emph{$\ell$-certifiably unique} if in addition this implication has a degree-$\ell$ sum-of-squares proof.

The following theorem shows that if polynomial constraints have certifiably unique solutions (under a given map $P$), then we can find them efficiently (under the map $P$).

\begin{theorem}[Informal statement of \pref{thm:general-tensor-decomposition}]
  Given a system of polynomial constraints $\cA$ and a polynomial map $P$ such that there exists $\ell$-certifiably unique solutions $a_1,\ldots,a_n$ for $\cA$, we can find in time $d^{O(\ell)}$ vectors $0.1$-close to $P(a_1),\ldots,P(a_n)$ in Hausdorff distance.
\end{theorem}
\section{Techniques}
\label{sec:techniques}
\newcommand{\tp}[1]{^{\otimes #1}}

Here is the basic idea behind using sum-of-squares for tensor decomposition:
Let $a_1,\ldots,a_n\in \R^d$ be unit vectors and suppose we have access to their first three moments $\cM_1,\cM_2,\cM_3$ as in \pref{eq:moments}.
Since the task of recovering $a_1,\ldots,a_n$ is easier the more moments we know, we would make a lot of progress if we could compute higher moments of $a_1,\ldots,a_n$, say the fourth moment $\cM_4$.
A natural approach toward that goal is to compute a probability distribution $D$ over the sphere of $\R^d$ such that $D$ matches the moments of $a_1,\ldots,a_k$ that we know, i.e., $\E_{D(u)} u=\cM_1$, $\E_{D(u)} u^{\otimes 2}=\cM_2$, $\E_{D(u)} u^{\otimes 3}=\cM_3$, and then use the fourth moment $\E_{D} u^{\otimes 4}$ as an estimate for $\cM_4$.

There are two issues with this approach:
(1) computing such a distribution $D$ is intractable and
(2) even if we could compute such a distribution it is not clear if its fourth moment will be close to the fourth moments $\cM_4$ we are interested in.

We address issue (1) by relaxing $D$ to be a pseudo-distribution (solution to sum-of-squares relaxations). Then, we can match the given moments efficiently.

Issue (2) is related to the uniqueness of the tensor decomposition, which relies on properties of the vectors $a_1,\ldots,a_n$.
Here, the general strategy is to first prove that this uniqueness holds for actual distributions and then transfer the uniqueness proof to the sum-of-squares proof system, which would imply that uniqueness also holds for pseudo-distributions.

In subsection~\ref{sec:techniques/rounding} below, we demonstrate our key rounding idea on the (nearly) orthogonal tensor decomposition problem. Then in subsection~\ref{sec:techniques/fourth} we discuss the high level insight for the robust 4th-order tensor decomposition algorithm and in subsection~\ref{sec:techniques/third} the techniques for random 3rd-order tensor decomposition.

\subsection{Rounding pseudo-distributions by matrix diagonalization}
\label{sec:techniques/rounding}

Our main departure from previous tensor decomposition algorithms based on sum-of-squares \cite{DBLP:conf/stoc/BarakKS15,DBLP:conf/approx/GeM15} lies in \emph{rounding}: the procedure to extract an actual solution from a pseudo-distribution over solutions.
The previous algorithms rounded a pseudo-distribution $D$ by directly using the first moments (or the mean) $\E_{D(u)} u$, which requires $D$ to concentrate strongly around the desired solution.
Our approach here instead uses Jennrich's (simultaneous) matrix diagonalization \cite{harshman1970foundations,zbMATH00446923}, to extract the desired solution as a \emph{singular vector} of a matrix of the form $\E_{D(u)} \inner{g,u}\dyad{u}$, for a random vector $g$.\footnote{
In previous treatments of simultaneous diagonalization, multiple matrices would be used for noise tolerance---increasing the confidence in the solution when more than one matrix agrees on a particular singular vector.
This is unnecessary in our setting, since as we'll see, the SoS framework itself suffices to certify the correctness of a solution.
}
This permits us to impose much weaker conditions on $D$.

For the rest of this subsection, we assume that we have an actual distribution $D$ that is supported on vectors close to some orthonormal basis $a_1,\ldots,a_d$ of $\R^d$, and we will design a rounding algorithm that extracts the vectors $a_i$ from the low-degree moments of $D$. This is a much simpler task than rounding from a pseudo-distribution, though it captures most of the essential difficulties. Since pseudo-distributions behave similarly to actual distributions on the low-degree moments, the techniques involved in rounding from actual distributions will turn out to be easily generalizable to the case of pseudo-distributions. 

Let $D$ be a distribution over the unit sphere in $\R^d$.
Suppose that this distribution is supported on vectors close to some orthonormal basis $a_1,\ldots,a_d$ of $\R^d$, in the sense that the distribution satisfies the constraint
\begin{equation}
  \Set{ \sum_{i=1}^d \langle  a_i, u \rangle^3 \ge 1-\e }_{D(u)}\mper
 \label{eq:main-constraint}
\end{equation}
(This constraint implies $\{\max_{i\in [d]} \langle  a_i,u \rangle\ge 1-\e\}_{D(u)}$ because $\sum_{i=1}^d \langle  a_i,u \rangle^3\le \max_{i\in [d]} \langle  a_i,u \rangle$ by orthonormality.)
The analysis of \cite{DBLP:conf/stoc/BarakKS15} shows that reweighing the distribution $D$ by a function of the form $u\mapsto \langle  g,u \rangle^{2k}$ for $g\sim \cN(0,\Id_d)$ and some $k\le O(\log d)$ creates, with significant probability, a distribution $D'$ such that for one of the basis vectors $a_i$, almost all of the probability mass of $D'$ is on vectors close to $a_i$, in the sense that
\begin{displaymath}
  \max_{i\in [d]}\E_{D'(u)} \langle  a_i,u \rangle \ge 1-O(\e)\mcom \text{where } D'(u) \propto \langle  g,u \rangle^{2k} D(u) \mper
\end{displaymath}
In this case, we can extract a vector close to one of the vectors $a_i$ by computing the mean $\E_{D'(u)} u$ of the reweighted distribution.
This rounding procedure takes quasi-polynomial time because it requires access to logarithmic-degree moments of the original pseudo-distribution~$D$.

To avoid this quasi-polynomial running time, our strategy is to instead modify the original distribution $D$ in order to create a small bias in one of the directions $a_i$ such that a modified moment matrix of $D$ has a one-dimensional eigenspace close to $a_i$.
(This kind of modification is much less drastic than the kind of modification in previous works.
Indeed, reweighing a distribution such that it concentrates around a particular vector seems to require logarithmic degree.)

Concretely, we will study the spectrum of matrices of the following form, for $g\sim \cN(0,\Id_d)$:
\begin{displaymath}
  M_g = \E_{D(u)} \langle  g,u \rangle \cdot\dyad u.
\end{displaymath}
Our goal is to show that with good probability, $M_g$ has a one-dimensional eigenspace close to one of the vectors $a_i$.

However, this is not actually true for a na\"{i}ve distribution: although we have encoded the basis vectors $a_i$
into the distribution $D$ by means of constraint \pref{eq:main-constraint},
we cannot yet conclude that the eigenspaces of $M_g$ have anything to do with them.
We can understand this as the error allowed in \pref{eq:main-constraint} being highly under-constrained.
For example, the distribution could be a uniform mixture of vectors of the form
$a_i + \e w$ for some fixed vector $w$, which causes $w$ to become by far the
most significant contribution to the spectrum of $M_g$.
More generally, an arbitrary spectrally small error could still completely displace all of the eigenspaces of $M_g$.

An interpretation of this situation is that we have permitted $D$
itself to contain a large amount of information that we do not actually possess.
Constraint \pref{eq:main-constraint} is consistent with a wide range of possible solutions,
yet in the pathological example above, the distribution does not at all
reflect this uncertainty, instead settling arbitrarily on some particular biased solution: it is this bias that disrupts the usefulness of the rounding procedure.

A similar situation has previously arisen in strategies for rounding convex
relaxations---specifically, when the variables of the relaxations were interpreted as the marginals of some probability distribution over solutions, then actual solutions were constructed by sampling from that distribution.
In that context, a workaround was to sample those solutions from the maximum-entropy distributions consistent with those marginals~\cite{gharan2014new},
to ensure that the distribution faithfully reflected the ignorance inherent in the relaxation solution rather than incorporating arbitrary information.
Our situation differs in that it is the solution to the convex relaxation itself
which is misbehaving, rather than some aspect of the rounding process,
but the same approach carries over here as well.

Therefore, suppose that $D$ satisfies the maximum-entropy constraint $\lVert \E_{D(u)} \dyad u\rVert\le 1/n$. This essentially enforces $D$ to be a uniform distribution over vectors close to $a_1,\dots, a_n$.  For the sake of demonstration, we assume that $D$ is a uniform distribution over $a_1,\dots, a_n$. Moreover, 
since our algorithm is invariant under linear transformations, we may assume that the components $a_1,\ldots,a_n$ are the standard basis vectors $e_1,\ldots,e_n\in \R^d$.
We first decompose $M_g$ along the coordinate $g_1$,
\begin{displaymath}
  M_g = g_1 \cdot M_{e_1} + M_{g'}, \quad\text{where } g'=g - g_1 \cdot e_1\mper
\end{displaymath}
Note that under our simplified assumption for $D$, by simple algebraic manipulation we have $M_{e_1} =  \E_{D(u)} u_1\dyad{u} = \dyad{e_1}$. Moreover, by definition, $g_1$ and $g'$ are independent.
It turns out that the entropy constraint implies $\E_{g'} \lVert  M_{g'} \rVert\lesssim \sqrt {\log d} \cdot 1/n$ (using concentration bounds for Gaussian matrix series \cite{zbMATH05946839}).
Therefore, if we condition on the event $g_1 >\eta^{-1} \sqrt {\log d}$, we have that $M_g = g_1\dyad{e_1} + M_{g'}$ consists of two parts: a rank-1 single part $g_1\dyad{e_1}$ with with eigenvalue larger than $\eta^{-1}\sqrt{\log d}$, and a noise part which has spectral norm at most $\lesssim \sqrt{\log d}$. Hence, by the eigenvector perturbation theorem we have that the top eigenvector is $O(\eta^{1/2})$-close to $e_1$ as desired.  

Taking $\eta =0.1$, we see with $1/\poly(d)$ probability the event $g_1 > \eta^{-1}\sqrt{\log d}$ will happen, and therefore by repeating this procedure $\poly(d)$ times, we obtain a vector that is $O(\eta^{1/2})$-close to $e_1$. We can find other vectors similarly by repeating the process (in a slightly more delicate way), and the accuracy can also be boosted (see Sections~\ref{sec:rounding} and~\ref{sec:decomp-sum-of-squares} for details).

\subsection{Overcomplete fourth-order tensor}
\label{sec:techniques/fourth}

In this section, we give a high-level description of a robust sum-of-squares version of the tensor decomposition algorithm FOOBI \cite{DBLP:journals/tsp/LathauwerCC07}. For simplicity of the demonstration, we first work with the noiseless case where we are given a tensor $T\in (\R^d)^{\otimes 4}$ of the form
\begin{equation}
T = \sum_{i=1}^n a_i^{\otimes 4}
\mper
\label{eq:techniques/decomposition}
\end{equation}

We will first review the key step of FOOBI algorithm and then show how to convert it into a sum-of-squares algorithm that will naturally be robust to noise.

To begin with, we observe that by viewing $T$ as a $d^2\times d^2$ matrix of rank $n$, we can easily find the span of the $a_i\tp{2}$'s by low-rank matrix factorization.
However, since the low rank matrix factorization is only unique up to unitary transformation, 
we are not able to recover the $a_i\tp{2}$'s from the subspace that they live in.
The key observation of~\cite{DBLP:journals/tsp/LathauwerCC07} is that the $a_i\tp{2}$'s are actually the only ``rank-1'' vectors in the span, under a mild algebraic independence condition.
Here, a $d^2$-dimensional vector is called ``rank-1'' if it is a tensor product of two vectors of dimension $d$.

\begin{lemma}[\cite{DBLP:journals/tsp/LathauwerCC07}]
\label{lem:uniqueness}
	Suppose the following set of vectors is linearly independent,
	\begin{equation}
	\label{eq:techniques/independent}
	\left\{ a_i^{\otimes 2}\otimes a_j^{\otimes 2} - (a_i\otimes a_j)^{\otimes 2} \Mid i\neq j\right\}\mper
	\end{equation}
	Then every vector $\cramped{x^{\otimes 2}}$ in the linear span of $a_1^{\otimes 2}, \ldots, a_n^{\otimes 2}$ is a multiple of one of the vectors $a_i^{\otimes 2}$.
\end{lemma}

This observation leads
to the algorithm FOOBI, which essentially looks for rank-1 vectors in the span of $a_i\tp{2}$'s.
The main drawback is that it uses simultaneous diagonalization as a sub-procedure, which is unlikely to tolerate noise better than inverse polynomial in $d$, and in fact no noise tolerance guarantee has been explicitly shown for it before.

Our approach starts with rephrasing the original proof of Lemma~\ref{lem:uniqueness} into the following SoS proof (which only uses polynomial inequalities that can be proved by SoS).

\begin{proof}[Proof of Lemma~\ref{lem:uniqueness}]
Let $\alpha_1,\ldots,\alpha_n$ be multipliers such that $x^{\otimes 2} = \sum_{i=1}^n \alpha_i \cdot a_i^{\otimes 2}$.\footnote{technically, $\alpha_1,\dots, \alpha_n$ are polynomials in $x$ so that $x^{\otimes 2} = \sum_{i=1}^n \alpha_i \cdot a_i^{\otimes 2}$ holds} 
Then, these multipliers satisfy the following quadratic equations:
\begin{gather*}
  x^{\otimes 4} = \sum\nolimits_{i,j} \alpha_i\alpha_j \cdot a_i^{\otimes 2}\otimes a_j^{\otimes 2} \mcom\\
  x^{\otimes 4} = \sum\nolimits_{i,j} \alpha_i\alpha_j \cdot (a_i\otimes a_j)^{\otimes 2} \mper
\end{gather*}
Together, the two equations imply that
\begin{displaymath}
  0 = \sum\nolimits_{i\neq j} \alpha_i\alpha_j \cdot \Paren{a_i^{\otimes 2}\otimes a_j^{\otimes 2} -  (a_i\otimes a_j)^{\otimes 2}} \mper
\end{displaymath}
By assumption, the vectors $a_i^{\otimes 2}\otimes a_j^{\otimes 2} -  (a_i\otimes a_j)^{\otimes 2}$ are linearly independent for $i\neq j$.
Therefore, from the equation above, we conclude $\sum_{i\neq j} \alpha_i^2\alpha_j^2 = 0$, meaning that at most one of $\alpha_i$ can be non-zero.
Furthermore this argument is a SoS proof, since for any matrix $A \in \R^{D \times D}$ with linearly independent columns and any vector polynomial $v \in \R[x]^D$, the inequality $\|v\|^2 \le \frac{1}{\sigma_{\min}(A)^2}\|Av\|^2$ can be proved by SoS (here $\sigma_{\min}(A)$ denotes the least singular value of matrix $A$).
So choosing $A$ to be the matrix with columns $a_i^{\otimes 2}\otimes a_j^{\otimes 2} -  (a_i\otimes a_j)^{\otimes 2}$ for $i \ne j$ and $v$ to be the vector with entries $\alpha_i\alpha_j$, we find by SoS proof that $\|\alpha\|_4^4 - \|\alpha\|_2^4 = 0$.
\end{proof}

When there is noise present, we cannot find the true subspace of the $a_i\tp{2}$'s and instead we only have an approximation, denoted by $V$, of that subspace. We will modify the proof above by starting with a polynomial inequality \begin{equation}
\|\Id_Vx\tp{2}\|^2 \ge (1-\delta)\|x\tp{2}\|^2\mcom\label{eqn:foobi_tech_constr}
\end{equation}
which constrains $x\tp{2}$ to be close to the estimated subspace $V$ (where $\delta$ is a small number that depends on error and condition number).
Then an extension of the proof of Lemma~\ref{lem:uniqueness} will show that equation~\eqref{eqn:foobi_tech_constr} implies (via a SoS proof) that for some small enough $\delta$, 
\begin{equation}
\sum_{i\neq j} \alpha_i^2\alpha_j^2\le o(1)\mper
\end{equation}

Note that $\alpha = Kx\tp{2}$ is a linear transformation of $x\tp{2}$, and furthermore $K$ is the pseudo-inverse of the matrix with columns $a_i\tp{2}$.
Moreover, if we assume for a moment that $\alpha$ has 2-norm 1 (which is not true in general), then the equation above further implies that
\begin{equation}
\sum_{i=1}^n \inner{K_i, x\tp{2}}^4  =\|\alpha\|_4^4 \ge 1-o(1)\mcom \label{eqn:eqn62}
\end{equation}
where $K_i\in \R^{d^2}$ is the $i$-th row of $K$.
This effectively gives us access to the 4-tensor $\sum_i K_i\tp{4}$ (which has ambient dimension $d^2$ when flattened into a matrix), since equation~\eqref{eqn:eqn62} is anyway the constraint that would have been used by the SoS algorithm if given the tensor $\sum_i K_i\tp{4}$ as input.
Note that because the $K_i$ are not necessarily (close to) orthogonal, we cannot apply the SoS orthogonal tensor decomposition algorithm directly.
However, since we are working with a 4-tensor whose matrix flattening has higher dimension $d^2$, we can whiten $K_i$ effectively in the SoS framework and then use the orthogonal SoS tensor decomposition algorithm to find the $K_i$'s, which will in turn yield the $a_i$'s.

Many details were omitted in the heuristic argument above (for example, we assumed $\alpha$ to have norm 1). The full argument follows in Section~\ref{sec:foobi}.

\subsection{Random overcomplete third-order tensor}
\label{sec:techniques/third}
In the random overcomplete setting, the input tensor is of the form
\[ T = \sum_{i=1}^n  a_i^{\ot 3} + E \mcom \]
where each $a_i$ is drawn uniformly at random from the Euclidean unit sphere,
we have $d < n \le d^{1.5}/(\log d)^{O(1)}$, and $E$ is some noise tensor such that
$\|E\|_{\{1\},\{2,3\}} < \e$ or alternatively such that a constant-degree
sum-of-squares relaxation of the injective norm of $E$ is at most $\e$.

Our original rounding approach depends on the target vectors $a_i$ being
orthonormal or nearly so.
But when $n \gg d$ in this overcomplete setting, orthonormality fails badly:
the vectors $a_i$ are not even linearly independent.

We circumvent this problem by embedding the vectors $a_i$ in a larger
ambient space---specifically by taking the tensor powers
$a_1' = a_1^{\otimes 2},\dots, a_n' =a_n^{\otimes 2}$.
Now the vectors $a_1',\dots,a_n'$ are linearly independent (with probability 1)
and actually close to orthonormal with high probability.
Therefore, if we had access to the order-6 tensor $\sum (a_i')^{\otimes 3} = \sum a_i^{\otimes 6}$, then we could (almost) apply our rounding method to
recover the vectors $a_i'$.

The key here will be to use the sum-of-squares method to generate a pseudo-distribution over the unit sphere having $T$ as its third-order moments tensor,
and then to extract from it the set of order-6 pseudo-moments estimating the moment tensor $\sum_i a_i^{\ot 6}$.
This pseudo-distribution would obey the constraint
$\{\transpose{\left(u \ot u \ot u\right)}T \ge 1 - \e\}$,
which implies the constraint $\{\sum_i \iprod{a_i, u}^3 \ge 1 - \e\}$,
saying, informally, that our pseudo-distribution is close to the actual uniform distribution over $\{a_i\}$.
Substituting $v = u^{\ot 2}$, we obtain an implied pseudo-distribution in $v$
which therefore ought to be close to the uniform distribution over $\{a_i'\}$,
and we should therefore be able to round the order-3 pseudo-moments of $v$ to recover $\{a_i'\}$.

Only two preconditions need to be checked:
first that $\sum_i \dyad{(a_i')}$ is not too large in spectral norm,
and second that our pseudo-distribution in $v$ satisfies the constraint
$\{\sum_i \iprod{a_i', v}^3 \ge 1 - O(\e)\}$.
The first precondition is true (except for a spurious eigenspace which can harmlessly be projected away) and is essentially equivalent to a line of matrix concentration arguments previously made in \cite{HopkinsSSS16}.
The second precondition follows from a line of constant-degree sum-of-squares proofs,
notably extending arguments made in \cite{DBLP:conf/approx/GeM15}
stating that the constraints $\{\sum_i \iprod{a_i, u}^3 \ge 1 - \e, \|u\|^2 = 1\}$
imply with constant-degree sum-of-squares proofs
that $\{\sum_i \iprod{a_i, u}^k \ge 1 - O(\e) - \tO(n/d^{3/2})\}$ for some higher powers $k$.
The rigorous verification of these conditions is detailed in Section \ref{sec:random_tensor}.
\section{Preliminaries}
\label{sec:preliminaries}

Unless explicitly stated otherwise, $O(\cdot)$-notation hides absolute multiplicative constants.
Concretely, every occurrence of $O(x)$ is a placeholder for some function $f(x)$ that satisfies $\forall x\in \R.\, \abs{f(x)}\le C\abs{x}$ for some absolute constant $C>0$.
Similarly, $\Omega(x)$ is a placeholder for a function $g(x)$ that satisfies $\forall x\in \R.\, \abs{g(x)} \ge \abs{x}/C$ for some absolute constant $C>0$.

For a matrix $A$, let $A^+$ denote the Moore-Penrose pseudo-inverse of $A$.
For a symmetric positive semidefinite matrix $B$, let $B^{1/2}$ denote the square root of $B$, i.e. the unique symmetric positive-semidefinite matrix $L$ such that $L^2 = B$.

The Kronecker product of two matrices $A$ and $B$ is denoted by $A\otimes B$.
A useful identity is that $(A\otimes B) (C\otimes D) = (AC)\otimes (BD)$ whenever the matrix multiplications are defined.
The norm $\|\cdot\|$ denotes the Euclidean norm for vectors and the spectral norm for matrices.

Let $T\in (\R^d)^{\otimes k}$ be a \emph{$k$-tensor} over $\R^d$ such that $T=\sum_{i_1,\ldots,i_k}T_{i_1\cdots i_k} e_{i_1}\otimes \cdots \otimes e_{i_k}$, where $e_1,\ldots,e_d$ is the standard basis of $\R^d$.
We say $T$ is \emph{symmetric} if the entries $(T_{i_1,\ldots,i_k})$ are invariant under permuting the indices.
The $k$ index positions of $T$ are called \emph{modes}.
The injective norm $\norm{T}_{\mathrm{inj}}$ is the maximum value of $\langle  T , x_1 \otimes \cdots \otimes x_k \rangle$ over all vectors $x_1,\ldots,x_k\in \R^d$ with $\norm{x_1}=\cdots=\norm{x_k}=1$.
A useful class of multilinear operations on tensors has the form $T\mapsto (A_1\otimes \cdots \otimes A_k) T$, where $A_1,\ldots,A_k$ are matrices with $d$ columns.
(This notation is the same as the Kronecker product notation for matrices, that is, $(A_1\otimes \cdots \otimes A_k) T=\sum_{i_1,\ldots,i_k} T_{i_1\cdots i_k} (A_1 e_{i_1})\otimes \cdots \otimes(A_k e_{i_k})$.)
If some of the matrices $A_i$ are row vectors, and the others are the identity matrix, then the corresponding operation is called tensor contraction.
For example, for a third-order tensor $T\in(\R^d)^{\otimes 3}$ and a vector $g\in\R^d$, we call $(\Id\otimes \Id \otimes \transpose{g}) T$ the contraction of the third mode of $T$ with $g$.
(Some authors use the notation $T(\Id,\Id,g)$ to denote this operation.)

For a bipartition $A,B$ of the index set $[k]$ of $T$, we let $\lVert  T \rVert_{A,B}$ denote the spectral norm of the matrix unfolding $T_{A,B}$ of $T$ with rows indexed by the indices in $A$ and columns indexed by indices in $B$.
Concretely,
\begin{displaymath}
  \lVert  T \rVert_{A,B} =
  \max_{\substack{
      x \in (\R^d)^{\otimes \abs{A}},\, y\in (\R^d)^{\otimes \abs{B}}\\
      \lVert x \rVert\le 1 ,\, \lVert  y \rVert \le 1}}
  \sum_{i_1,\ldots,i_k} T_{i_1\cdots i_k} \cdot  x_{i_A}  y_{i_B}\mcom
\end{displaymath}
Here, $i_A=i_{a_1}\cdots i_{a_{\abs A}}$ and $i_B=i_{b_1}\cdots i_{b_{\abs B}}$ are multi-indices, where $A=\{a_1,\ldots,a_{\abs A}\}$ and $B=\{b_1,\ldots,b_{\abs B}\}$.
For $k=2$, $\lVert  T \rVert_{\{1\},\{2\}}$ is the spectral norm of $T$ viewed as a $d$-by-$d$ matrix.
For $k=3$, $\lVert  T \rVert_{\{1,2\},\{3\}}$ is the spectral norm of $T$ viewed as a $d^2$-by-$d$ matrix with rows indexed by the first two modes of $T$ and columns indexed by the last index of $T$.
For symmetric $3$-tensors, all norms $\norm{T}_{\{1,2\},\{3\}}$, $\norm{T}_{\{1,3\},\{2\}}$, and $\norm{T}_{\{2,3\},\{1\}}$ are the same.

\subsection{Pseudo-distributions}
\label{sec:pseudo-distribution}

Pseudo-distributions generalize probability distributions in a way that allows us to optimize efficiently over moments of pseudo-distributions.
We represent a discrete probability distribution $D$ over $\R^n$ by its probability mass function $D\from \R^n \to \R$ such that $D(x)$ is the probability of $x$ under the distribution for every $x\in \R^n$.
This function is nonnegative point-wise and satisfies $\sum_{x\in \supp(D)} D(x)=1$.
For pseudo-distributions we relax the nonnegative requirement and only require that the function passes a set of simple nonnegativity tests.

A \emph{degree-$d$ pseudo-distribution over $\R^n$} is a finitely\footnote{We restrict these functions to be finitely supported in order to avoid integrals and measurability issues. It turns out to be without loss of generality in our context.} supported function $D\from \R^n \to \R$ such that $\sum_{x\in \supp(D)} D(x)=1$ and $\sum_{x\in \supp(D)} D(x)f(x)^2\ge 0$ for every function $f\from \R^n \to \R$ of degree at most $d/2$.
We define the \emph{pseudo-expectation} of a (possibly vector-valued or matrix-valued) function $f$ with respect to $D$ as
\begin{displaymath}
  \tE_{D} f \defeq \sum_{x\in \supp(D)} D(x) f(x) \mper
\end{displaymath}
In order to emphasize which variable is bound by the pseudo-expectation, we write $\tE_{D(x)} f(x)$.
(This notation is useful if $f(x)$ is a more complicated expression involving several variables.)

Note that a degree-$\infty$ pseudo-distribution $D$ satisfies $D(x)\ge 0$ for all $x\in \R^n$.
Therefore, $D$ is an actual probability distribution (with finite support).
The pseudo-expectation $\tE_{D}f=\E_D f$ of a function $f$ is its expected value under the distribution $D$.

Our algorithms will not work with pseudo-distributions (as finitely-supported functions on $\R^n$) directly.
Instead the algorithms will work with moment tensors $\pE_{D(x)} (1,x_1,\ldots,x_n)^{\otimes d}$ of pseudo-distributions and the associated linear functional $p\mapsto \pE_{D(x)}p(x)$ on polynomials $p$ of degree at most $d$.

Unlike actual probability distribution, pseudo-distributions admit general, efficient optimization algorithms.
In particular, the set of low-degree moments of pseudo-distributions has an efficient separation oracle.
\begin{theorem}[\cite{
    MR931698,
    parrilo2000structured,
    MR1814045}]
  \label{thm:sos}
  For $n,d\in \N$, the following set admits an $n^{O(d)}$-time weak separation oracle (in the sense of \cite{DBLP:journals/combinatorica/GrotschelLS81}),
  \begin{displaymath}
    \Set{ \pE_{D(x)} (1,x_1,\ldots,x_n)^{\otimes d}
      \Mid \text{degree-$d$ pseudo-distribution $D$ over $\R^n$ }}\mper
  \end{displaymath}
\end{theorem}
This theorem, together with the equivalence of separation and optimization \cite{DBLP:journals/combinatorica/GrotschelLS81} allows us to solve a wide range of optimization and feasibility problems over pseudo-distributions efficiently.

The following definition captures what kind of linear constraints are induced on a pseudo-distribution over $\R^n$ by a system of polynomial constraints over $\R^n$.

\begin{definition}
  Let $D$ be a degree-$d$ pseudo-distribution over $\R^n$.
  For a system of polynomial constraints $\cA = \{f_1 \ge 0, \ldots, f_m \ge 0\}$ with $\deg(f_i)\le \ell$ for every $i$, 
  we say that \emph{$D$ satisfies the polynomial constraints $\cA$ at degree $\ell$}, denoted $D \models_\ell \cA$,
  if $\pE_D \left(\prod_{i \in S}f_i\right)h \ge 0$ for every $S \subseteq [m]$ and every sum-of-squares polynomial $h$ on $\R^n$ with $|S|\ell + \deg h \le d$.
\end{definition}

This is a relaxation (to pseudo-distributions) of the statement that the probability mass of a true distribution contains only solutions to $\cA$. Indeed, if an actual distribution $D$ is supported on the solutions to $\cA$, then 
$D$ satisfies $D \models_{\ell} \cA$ regardless of the value of $\ell$. 

We say that $D$ satisfies $\cA$ (without further specifying the degree) if $D \models_{\ell} \cA$ for $\ell=\max_{\{f \ge 0\} \subseteq \cA }\deg f$.
We say that a system $\cA$ of polynomial constraints in variables $x$ is \emph{explicitly bounded} if it contains a constraint of the form $\{\norm{x}^2\le M\}$.
The following theorem follows from \pref{thm:sos} and \cite{DBLP:journals/combinatorica/GrotschelLS81}. We give a proof in Appendix~\ref{sec:prelim_appendix} for completeness.

\begin{theorem}\label{thm:inequality_sos}
  There exists a $(n+\card{\cA})^{O(d)}$-time algorithm that, given any explicitly bounded and satisfiable system $\cA$ of polynomial constraints in $n$ variables, outputs (up to arbitrary accuracy) a degree-$d$ pseudo-distribution that satisfies $\cA$.
\end{theorem}

\Dnote{}
\Jnote{}

\subsection{Sum of squares proofs}
\label{sec:preliminaries/proofs}

Let $x=(x_1,\ldots,x_n)$ be a tuple of indeterminates.
Let $\R[x]$ be the set of polynomials in these indeterminates with real coefficients.
A polynomial $p$ is a \emph{sum-of-squares} if there are polynomials $q_1,\ldots,q_r$ such that $p=q_1^2+\dots+q_r^2$.
Let $f_1,\ldots,f_r$ and $g$ be multivariate polynomials in $\R[x]$.
A \emph{sum-of-squares proof} that the constraints $\{f_1\ge 0,\ldots,f_r\ge 0\}$ imply the constraint $\{g\ge 0\}$ consists of sum-of-squares polynomials $(p_S)_{S\subseteq [n]}$ in $ \R[x]$ such that
\begin{displaymath}
  g = \sum_{S\subseteq [n]} p_S \cdot \prod_{i\in S} f_i \mper
\end{displaymath}
We say that this proof has \emph{degree $\ell$} if every set $S\subseteq [n]$ satisfies $\deg\paren{p_S\cdot \prod_{i\in S} f_i}\le \ell$
(in particular, this would imply $p_S = 0$ for every set $S$ such that $\deg \prod_{i\in S} f_i>\ell$).
If there exists a degree-$\ell$ sum-of-squares proof that $\{f_1\ge 0,\ldots,f_r\ge 0\}$ implies $\{ g\ge 0\}$, we write
\begin{displaymath}
  \{ f_1\ge 0,\ldots,f_r\ge 0\} \vdash_\ell \{ g \ge 0\}\mper
\end{displaymath}
In order to emphasize the indeterminates for the proofs, we sometimes write $  \{ f_1(x)\ge 0,\ldots,f_r(x)\ge 0\} \vdash_{x,\ell} \{ g(x) \ge 0\}\mper$

Sum-of-squares proofs obey the following inference rules, for all polynomials $f,g\from \R^n \to \R$ and $F\from \R^n \to \R^m, G\from \R^n \to \R^k, H\from \R^p\to \R^n$,
\begin{align}
  &\frac{\cA \vdash _\ell \{ f\ge 0 , g \ge 0\}}
    {\cA \vdash_{\ell} \{ f+g\ge 0\}},
    \quad
    \frac{\cA \vdash _{\ell} \{ f\ge 0\},~ \cA \vdash_{\ell'} \{g \ge 0\}}
    {\cA \vdash_{\ell+\ell'} \{ f\cdot g \ge 0\}}\mcom
    \tag{addition and multiplication}
  \\[0.5em]
  &\frac{\cA \vdash _\ell \cB,~ \cB\vdash _{\ell'} \cC}
    {\cA \vdash _{\ell\cdot \ell'} \cC}\mcom
    \tag{transitivity}
  \\[0.5em]
  &\frac{\{F\ge 0\} \vdash_{\ell} \{ G \ge 0\}}
    {\{F(H))\ge 0\} \vdash_{\ell\cdot \deg(H)} \{ G(H)\ge 0\}}
    \mper
    \tag{substitution}
\end{align}

Sum-of-squares proofs are sound and complete for polynomial constraints over pseudo-distributions, in the sense that sum-of-squares proofs allow us to reason about what kind of polynomial constraints are satisfied by a pseudo-distribution.
We defer the proofs of the following lemmas to \pref{app:prelim_appendix}.

\begin{lemma}[Soundness]\label{lem:soundness_sos}
  If $D \models_{\ell} \cA$ for a pseudo-distribution $D$ and there exists a sum-of-squares proof $\cA \vdash_{\ell'} \cB$, then $D \models_{\ell\cdot\ell'} \cB$.
\end{lemma}

\begin{lemma}[Completeness]\label{lem:completeness_sos}
	Suppose $d\ge \ell' \ge \ell$, and $\cA$ is a collection of polynomial constraints with degree at most $\ell$, and $\cA\vdash \{x_1^2+\dots + x_n^2 \le B\}$ for some finite $B$. Let $\{g\ge 0\}$ be a polynomial constraint with degree $\ell'$. 
	If every degree-$d$ pseudo-distribution $D$ that satisfies $D \models_{\ell} \cA$ also satisfies $D \models_{\ell'} \{g\ge 0\}$, then for every $\epsilon > 0$, there is a sum-of-squares proof $\cA \vdash_{d} \{g\ge -\epsilon\}$.
	\footnote{The completeness claim stated here does not match the strength of the corresponding soundness claim.
		This reflects an impreciseness in how we count the degrees of intermediate sum-of-squares proofs (in particular our degree accounting is not tight under proof composition), and does not reflect than the power of the proofs themselves.}
\end{lemma}

\subsection{Matrix constraints and sum-of-squares proofs}

In sections~\ref{sec:rounding} and~\ref{sec:rounding_general_components}, we still state positive-semidefiniteness constraints on matrices, which will be implied by sum-of-squares proofs.
We define notation to express what it means for a matrix constraint to be implied by sum-of-squares.
While the duality between proof systems and convex relaxations also holds in the matrix case~\cite{CIMPRIC201289}, and it is possible to give a full treatment of matrix constraints in sum-of-squares, here we give an abridged and simplified treatment sufficient for our purposes.

\begin{definition}Let $\cA$ be a set of polynomial constraints in indeterminant $x$, and $M$ is a symmetric $p\times p$ matrix with entries in $\R[x]$. Then we write 
$\cA \vdash_{\ell} \{M \succeq 0\}$
if there exists a set of polynomials $q_1(x),\dots q_m(x)$ and a set of vectors $v_1(x),\dots, v_m(x)$ of $p$-dimension with entries in $\R[x]$ such that $\cA\vdash_{\ell_i} \{q_i\ge 0\}$ where $\ell_i + 2\deg(v_i)\le \ell$ for every $i$, and 
\begin{equation}
M = \sum_{i=1}^{m} q_i(x)\dyad{v_i(x)}\mper \label{eqn:def:matrix_sos}
\end{equation}
\end{definition}

The proof that sum-of-squares is sound for these matrix constraints is very similar to the analogous proof of Lemma~\ref{lem:soundness_sos} (see Appendix~\ref{sec:prelim_appendix}).

\begin{lemma}\label{lem:matrix_sos_soundness}
	Let $D$ be a pseudo-distribution of degree $d$ and $d\ge \ell\ell'$. Suppose $D\models_{\ell} \cA$, and $\cA \vdash_{\ell'} M\succeq 0$. 
	Then $\pE\left[M\right]\succeq 0$. 
\end{lemma}

We now give some basic properties of these matrix sum-of-squares proofs.

\begin{lemma}\label{lem:tensor_power_psd}
	Suppose $A,B'$ are symmetric matrix polynomials such that $\,\vdash \{A \succeq 0$, $B \succeq 0\}$.
	Then $\,\vdash \{A\otimes B \succeq 0\}$.
\end{lemma}
\begin{proof}
	Express $A = \sum_{i=1}^{n} q_i(x)\dyad{u_i(x)}$ and $B = \sum_{i=1}^{m} r_i(x)\dyad{v_i(x)}$.
	Then \[A \ot B = \sum_{i=1}^{n}\sum_{j=1}^{m}q_i(x)r_j(x)\dyad{\big[u_i(x) \ot v_j(x)\big]} \mper \qedhere \]
\end{proof}
\begin{lemma}\label{lem:tensor_product_psd}
	Suppose $A,B,A',B'$ are symmetric matrix polynomials such that $\,\vdash \{A \succeq 0$, $B' \succeq 0$, $ A\succeq A'$, $B \succeq B'\}$.
	Then $\,\vdash \{A\otimes B \succeq A'\otimes B'$, $B\otimes A \succeq B'\otimes A' \}$.
\end{lemma}
\begin{proof}
	By \pref{lem:tensor_power_psd}, we have $\,\vdash A\otimes (B-B') \succeq 0$ and $\,\vdash (A - A')\otimes B' \succeq 0$.
	Adding the two equations we complete the proof.
	We may also take the tensor powers in the other order.
\end{proof}

\begin{lemma}\label{lem:sos_psd_1}\label{lem:matrix-sos}
	Let $u=[u_1,\dots, u_d]$ be an indeterminate. Then $\,\vdash \{\dyad{u} \preceq \|u\|^2 \cdot \Id_d \}$.
\end{lemma}
\begin{proof}
	The conclusion follows from the following explicit decomposition
	\[
		\vdash \|u\|^2 \Id- uu^{\trans} = \sum_{1\le i < j \le d} \dyad{(u_ie_j - u_je_i)}\succeq 0\qedhere
	\]
\end{proof}

\section{Rounding pseudo-distributions}\label{sec:rounding}

\subsection{Rounding by matrix diagonalization}

The following theorem analyzes a form of Jennrich's algorithm for tensor decomposition through matrix diagonalization, when applied to the moments of a pseudo-distribution.
We show that if the pseudo-distribution $D(u)$ has good correlation with some vector $a^{\otimes k}$, then with good chance a simple random contraction of the $(k+2)$-th moments of the pseudo-distribution will return a matrix with top eigenvector close to  $a$. 

Theorem~\ref{thm:reproducing_bks} below is the key ingredient toward a polynomial-time algorithm.
It states that in order for Jennrich's approach to successfully extract a solution in polynomial time, the correlation of the desired solution with the $(k+2)$-th moments of the pseudo-distribution only needs to be large compared to the spectral norm of the covariance matrix of the pseudo-distribution.
This covariance matrix can be made as small as $O(1/n)$ in spectral norm in many situations, including---as a toy example---when $D$ is a uniform distribution over $n$ orthogonal unit vectors.
Therefore in this sense the condition~\eqref{eqn:eqn59} below is a fairly weak requirement, which is key to the polynomial-time algorithm in Section~\ref{sec:general-tensor-decomposition}.

\begin{theorem}\label{thm:reproducing_bks}
  \label{thm:rounding-bks}
  Let $k\in \N$ be even and $\e \in (0,1)$.
  Let $D$ be a degree-$O(k)$ pseudo-distribution over $\R^d$ that satisfies $\{\norm{u}^2\le 1\}_{D(u)}$, let $a\in \R^d$ be a unit vector.
  Suppose that
  \begin{gather}
    \pE_{D(u)} \inner{a,u}^{k+2} \ge \Omega\Paren{\frac{1}{\e \sqrt k }} \cdot \Norm{\tilde\E _{D(u)} \dyad u} \mper\label{eqn:eqn59}
  \end{gather}
  Then, with probability at least $1/d^{O(k)}$ over the choice of $g\sim \cN(0,\Id_{d}^{\otimes k})$, the top eigenvector~$u^\star$ of the following matrix $M_g$ satisfies $\langle  a,u^\star \rangle^2\ge 1-O(\e)$,
  \begin{equation}
    M_g \seteq \pE_{D(u)}\iprod{g,u^{\otimes k}}\cdot \dyad{u}\mper
  \end{equation}

\end{theorem}

As before, we decompose $M_g$ into two parts, with $M_{a^{\otimes k}}$ and $M_{g'}$ defined in analogy with $M_g$.
\begin{equation}
\label{eq:20}
  M_g = \langle  g, a^{\otimes k} \rangle \cdot M_{a^{\otimes k}} + M_{g'}
  \quad \text{where } g' = g - \langle  g,a^{\otimes k} \rangle \cdot a^{\otimes k}\,.
\end{equation}

Our proof of \pref{thm:rounding-bks} consists of two propositions: one about the good part $M_{a^{\otimes k}}$ and one about the noise part $M_{g'}$.
The first proposition shows that $M_{a^{\otimes k}}$ is close to a multiple of $\dyad {a}$ in spectral norm
(which means that the top eigenvector of $M_{a^{\otimes k}}$ is close to $a$).
\begin{proposition}
\label{prop:rounding-bks-1} %
  In the setting of \pref{thm:rounding-bks}, for $t=\pE_{D(u)} \langle a,u \rangle^{k+2}$,
  \begin{equation}
    \Norm{M_{a^{\otimes k}} - t \cdot \dyad{a}} \le O\Paren{\e} \cdot t\mper
  \end{equation}
\end{proposition}

The second proposition shows that $M_{g'}$ has small spectral norm in expectation.
\begin{proposition}
  \label{lem:bks_noise_part} %
  \label{prop:bks_noise_part} %
  \label{prop:rounding-bks-2} %
  In the setting of \pref{thm:reproducing_bks}, let $g'=g - \langle  g,a^{\otimes k} \rangle \cdot a^{\otimes k}$.
  Then, for $t=\pE_{D(u)} \langle a,u \rangle^{k+2}$,
  \begin{displaymath}
    \E_{g'} \Norm{ M_{g'}} \le  O({\e^2 k \log d})^{1/2} \cdot t\,.
  \end{displaymath}
\end{proposition}

Before proving the above propositions, we demonstrate how they allow us to prove \pref{thm:rounding-bks}.
\begin{proof}[Proof of \pref{thm:rounding-bks}]
  We are to show that with probability $1/d^{O(k)}$ over the choice of the Gaussian vector~$g$, there exists $s\in\R$ such that $\Norm{s\cdot M_g -  \dyad{a}}\le O(\e)$. By Davis-Kahan Theorem (see Theorem~\ref{thn:davis-kahan-rank-1}), this implies the conclusion of Theorem~\ref{thm:rounding-bks}. 
  \TMnote{Fixed the references. Though it seems that we can get better error using Davis-Kahan -- it was $O(\epsilon)$ in $\ell_2$-square distance but in fact Davis-Kahan gives $O(\epsilon)$ in $\ell_2$ error}%
  Let $t=\pE_{D(u)} \langle a,u \rangle^{k+2}$.
  For a parameter $\tau = \Omega(k \log d)^{1/2}$, we bound the spectral norm conditioned on the event $\iprod{g,a^{\otimes k}}\ge \tau$,
  \begin{align}
    & \E_{g} \Brac{\; \Norm{\tfrac{1} {\langle  g,a^{\otimes k} \rangle\cdot t }M_g - \dyad{a}}\;\;\middle\vert\;\; \iprod{g,a^{\otimes k}}\ge \tau \;}\notag\\
    & \le \Norm{\tfrac 1t M_{a^{\otimes k}} - \dyad{a}} + \E_{g} \Brac{\,\tfrac 1 {\langle  g,a^{\otimes k} \rangle\cdot t}\Norm{M_{g'}} \;\;\middle\vert\;\; \iprod{g,a^{\otimes k}}\ge \tau\;}
      \tag{by \pref{eq:20}}\\
    & \le \Norm{\tfrac 1t M_{a^{\otimes k}} - \dyad{a}} + \tfrac 1 {\tau\cdot t}\cdot \E_{g'} \Norm{M_{g'}}
      \tag{by independence of $\langle  g,a^{\ot k} \rangle$ and $g'$}
      \\
    & \le O(\e)  + \tfrac 1 \tau \cdot O(\e^2 k\log d)^{1/2}
      \tag{by \pref{prop:rounding-bks-1} and \ref{prop:rounding-bks-2}}\\
      & \le O(\e) \mper
  \end{align}
  By Markov's inequality, it follows that conditioned on $\langle  g,a^{\ot k} \rangle\ge \tau$, the event $\Norm{\tfrac{1} {\langle  g,a^{\ot k} \rangle\cdot t }M_g - \dyad{a}}\le O(\e)$ has probability at least $\Omega(1)$.
  The theorem follows because the event $\langle  g,a^{\ot k} \rangle\ge \tau$ has probability at least $d^{-O(k)}$.
\end{proof}

\begin{proof}[Proof of \pref{prop:rounding-bks-1}]
  We are to bound the spectral norm of $M_{a^{\otimes k}} - t \cdot \dyad{a}$ for $t=\pE_{D(u)} \langle a,u \rangle^{k+2}$.
  Let $\alpha=\norm{\pE_{D(u)}\dyad u}$.
  Let $\Id_{1}=\dyad {a}$ be the projector onto the subspace spanned by $a$ and let $\Id_{-1}=\Id - \Id_{1}$ be the projector on the orthogonal complement.
  By our choice of $t$, we have $\Id_1 M_{a^{\otimes k}} \Id_1 = t \cdot \dyad{a}$.

  Since $\Id_{-1}\Id_1=0$, we can upper bound the spectral norm of ${M_{a^{\otimes k}}}-t\cdot \dyad {a^{\otimes k}}$,
  \TMnote{low priority thing: This part of the proof can be replaced by Lemma~\ref{lem:improved_wedin}}
  \begin{align}
    \Norm{M_{a^{\otimes k}}-t\cdot \Id_1}
    &\le \Norm{\Id_{1} (M_{a^{\otimes k}}-t\cdot \Id_1)\Id_{1}}
    +  \Norm{\Id_{-1} M_{a^{\otimes k}}\Id_{-1}}
    + 2 \Norm{\Id_{1} M_{a^{\otimes k}}\Id_{-1}}\notag
    \\
    & \le \Norm{\Id_{-1} M_{a^{\otimes k}}\Id_{-1}}
    + 2 \Norm{\Id_{1} M_{a^{\otimes k}}\Id_{-1}}
      \tag{because $\Id_{1} M_{a^{\otimes k}}\Id_{1}=t\cdot \Id_1$}\\
    & \le \Norm{\Id_{-1} M_{a^{\otimes k}}\Id_{-1}}
    + 2 \Norm{\Id_{1} M_{a^{\otimes k}}\Id_1}^{1/2}
      \cdot \Norm{\Id_{-1} M_{a^{\otimes k}}\Id_{-1}}^{1/2}
      \tag{because $M_{a^{\otimes k}}\succeq 0$}\\
    & \le \Norm{\Id_{-1} M_{a^{\otimes k}}\Id_{-1}}
    + 2 \sqrt{\alpha}
      \cdot \Norm{\Id_{-1} M_{a^{\otimes k}}\Id_{-1}}^{1/2}\,.
      \label{eq:19}
  \end{align}
  It remains to bound the spectral norm of $\Id_{-1} M_{a^{\otimes k}} \Id_{-1}$,
  \begin{align}
    \Norm{\Id_{-1} M_{a^{\otimes k}} \Id_{-1}}
    & = \Norm{\pE_{D(u)} \langle a,u\rangle^k \cdot \Id_{-1}\dyad{u}\Id_{-1}} \notag\\
    & \le \pE_{D(u)} \langle a,u\rangle^k \cdot (1-\langle  a,u \rangle^2)
      \tag{because $\vdash \Id_{-1}\dyad u \Id_{-1}\preceq (\lVert  u \rVert^2-\langle  a_1,u \rangle ^2)\Id$}\\
    & \le \tfrac 2 {k-2} \pE_{D(u)} \langle a,u\rangle^2
      \tag{using $\vdash_{k+2} x^{k-2} \cdot (1-x^2) \le \tfrac 2 {k-2}$; see below}\\
    & \le \frac {2}{k-2} \cdot\alpha \label{eq:8}
  \end{align}
  Basic calculus shows that the inequality $x^{k-2} \cdot (1-x^2) \le \frac 2 {k-2}$ holds for all $x\in \R$.
  Since it is a true univariate polynomial inequality in $x$, it has a sum-of-squares proof with degree no larger than the degree of the involved polynomials, which is $k+2$ in our case. %
  Combining \pref{eq:19} and \pref{eq:8}, yields as desired that
  \Dnote{}
  \begin{displaymath}
    \Norm{M_{a^{\otimes k}} - t\cdot \dyad a}
    \le O\Paren{\tfrac{1}{\sqrt k}} \cdot \alpha
    \le O\Paren{\e} \cdot t
    \mcom
  \end{displaymath}
  where the second step uses the condition of \pref{thm:rounding-bks} on $t=\pE_{D(u)}\langle a,u \rangle$ and $\alpha=\norm{\pE_{D(u)}\dyad{u}}$.
\end{proof}

\begin{proof}[Proof of \pref{prop:rounding-bks-2}]
  The matrix $M_{g'}=\pE_{D(u)} \langle  g', u^{\otimes k}\rangle \cdot \dyad u$, whose spectral norm we are to bound, is a random contraction of the third-order tensor $T = \pE_{D(u)} u \otimes u \otimes (u^{\otimes k})$.
  \pref{cor:matrix_concentration_general}
  gives the following bound on the expected norm of a random contraction in terms of spectral norms of two matrix unfoldings of $T$---which turn out to be the same in our case due to the symmetry of $T$.
  \begin{equation}
    \label{eq:21}
    \E_{g'} \lVert  M_{g'} \rVert \le O({\log d})^{1/2} \cdot \max\{\lVert  T \rVert_{\{1\},\{23\}}, \lVert  T \rVert_{\{2\},\{13\}}\}
    = O({\log d})^{1/2}\cdot \biggl\lVert  \pE_{D(u)} u^{\otimes k} \transpose u \biggr\rVert \,.
  \end{equation}
  \pref{thm:higher-tensor-norm} shows that for any pseudo-distribution $D$ that satisfies $\{\norm{u}^2\le 1\}$,
  \begin{equation}
    \label{eq:22}
    \biggl\lVert  \pE_{D(u)} u^{\otimes k} \transpose u \biggr\rVert
    \le \biggl\lVert  \pE_{D(u)} \dyad u \biggr\rVert\mper
  \end{equation}
  The statement of the lemma follows by combining the previous bounds \pref{eq:21} and \pref{eq:22},
  \begin{displaymath}
    \E_{g'} \lVert  M_{g'} \rVert \le O({\log d})^{1/2}\cdot \biggl\lVert  \pE_{D(u)} u^{\otimes k} \transpose u \biggr\rVert
    \le O({\log d})^{1/2}\cdot \biggl\lVert  \pE_{D(u)} \dyad u \biggr\rVert \le O({\e^2 k \log d})^{1/2}\cdot t \mcom
  \end{displaymath}
  using condition \pref{eqn:eqn59} of \pref{thm:rounding-bks} which yields $t=\pE_{D(u)}\langle a,u \rangle^{k+2} \ge \Omega\left((\e \sqrt{k})^{-1}\right) \norm{\pE_{D}\dyad{u}}$.
\end{proof}

\subsection{Improving accuracy of a found solution}
\label{sec:boost-accuracy}
We need one more technical ingredient before analyzing our main algorithm.
Previously, the run-time of the sum-of-squares algorithm in~\cite{DBLP:conf/stoc/BarakKS15} (on which our algorithm is based) depended exponentially on the accuracy parameter $1/\epsilon$, and we give here a simple boosting technique that allows us to remove this dependency and achieve polynomially small error.

Here we have a set of nearly isotropic vectors $a_1,\dots ,a_n$.
We give a sum-of-squares proof that if $\sum_{i=1}^n \inner{a_i,u}^4$ is only $\epsilon$ off from its maximum possible value, and if $u$ has constant correlation with some $a_i$, then $u$ must in fact be $(1-O(\epsilon))$-correlated with $a_i$.
Intuitively, the former constraint forces $D$ to roughly be a mixture distribution over vectors that are close to $a_1,\dots, a_n$, and the latter one forces it to actually only be close to $a_i$.
We then briefly show how this proof implies an algorithm to boost the accuracy when we already know a vector $b$ that is $0.01$-close to a solution, by solving for a pseudo-distribution with the added constraint $\{\iprod{b,u}^2 \ge 0.9\}$.

\begin{theorem}
	\label{thm:boosting-accuracy}
	Let $\e>0$ be smaller than some constant.
	Let $a_1,\ldots,a_n\in \R^d$ be unit vectors such that $\norm{\sum_{i=1}^n \dyad {a_i} }\le 1+\e$.
	Define the following systems of constraints, for each $j \in [n]$ or unit vector $b \in \R^d$:
	\begin{align*}
		\cA_j &:= \left\{\norm{u}^2\le 1, \sum\nolimits_{i=1}^n \iprod{a_i,u}^4 \ge 1-\e,\; \iprod{a_j,u}^2\ge \tfrac{1}{2}\right\} _{D(u)} \\
		\cB_b &:= \left\{\norm{u}^2\le 1, \sum\nolimits_{i=1}^n \iprod{a_i,u}^4 \ge 1-\e,\; \iprod{b,u}^2\ge 0.9 \right\} _{D(u)}\mper
	\end{align*}
	Then $\cA_j \vdash_4 \{\iprod{a_j,u}^2 \ge 1 - 10\e\}$ for all $j \in [n]$, and also $\cB_b \vdash \cA_j$ and $\iprod{a_i,b}^2 \ge 0.8$ for some $j \in [n]$.
\end{theorem}

\begin{proof}
	We have the following sum-of-squares proof:
	\begin{align}
		\cA \vdash_{u,4} \;\; 1-\e
		& \le \sum\nolimits_{i=1}^n \iprod{a_i,u}^4 \notag\\
		& \le \iprod{a_j,u}^4 + \left( \sum\nolimits_{i\ne j} \iprod{a_i,u}^2\right)^2
		\tag{by adding only square terms}\\
		& \le \iprod{a_j,u}^4 + \left(1+\e- \iprod{a_j,u}^2\right)^2
		\tag{using $\vdash_u \sum_{i=1}^n\iprod{a_i, u}^2 \le (1+\e)\norm{u}^2$}\\
		& \le \iprod{a_j,u}^2 + \left(\tfrac12 +\e\right) \left(1+\e - \iprod{a_j,u}^2\right)\tag{since $\vdash 1/2 \le \iprod{a_j,u}^2 \le 1$}\\
		& \le \left(\tfrac12 -\e \right )\iprod{a_j,u}^2 + \tfrac12 +2\e \mcom
	\end{align}
	which means that $\cA \vdash_{u,4} \iprod{a_j,u}^2 \ge  1-10\e$ for $\e>0$ small enough.
	
	To show that $\cB_b \vdash \cA_i$ for some $i$, it is enough to show that if $\cB_b$ is consistent (i.e. there exists a pseudo-distribution satisfying $\cB_b$), then there exists $i\in[n]$ such that $\iprod{a_i,b}^2 \ge 0.8$, because it implies $\{\iprod{a_i,u}^2 \ge 1/2\}$ by triangle inequality.
	
	For the sake of contradiction, assume that $\iprod{a_i,b}^2 < 0.8$ for all $i\in [n]$.
	Then, by triangle inequality (see \pref{lem:helper_lemma1}), $\cB_b \vdash \{\forall i\in [n].\,\iprod{a_i,u}^2 \le 0.99\}$ which when combined with $\left\|\dyad{a_i}\right\|^2 \le 1 + \e$ using substitution, contradicts the assumption that $\cB_b \vdash \{\sum_{i=1}^n \langle  a_i,u \rangle^4 \ge 1-\e\}$ for small enough $\e>0$.
\end{proof}

\begin{corollary}
	\label{cor:boosting-accuracy}
	Let $D$ be a degree-$\ell$ pseudo-distribution over $\R^d$ such that $D \models_{\ell/4} \cB_b$, with $\cB_b$ as defined in \pref{thm:boosting-accuracy}.
	Then, there exists $i\in[n]$ such that $\Norm{\pE_{D(u)} u^{\otimes 2} - a_i^{\otimes 2}}^2\le O(\e)$ and $\iprod{a_i, b}^2 \ge 0.8$.
\end{corollary}
\begin{proof}
	By \pref{thm:boosting-accuracy}, $\cB_b \vdash_4 \{\iprod{a_i, u}^2 \ge 1 - 10\e\}$ for some $i$.
	It follows by \pref{lem:soundness_sos} that
	\begin{displaymath}
	\Norm{\pE_{D(u)} u^{\otimes 2} - a_i^{\otimes 2}}^2 \le 2 - 2 \left\langle \pE_{D(u)}u^{\otimes 2}, a_i^{\otimes 2} \right\rangle
	= 2- 2 \pE_{D(u)}\langle  a_i,u \rangle^2 \le 20\e \mper
	\qedhere
	\end{displaymath}
\end{proof}

\section{Decomposition with sum-of-squares}
\label{sec:decomp-sum-of-squares}
In this section, we give a generic sum-of-squares algorithm (Algorithm~\ref{alg:tensor-general} and Theorem~\ref{thm:general-tensor-decomposition}) that will be used for various different settings in the following subsections (Section~\ref{subsec:orthogonal} for orthogonal tensors, Section~\ref{subsec:separated} for tensors with separated components), and in the section~\ref{sec:random_tensor} for random 3-tensor and Section~\ref{sec:foobi} for robust FOOBI.

\subsection{General algorithm for tensor decomposition}
\label{sec:general-tensor-decomposition}

In this section, we provide a general sum-of-squares tensor decomposition that serve as the main building block for sections later. 
We will need the following lemma, which appears in \cite[Proof of Lemma 6.1]{DBLP:conf/stoc/BarakKS15}.

\begin{lemma}
\label{lem:higher-power}
  Let $\e\in(0,1)$ and $\{a_1,\ldots,a_n\}$ be a set of unit vectors in $\R^d$ with $\norm{\sum_{i=1}^n \dyad{a_i}} \le 1+\e$.  %
  Then, for all even integers $k\in \N$, there exists a sum-of-squares proof that
  \begin{gather}
    \Set{ \norm{u}^2\le 1,  \sum_{i=1}^n \langle a_i, u \rangle^{4} \ge 1-\e  }
    \vdash_{u,\, k+2}
    \Set{\sum_{i=1}^n \langle a_i, u \rangle^{k+2} \ge 1-O(k\e)}\mper
    \label{eqn:large-k-power}
  \end{gather}
\end{lemma}

Our main algorithm below finds the solutions to a system of polynomial constraints $\cA$, when given a ``hint'' in the form of a polynomial transformation of formal variables $P(\cdot)$.
Roughly $P$ should be an ``orthogonalizing'' map so that if $a_1, \dots, a_n$ are the desired solutions to the constraints $\cA$, then $P(a_1), \dots, P(a_n)$ are nearly an orthonormal basis, or more precisely $\|\sum_{i=0}^n \dyad{P(a_i)}\| \le 1 + \e$ while $\|P(a_i)\|^2 \ge 1 - \e$ for all $i$.
We then only require that a sum-of-squares proof exists certifying that the solutions to $\cA$ after being mapped by $P$ are actually close to $P(a_1), \dots P(a_n)$; more precisely, that $\cA \vdash_{\ell} \{\sum_{i=1}^n \inner{P(a_i),P(u)}^4\ge 1-\epsilon\}_{u}$ for some $\ell$.
The existence of this sum-of-squares certificate then allows us to recover the solutions $P(a_i)$ up to $O(\e)$ accuracy by solving for pseudo-distributions and then rounding them.

We later show how \pref{alg:tensor-general} can be applied to a variety of tensor rank decomposition problems by the design of an appropriate orthogonalizing transform $P$.
For example, in \pref{sec:random_tensor} $P(\cdot)$ orthogonalizes an overcomplete tensor by lifting the variables to a higher-dimensional space, and $P(\cdot)$ serves as a whitening transformation on a far-from-orthogonal tensor in \pref{sec:foobi}.

The main technical difficulty in this analysis was in making the run-time polynomial (as opposed to quasi-polynomial in~\cite{DBLP:conf/stoc/BarakKS15}) for the nearly-orthogonal case where $P$ is the identity transform. 
\begin{theorem}\label{thm:general-tensor-decomposition}
  For every $\ell\in \N$,
  there exists an $n^{O(\ell)}$-time algorithm (see \pref{alg:tensor-general}) with the following property:
  Let $\e>0$ be smaller than some constant.
  Let $d,d'\in\N$ be numbers.
  Let $P\from \R^d\to \R^{d'}$ be a polynomial with $\deg P\le \ell$.
  Let $\{a_1,\ldots,a_n\} \subseteq \R^{d}$ be a set of vectors such that $b_1=P(a_1),\ldots,b_n=P(a_n)\in\R^{d'}$ all have norm at least $1-\e$ and $\norm{\sum_{i=1}^n \dyad{b_i}}\le 1+\e$.
  Let $\cA$ be a system of polynomial inequalities in variables $u=(u_1,\ldots,u_d)$ such that  the vectors $a_1,\ldots,a_n$ satisfy $\cA$ and
  \begin{equation}
    \label{eq:14}
    \cA \vdash_{u,\ell}
    \Set{\vbig
      \sum_{i=1}^n \langle  b_i, P(u) \rangle^4 \ge (1-\e)\Norm{P(u)}^4
    }\mper
  \end{equation}
  Then, the algorithm on input $\cA$ and $P$ outputs a set of unit vectors $\{b'_1,\ldots,b'_n\}\subseteq \R^{d'}$ such that
  \begin{displaymath}
    \dist_H\Paren{\vbig \left\{b_1^{\otimes 2},\ldots,b_n^{\otimes 2}\right\}, \left\{(b_1')^{\otimes 2},\ldots,(b'_n)^{\otimes 2}\right\}} \le O(\e)^{1/2}\mper
  \end{displaymath}
\end{theorem}

\begin{algorithm}
  \caption{General tensor decomposition algorithm}
  \label{alg:tensor-general}

  \textbf{Parameters:}
  numbers $\e>0,\, n,\ell\in \N$.

  \textbf{Given:}
  system $\cA$ of polynomial inequalities over $\R^d$ and polynomial $P\from \R^d\to\R^{d'}$.

  \textbf{Find:}
  vectors $b_1',\ldots,b'_n\in \R^{d'}$.

\textbf{Algorithm:}
  \begin{itemize}
  \item For $i$ from $1$ to $n$, do the following:
    \begin{enumerate}
    \item Compute a degree-$(k+2)\ell$ pseudo-distribution $D(u)$ over $\R^d$, with $k = O(1)$, that
      satisfies the constraints
      \begin{align}
        & \cA \cup \{ 1+\e \ge \Norm{P(u)}^2 \ge 1-\e\} \nonumber
        \\
        & \Norm{\pE\nolimits_{D(u)} \dyad{P(u)}} \le \frac{1+\e}{n-i+1}\mper
          \label{eq:16}
      \end{align}
    \item Choose standard Gaussian vectors $\super g 1,\ldots,\super g T\sim \cN(0,\Id_{(d')^k})$ and $T = d^{O(1)}$ and compute the top eigenvectors of the following matrices for all $t\in [T]$:
      \begin{equation}
        \label{eq:18}
        \pE_{D(u)} \langle \super g t,P(u)^{\otimes k} \rangle\cdot \dyad{P(u)} \in \R^{d'\times d'}\mper
      \end{equation}
    \item Check if for one of the normalized top eigenvectors $b^\star$ computed in the previous step, there exists a degree-$4\ell$ pseudo-distribution $D'(u)$ that satisfies the constraints
      \begin{equation}
        \label{eq:15}
        \cA \cup \bigl\{1+\e \ge  \Norm{P(u)}^2 \ge 1-\e, \langle  b^\star,P(u) \rangle^2 \ge 0.99 \bigr\}\mper
      \end{equation}
    \item Set $b'_i$ to be the top eigenvector of the matrix $\pE_{D'(u)} \dyad{P(u)}$ and add to $\cA$ the constraint $\{\langle  P(u), b'_i \rangle^2 \le 0.01\}$.
    \end{enumerate}

  \end{itemize}
\end{algorithm}

\begin{proof}[Proof of \pref{thm:general-tensor-decomposition}]
  We analyze \pref{alg:tensor-general}.
  By \pref{cor:boosting-accuracy}, if there exists a pseudo-distribution $D'(u)$ that satisfies constraints \pref{eq:15}, then the top eigenvector of  $\pE_{D'(u)} \dyad{P(u)}$ is $O(\e)^{1/2}$-close to one of the vectors $b_1,\ldots ,b_n$.
  The fact that we add in step 4, the constraint $\{\iprod{P(u),b_i'}\le 0.1\}$ also implies by \pref{cor:boosting-accuracy} that in some iteration $i$, we can never find a vector $b'_i$ that is close to one vector $b'_j$ from a previous iteration $j<i$.
  Therefore, it remains to show that in each of the $n$ iterations with high probability we can find a pseudo-distribution $D'(u)$ that satisfies \pref{eq:15}.

  Consider a particular iteration $i_0\in [n]$ of \pref{alg:tensor-general}.
  We may assume that the vectors $b_1',\ldots,b_{i_0-1}'$ are close to $b_1,\ldots,b_{i_0-1}$.
  First we claim that there exists a pseudo-distribution satisfying conditions \pref{eq:16} in step 1, including the additional constraints added to $\cA$ in previous iterations.
  Indeed, the uniform distribution over vectors $a_{i},\ldots,a_{n}$ satisfies all of those conditions.
  By assumption, we have a sum-of-squares proof $\cA \vdash_{u,\ell} \{\sum_{i=1}^n \langle  b_i, P(u) \rangle^4 \ge 1-\e \}$.
  \pref{lem:higher-power} then implies $\cA \vdash_{u,(k+2)\ell} \{\sum_{i=1}^n \langle  b_i,P(u) \rangle^k\ge 1-O(k \e)\}$ for an absolute constant parameter $k$ to be determined later.
  Since $\cA$ includes the added constraints $\{\langle  b_1,P(u) \rangle^2 \le 0.1,\ldots,\langle b_{i_0-1},P(u) \rangle^2 \le 0.1\}$, it follows by $\Norm{\sum_{i=1}^n \dyad{b_i}}{\vphantom{\sum_i^n}}^2 \le 1 + O(\e)$ and substitution that $\cA \vdash \{\sum_{i=1}^{i_0-1} \langle  b_i,P(u) \rangle^k\le (0.1)^{k-2}\cdot (1+O(\e))\}$, here choosing $k$ so that $(0.1)^{k-2}\cdot (1+O(\e))\le 0.001$.
  Therefore, $\cA \vdash_{(k+2)\ell} \{\sum_{i=i_0}^{n} \iprod{b_i,P(u)}^k\ge 0.99\}$ and so $\pE_{D(u)} \sum_{i=i_0}^{n} \iprod{b_i,P(u)}^k\ge 0.99$ for any degree-$(k+2)\ell$ pseudo-distribution $D$ that satisfies constraints \pref{eq:16}.
  In particular, by averaging, there exists an index $i^\star\in \{i_0,\ldots,n\}$ such that
  \begin{displaymath}
    \pE_{D(u)}\langle  b_{i^\star}, P(u) \rangle^k \ge \frac {0.99}{n-i_0+1} \ge 0.9 \cdot \Norm{\pE_{D(u)} \dyad{P(u)}}\mper
  \end{displaymath}
  By \pref{thm:rounding-bks}, for each of the matrices \pref{eq:18} in step 2, its top eigenvector is $0.001$-close to $b_{i^\star}$ with probability at least $d^{-O(1)}$.
  Therefore, we find at least one of those vectors with probability no smaller than $1-d^{\Omega(1)}$.
  In this case, a pseudo-distribution $D'(u)$ as required in step 3 exists, as an atomic distribution supported only on $b_{i^\star}$ is an example that satisfies the conditions.
\end{proof}

\subsection{Tensors with orthogonal components}
\label{subsec:orthogonal}
We apply Theorem~\ref{thm:general-tensor-decomposition} to orthogonal tensors with noise. 

\LetLtxMacro\Oldfootnote\footnote
\renewcommand{\footnote}[2][]{\relax}
\restatetheorem{thm:orthogonal-vectors}
\LetLtxMacro\footnote\Oldfootnote%

\begin{proof}
	We feed Algorithm~\ref{alg:tensor-general} with the inputs $P(u) = u$ and $\cA = \Set{\inner{T,u\tp{3}}\ge 1-\epsilon}$ where $\e = \|E\|_{\{1\},\{2,3\}}$ and $E=T-\sum_i a_i^{\otimes 3}$.
    We have \begin{align} \cA \vdash_4 \sum_{i=1}^n \inner{a_i,u}^3 & = \inner{T,u^{\otimes 3}} - \inner{E,u^{\otimes 3}} \nonumber\\
        & =\inner{T,u^{\otimes 3}} - \epsilon   \nonumber \\
          & \ge  1-2\epsilon \mper \nonumber
         \end{align}
         Here at the second line we used that 
         \begin{align}
	         \vdash \inner{E,u^{\otimes 3}} \le \Norm{E}_{\{1\},\{2,3\}}\le \epsilon \mper \label{eqn:47}
         \end{align}
        We verify that $\cA$ satisfies the requirement~\eqref{eq:14},
	\begin{align}
	\cA \vdash \sum_{i=1}^n \inner{a_i,u}^4 & \ge \left(\sum_{i=1}^n \inner{a_i,u}^4\right)\left(\sum_{i=1}^n \inner{a_i,u}^2\right) \tag{using orthonormality}\\
	& \ge \left(\sum_{i=1}^n \inner{a_i,u}^3\right)^2\tag{Cauchy-Schwarz: \pref{lem:cauchy-schwarz}}\\
	& \ge 1-4\epsilon\mper\nonumber
	\end{align}
	Therefore calling Algorithm~\ref{alg:tensor-general}, we can recover $\hat{a}_i$ which is, up to sign flip,  close to $a_i$ with error $O(\epsilon^{1/2})$. We determine the sign by finding the $\tau\in \{-1,+1\}$ such that  $\inner{T,\tau \hat{a}_i\tp{3}}\ge 1-\epsilon$ and set the output $a_i'$ to $\tau \hat{a}_i$.
\end{proof}

\begin{remark}\label{remark:sosnorm}
	Note that in the proof of~\pref{thm:orthogonal-vectors}, the conclusion of equation~\eqref{eqn:47} is the only thing we used about the error term $E$. Therefore, define the following SoS relaxation of the injective norm:
	\begin{align}
		\Norm{E}_{\textup{SoS}} = \inf_{c\in \R}\Brac{\vphantom{\sum}\left\{\norm{u}^2\le 1\right\}\vdash \left\{\inner{E,u^{\otimes 3} }\le c\right\}}\mper\nonumber
	\end{align}
	Then we can replace the right hand side of equation~\eqref{eqn:48} by $O(1)\cdot \Norm{T - \sum\nolimits_{i=1}^n a_i^{\otimes 3}}_{\textup{SoS}}$. 
\end{remark}

\subsection{Tensors with separated components}

\label{subsec:separated}
The following lemma shows that for separated vectors the sum of higher-order outer products has spectral norm that decrease exponentially with the tensor order.

\begin{lemma}
  \label{lem:separated-spectral-norm}
  Let $a_1,\ldots,a_n$ be unit vectors in $\R^d$.
  Then, for every $k\in \N$,
  \begin{displaymath}
    \Norm{\sum_{i=1}^n \Paren{\dyad{a_i}}^{\otimes (k+1)}} \le 1 + \Paren{\max_{i\neq j} |\iprod{a_i,a_j}|}^{k} \cdot \Norm{\sum_{i=1}^n \dyad{a_i}}.
  \end{displaymath}
\end{lemma}

\begin{proof}
  Let $A = \sum_{i} \Paren{\dyad{a_i}}^{\otimes k+1}$.
  For a unit vector $x \in (\R^d)^{\ot k+1}$ we'll bound the quadratic form $\transpose{x}Ax$.
  
  First, without loss of generality we can assume that $x$ is in the subspace $V$ spanned by $\{a_i^{\ot k+1}\}_i$.
  This is because if $x$ had a component $y$ orthogonal to $V$, then $\transpose{y}a_i^{\ot k+1} = 0$ for all $i \in [n]$ by definition, so that $A y = 0$ and $y$ can make no nonzero contribution to the quadratic form above.

  Also let $W = (A^{1/2})^{+}$ so that $W$ is a whitening transform and $WAW$ is a projector onto $V$.
  Then suppose $x = \sum_{i} c_i Wa_i^{\ot k+1}$, so that $\sum_i c_i^2 = \|x\|^2 = 1$.
  Then
  \begin{align*}
  \transpose{x}Ax &= \sum_{ij} c_ic_j\transpose{(a_i^{\ot k+1})} WAWa_j^{\ot k+1} \\
  &= \sum_{i}c_i^2 + \sum_{i\ne j} c_ic_j\iprod{a_i,a_j}^{k+1} \\
  &\le \sum_{i}c_i^2 + \Paren{\max_{i\neq j} |\iprod{a_i,a_j}|}^k \sum_{i\ne j} c_ic_j\iprod{a_i,a_j} \\
  &\le 1 + \Paren{\max_{i\neq j} |\iprod{a_i,a_j}|}^k \Norm{\sum_{i=1}^n \dyad{a_i}}
  \mcom
  \end{align*}
  where in the last step we let  $A' = \sum_{i}\dyad{a_i}$ and $W' = (A'^{1/2})^{+}$, and apply the inequality $\sum_{i,j} c_ic_j\iprod{a_i,a_j} = \sum_{ij} c_ic_j\transpose{a_i} W'A'W'a_j = \transpose{x'}A'x' \le \|A'\|$,
  where $x' = \sum_{i} c_i W'a_i$ is a unit vector.
\end{proof}

\begin{lemma}
\label{lem:recovery}
  Let $a\in \R^d$ and $b\in (\R^d)^{\otimes k}$ be unit vectors such that $\langle  a^{\otimes k},b \rangle^2\ge 1-\e$.
  Let $B$ be the reshaping of the vector $b$ into a $d$-by-$d^{k-1}$ matrix.
  Then the top left singular vector $a'\in \R^d$ of $B$ satisfies $\langle  a',a \rangle^2 \ge 1-O(\e)$.
\end{lemma}

\begin{proof}
  Let $c$ be the top right singular vector of $B$.
  Then, $\langle  a'\otimes c, b \rangle\ge \langle  a^{\otimes k}, b \rangle\ge 1-\e$.
  Therefore, $\norm{a'\otimes c - b}\le O(\e)^{1/2}$.
  By triangle inequality, $\norm{a'\otimes c - a\otimes a^{\otimes k-1}}\le O(\e)^{1/2}$, which means that as desired $\abs{\langle  a,a' \rangle}\ge \langle  a'\otimes c,a\otimes a^{\otimes (k-1)} \rangle\ge 1-O(\e)$.
\end{proof}

\restatetheorem{thm:separated-unit-vectors}

\begin{proof}
  We use \pref{alg:tensor-general} from \pref{thm:general-tensor-decomposition}.
  Let $E=T-\sum_i a_i^{\otimes k}$.
  We may assume that $\norm{E}_{\{1,\ldots,\lfloor k/2\rfloor \},\{\lfloor k/2\rfloor +1,\ldots,k\}} \le \eta$, since otherwise the theorem follows from the case when $\eta =\norm{E}_{\{1,\ldots,\lfloor k/2\rfloor \},\{\lfloor k/2\rfloor +1,\ldots,k\}}$. 
  \Jnote{}
  \TMnote{modified}
  Let $P$ be the polynomial map $P(x)=x^{\otimes \lceil k/4\rceil}$ and let $\cA$ be the system of polynomial inequalities
  \begin{equation}
    \cA=\{\langle  T, u^{\otimes k} \rangle\ge 1-\eta, \norm{u}^2 =1\}\mper
  \end{equation}
  in variables $u=(u_1,\ldots,u_d)$.
  Since $\norm{E}_{\{1,\ldots,\lfloor k/2\rfloor \},\{\lfloor k/2\rfloor +1,\ldots,k\}}\le \eta$, all of the vectors $a_1,\ldots,a_n$ satisfy $\cA$.
  Let $b_1,\ldots,b_n$ be the unit vectors $b_i=P(a_i)$.
  By \pref{lem:separated-spectral-norm} and the condition on $k$, these vectors satisfy $\norm{\sum_i \dyad{b_i}}\le 1+\rho^{\lceil k/4\rceil} \sigma \le 1+\eta$.
  Then, we have the following sum-of-squares proof
  \begin{align}
    \cA \vdash_{u,k} & \sum_{i=1}^n \langle  b_i, P(u) \rangle^4 = \sum_{i=1}^n \langle  a_i,u \rangle^{4\lceil k/4\rceil} \ge \sum_{i=1}^n \langle  a_i,u \rangle^{k} = \langle  T,u^{\otimes k} \rangle - \langle  E ,u^{\otimes k}\rangle
    \\
    & \ge 1-\eta - \Norm{T}_{\{1,\ldots,\lfloor k/2\rfloor\},\{\lfloor k/2\rfloor+1,\ldots,k\}} \ge 1-2\eta\mper
  \end{align}
  It follows that $\cA$ and $P$ satisfy the conditions of \pref{thm:general-tensor-decomposition}.
  Thus, \pref{alg:tensor-general} on input $\cA$ and $P$ recovers vectors $b'_1,\ldots,b'_n$ with Hausdorff distance at most $O(\sqrt{\eta})$ from $b_1,\ldots,b_n$.
  By \pref{lem:recovery}, the top left singular vectors of the $d$-by-$d^{\lceil k/4\rceil-1}$ matrix reshapenings of $b_1',\ldots,b_n'$ are $O(\sqrt{\eta })$-close to the vectors $a_1,\ldots,a_n$ up to sign.
  (If $k$ is odd, then we may determine the signs of the $a_i$ by checking if $\inner{T,a_i'{}\tp{k}}\ge 1-O(\eta)$ or $\inner{T,a_i'{}\tp{k}}\le -1+O(\eta)$ for each output vector $a_i'$.)
\end{proof}

\section{Spectral norms and tensor operations}

In this section, we provide several bounds regarding the spectral norms of moments of the lifted vectors, and the spectral norm of random contraction of a tensor, which are crucial in our analysis in previous sections. We suggest readers who are more interested in applications of the algorithms jump to Section~\ref{sec:random_tensor} and~\ref{sec:foobi}. 

\subsection{Spectral norms and pseudo-distributions}
\begin{theorem}
\label{thm:higher-tensor-norm}
  Let $D$ be a degree-$4(p+q)$ pseudo-distribution over $\R^d$ that satisfies $\{\norm{u}^2\le 1\}_{D(u)}$.
  Then, for all $p,q\in \N$,
  \begin{equation}
    \Norm{\pE_{D(u)} u^{\otimes p} \transpose{\Paren{u^{\otimes q}}}}
    \le \Norm{\pE_{D(u)} \dyad u}\,.
  \end{equation}
\end{theorem}

The theorem follows by combining \pref{lem:spectral_norm_1} and \pref{lem:spectral_norm_2} proved below.
Lemma~\ref{lem:spectral_norm_1} reduces Theorem~\ref{thm:higher-tensor-norm} to the case when $p = q$. 
\begin{lemma}\label{lem:spectral_norm_1}
  Let $D$ be a degree-$4(p+q)$ pseudo-distribution over $\R^d$ that satisfies $\{\norm{u}^2\le 1\}_{D(u)}$.
  Then, for all $p,q\in \N$,
  \begin{equation}
    \Norm{\pE_{D(u)} u^{\otimes p} \transpose{\Paren{u^{\otimes q}}}}^2
    \le \Norm{\pE_{D(u)} \dyad{\Paren{u^{\otimes p}}}}
    \cdot \Norm{\pE_{D(u)} \dyad{\Paren{u^{\otimes q}}}}
    \,.
  \end{equation}
\end{lemma}
\begin{proof}
  For all unit vectors $x \in (\R^d)^{\otimes p}$ and $y \in (\R^d)^{\otimes q}$ 
  \begin{align}
    \langle  x, \Paren{\pE_{D(u)} u^{\otimes p} \transpose{\Paren{u^{\otimes q}}}} y\rangle
    & =  \pE_{D(u)} \langle  x, u^{\otimes p} \rangle \langle {{u^{\otimes q}}},y  \rangle\notag
    \\
    & \le \Paren{ \pE_{D(u)} \langle
  x, u^{\otimes p} \rangle^2 }^{1/2}\cdot \Paren{\pE_{D(u)}  \langle {{u^{\otimes q}}},y  \rangle^2}^{1/2}
      \tag{Cauchy--Schwarz for pseudo-expectations}
    \\
    & \le \Norm{\pE_{D(u)} \dyad{\Paren{u^{\otimes p}}}}^{1/2}
    \cdot \Norm{\pE_{D(u)} \dyad{\Paren{u^{\otimes q}}}}^{1/2}\,.
  \end{align}
  The lemma follows from this bound by choosing $x$ and $y$ as the top left and right singular vectors of the matrix $\pE_{D(u)} u^{\otimes p} \transpose{\Paren{u^{\otimes q}}}$.
\end{proof}

Towards proving Theorem~\ref{thm:higher-tensor-norm} for the case of $p=q$, we first establish the following lemma which says that tensoring with vector with norm less 1 won't increase the spectral norm. 
\begin{lemma}\label{lem:spectral1}
	Let $g(u,v)$ be a polynomial in indeterminates $u,v$. Let $D$ be a degree-$4$ pseudo-distribution over $\R^d$ that satisfies $\{\norm{u}^2\le 1, g(u,v)\ge 0\}_{D(u,v)}$.
	Then, for all $p\in \N$,
	\begin{equation}
	\Norm{\pE_{D(u,v)} g(u,v)\dyad{\Paren{u\otimes v}}}
	\le \Norm{\pE_{D(v)} g(u,v)\dyad{v}}
	\,.
	\end{equation}
\end{lemma}

\begin{proof}
	We have the sum-of-squares proof that 
	\begin{align}
	\vdash {} & g(u,v)\left(\Id \otimes vv^{\top}  - \dyad{\Paren{u\otimes v}}\right)\nonumber\\
	& = g(u,v)\left(\Id - uu^{\top}\right) \otimes \dyad v
	\nonumber\\
	& = g(u,v)\left(1- \|u\|^2\right) \Id  \otimes vv^{\top} + g(u,v)(\|u\|^2 \Id -uu^{\top})\otimes vv^{\top}	\nonumber\\
	& \succeq 0		\tag{by $1-\|u\|^2 \ge 0$ and $\vdash \|u\|^2 \Id -uu^{\top}\succeq 0$ (Lemma~\ref{lem:sos_psd_1})}
	\end{align}
	Therefore, we obtain that 
	\begin{align}
	\pE_{D(u,v)}[g(u,v)\Id \otimes vv^{\top}] - \pE_{D(u,v)}\left[g(u,v)\dyad{\Paren{u\otimes v}}\right] \succeq 0\nonumber
	\end{align}
	The desired inequality follows,
	\begin{displaymath}
	\Norm{\pE_{D(u,v)} g(u,v)\dyad{\Paren{u\otimes v}}}
	\le \Norm{\pE_{D(u,v)}\Id \otimes g(u,v) vv^{\top}} =  \Norm{\pE_{D(v)} g(u,v)\dyad{v}}
	\end{displaymath}
\end{proof}

The following statement follows straightforward from the Lemma~\ref{lem:spectral1} by induction on $p$. 

\begin{lemma}\label{lem:spectral_norm_2}
  Let $D$ be a degree-$4p$ pseudo-distribution over $\R^d$ that satisfies $\{\norm{u}^2\le 1\}_{D(u)}$.
  Then, for all $p\in \N$,
  \begin{equation}
    \Norm{\pE_{D(u)} \Big(\dyad{u}\Big)^{\otimes p}}
    \le \Norm{\pE_{D(u)} \dyad{u}}
    \mper
  \end{equation}
\end{lemma}

\subsection{Spectral norm of random contraction}
\label{sec:matrix_concentration}

The following theorem shows that a random contraction of a 3-tensor has spectral norm at most $O(\sqrt{\log d})$ factor larger than the spectral norm of its matrix unfoldings.

\begin{theorem}\label{thm:matrix_concentration}
  Let $T \in \R^{p}\otimes \R^{q}\otimes \R^{r}$ be an order-3 tensor. Let $g\in \cN(0,\Id_r)$.
  Then for any $t \ge 0$,
  \begin{equation}
    \Pr_g \Set{\vbig \Bignorm{\left(\Id\otimes \Id \otimes \transpose{g}\right)T}_{\{1\},\{2\}}
      \ge t \cdot \max\bigl\{\lVert  T\rVert_{\{1\},\{2,3\}},\|T\|_{\{2\},\{1,3\}}\bigr\} } \le 2(p+q)\cdot e^{-t^2/2}\mcom
    \label{eqn:thm:matrix_concentration}
  \end{equation}
  and consequently,\footnote{For large enough $p$ and $q$, the constant hidden in the big-Oh notation below is at most 2}
  \begin{equation}
  \E_g \left[\vbig \Bignorm{\left(\Id\otimes \Id \otimes \transpose{g}\right)T}_{\{1\},\{2\}}\right]
\le O(\log(p+q))^{1/2} \cdot \max\bigl\{\lVert  T\rVert_{\{1\},\{2,3\}},\|T\|_{\{2\},\{1,3\}}\bigr\} \mper
  \end{equation}
\end{theorem}

\begin{proof}
  Let $T_i$ denote the $i$th third-mode slice of $T$ so that $T_i = \left(\Id \otimes \Id \otimes \transpose{e_i}\right)T$ reshaped as a $p$-by-$q$ matrix.
  Note that when regarded as a $p$-by-$q$ matrix, the contraction $\left(\Id \otimes \Id \otimes \transpose g\right)T$ is a Gaussian matrix series with coefficients $T_1,\ldots,T_r$, so that
  \begin{displaymath}
    \Norm{\vbig \left(\Id \otimes \Id \otimes \transpose g\right)T}_{\{1\},\{2\}} = \Norm{\sum\nolimits_{i=1}^r g_iT_i}\mcom
  \end{displaymath}
 where $g_1,\ldots,g_r$ are independent standard Gaussians with $g_i=\iprod{g,e_i}$.
  Therefore, by concentration of Gaussian matrix series \cite[Theorem 1]{zbMATH05946839} (also see \cite[Corollary 4.2]{DBLP:journals/focm/Tropp12}), we have
  \begin{displaymath}
    \Pr\left\{ \lVert \left(\Id \otimes \Id \otimes \transpose g\right)T \rVert  \ge t \sigma\right\} \le 2 (p+q) e^{-t^2/2}\mcom
  \end{displaymath}
  where $\sigma = \max\left\{\left\lVert \sum_{i} \dyad{T_i}\right\rVert, \left\lVert\sum_{i} \transpose{T_i}T_i\right\rVert\right\}^{1/2}$.

  For $U$ and $V$ sets of indices, let $T_{U,V}$ denote the matrix unfolding of $T$
  with rows indexed by $U$ and columns indexed by $V$, so that $\|T\|_{U,V} = \|T_{U,V}\|$.
  We claim that $\sum_{i} \dyad{T_i} = \transpose{(T_{\{1\},\{2,3\}})}(T_{\{1\},\{2,3\}})$ and $\sum_{i} \transpose{T_i}T_i = \transpose{(T_{\{2\},\{1,3\}})}(T_{\{2\},\{1,3\}})$, which completes the proof.
  These identities are forced by the observations that both of these objects are matrix quantities that are quadratic in $T$, with the first object being a sum over the 2nd and 3rd indices of the two copies of $T$, and the second object being a sum over the 1st and 3rd indices.
\end{proof}

The following corollary of \pref{thm:matrix_concentration} handles a larger class of random contractions.

\begin{corollary}
  \label{cor:matrix_concentration_general}
  Let $T \in \R^{p}\otimes \R^{q}\otimes \R^{r}$ be an order-3 tensor.
  Let $g\sim \cN(0,\Sigma)$ with covariance matrix $\Sigma$ satisfying $0\preceq \Sigma\preceq \Id_r$. Then for any $t\ge 0$,
 \begin{equation}
 \Pr_g \Set{\vbig \Bignorm{\left(\Id\otimes \Id \otimes \transpose g\right)T}_{\{1\},\{2\}}
 	\ge t \cdot \max\bigl\{\lVert  T\rVert_{\{1\},\{2,3\}},\|T\|_{\{2\},\{1,3\}}\bigr\} } \le 4(p+q)\cdot e^{-t^2/2}\mper
 \label{eqn:cor:matrix_concentration}
 \end{equation}
\end{corollary}

\begin{proof}
  We reduce to the case $\Sigma = \Id_p$ and apply \pref{thm:matrix_concentration}.
  Concretely, let $g' =  g + h$ and $g'' = g-h$ where $h$ is a random variable with distribution $\cN(0,\Id_p-\Sigma)$ that is independent of $g$.
  By this construction, $g'$ and $g''$ both have marginal distribution $\cN(0,\Id_p)$, and $g = \frac{1}{2}(g'+g'')$.
  Therefore we can invoke \pref{thm:matrix_concentration} for random variables $g'$ and $g''$.
  Letting $\sigma = \max\{\lVert  T\rVert_{\{1\},\{2,3\}},\|T\|_{\{2\},\{1,3\}}\} $, using the union bound and the triangle inequality, we have that
  \begin{align}
    & \Pr_{g} \Set{\vbig \left\lVert  \left(\Id \ot \Id \ot \transpose g \right)T \right\rVert_{\{1\},\{2\}} \ge t  \sigma }\notag\\
    & = \Pr_{g,h}\Set{\vbig \left\lVert  \left(\Id \ot \Id \ot \transpose{(g'+g'')}\right)T\right\rVert_{\{1\},\{2\}} \ge 2t \sigma }\nonumber\\
    & \le \Pr_{g,h}\Set{\vbig \left\lVert  \left(\Id \ot \Id \ot \transpose{(g')}\right)T\right\rVert_{\{1\},\{2\}} + \left\lVert  \left(\Id \ot \Id \ot \transpose{(g'')}\right)T\right\rVert_{\{1\},\{2\}}  \ge 2t \sigma }\nonumber\\
    & \le \Pr_{g,h}\Set{\vbig \left\lVert  \left(\Id \ot \Id \ot \transpose{(g')}\right)T \right\rVert_{\{1\},\{2\}} \ge t  \sigma } + \Pr_{g,h}\Set{\vbig \left\lVert \left(\Id \ot \Id \ot \transpose{(g'')}\right)T \right\rVert_{\{1\},\{2\}} \ge t \sigma }  \nonumber\\
    & \le 4(p+q)\cdot e^{-t^2/2}\mcom\nonumber
  \end{align}
  where the second line uses the triangle inequality, the third line uses the union bound, and the fourth line uses \pref{thm:matrix_concentration} applied to $g'$ and $g''$.
\end{proof}

\pref{cor:matrix_concentration_general} and \pref{thm:higher-tensor-norm} together imply the following theorem..

\begin{theorem}
  Let $k\in \N$  and $D$ be a degree-$(4k+10)$ pseudo-distribution over $\R^d$ that satisfies $\{\norm{u}^2\le 1\}_{D(u)}$.
  Let $g\sim \cN(0,\Sigma)$ be a Gaussian vector with covariance $\Sigma \preceq  \Id_d^{\otimes k}$.
  Then,
  \begin{equation}
    \E_{g\sim \cN(0,\Id_d)}\Norm{\pE_{D(u)} \langle  g,u^{\otimes k} \rangle\cdot \dyad u}
    \precsim \sqrt{k\log d\,} \Norm{\pE_{D(u)} \dyad u}\,.
  \end{equation}
\end{theorem}

We can apply Corollary~\ref{cor:matrix_concentration_general} repeatedly to obtain a bound for random contraction over a larger number of modes. 

\begin{theorem}\label{thm:matrix_concentration_multi_mode}
	Let $T \in \R^{p\times q\times r_1\times \dots \times r_s}$ be an order-$(s+2)$ tensor, and $g_1\sim \cN(0,\Sigma_1) ,\dots,g_s\sim \cN(0,\Sigma_s)$ be independent Gaussian random variables with covariance $\Sigma_i \preceq \Id_{r_i}$ for each $i\in [r]$.
	Let $\bar{r} = \max_{i\in [s]}\{r_i+2\}$.
	Then for any $t\ge 0$, 
	\begin{equation}
	\Pr_g \Set{\vbig \Bignorm{\left(\Id\otimes \Id \otimes \transpose{g_1} \otimes \ldots \otimes \transpose{g_s}\right)T}_{\{1\},\{2\}}
		\ge t^s \cdot \max_{S \subset[s]: 1\in S, 2\not\in S}\bigl\{\lVert  T\rVert_{S,S^c}\bigr\} } \le 4(p+q)\,\bar{r}^{s-1}\, e^{-t^2/2}\mper
	\label{eqn:cor:matrix_concentration_multi_modes}
	\end{equation}
\end{theorem}

\begin{proof}
	We prove by induction on $s$.
	The base case is exactly Corollary~\ref{cor:matrix_concentration_general}.
	For $s\ge 2$, suppose we have proved the $(s-1)$-case. 
	
	Let $T' = \left(\Id\otimes \Id \otimes \Id \otimes \transpose{g_2}\dots \otimes \transpose{g_s}\right)T$ be an order-3 tensor. Then we have that \begin{equation}
	\left(\Id\otimes \Id \otimes \transpose{g_1} \otimes \ldots \otimes \transpose{g_s}\right)T = \left(\Id \otimes \Id \otimes \transpose{g_1}\right) T' \mper\nonumber
	\end{equation} Then using Corollary~\ref{cor:matrix_concentration_general} on $T'$ and $g_1$, and then taking the expectation over $g_2,\dots, g_s$, we have 
	\begin{align}
		\Pr_{g_1,\dots, g_s} \Set{\vbig \Bignorm{\left(\Id\otimes \Id \otimes \transpose{g_1} \otimes \ldots \otimes \transpose{g_s}\right)T}_{\{1\},\{2\}}\ge  t \cdot \max\bigl\{\lVert  T'\rVert_{\{1\},\{2,3\}},\|T'\|_{\{2\},\{1,3\}}\bigr\}} \le 4(p+q)\, e^{-t^2/2}\mper \label{eqn:eqn83}
	\end{align}
	
	We view $T'$ as an order-$(s+1)$ tensor by merging the 2nd and 3rd modes, that is, $\left(\Id\otimes (\Id \otimes \Id) \otimes \transpose{g_2}\dots \otimes \transpose{g_s}\right)T$, and then apply the  inductive hypothesis. We obtain 
		\begin{align}
			\Pr_{g_2,\dots, g_s} \Set{\vbig \Norm{T'}_{\{1\},\{2,3\}}\ge  t^{s-1} \cdot \max_{S \subset[s]: 1\in S, \{2,3\}\cap S =\emptyset}\bigl\{\lVert  T\rVert_{S,S^c}\bigr\} } \le 4(p+qr_1)\, \bar{r}^{s-2}\cdot e^{-t^2/2}\mper\label{eqn:eqn84}
		\end{align}
	Similarly, we have that
 		\begin{align}
	\Pr_{g_2,\dots, g_s} \Set{\vbig \Norm{T'}_{\{1,3\},\{2\}}\ge  t^{s-1} \cdot \max_{S \subset[s]: \{1,3\}\in S, 2\not\in  S =\emptyset}\bigl\{\lVert  T\rVert_{S,S^c}\bigr\} } \le 4(pr_1+q)\, \bar{r}^{s-2}\cdot e^{-t^2/2}\mper\label{eqn:eqn85}
\end{align}
Using equations \pref{eqn:eqn83}, \pref{eqn:eqn84}, \pref{eqn:eqn85}, and applying union bound we obtain 
\begin{equation}
\Pr_g \Set{\vbig \Bignorm{\left(\Id\otimes \Id \otimes \transpose{g_1} \otimes \dots \otimes \transpose{g_s}\right)T}_{\{1\},\{2\}}
	\ge t^s \cdot \max_{S \subset[s]: 1\in S, 2\not\in S}\bigl\{\lVert  T\rVert_{S,S^c}\bigr\} } \le 4(p+q)\,\bar{r}^{s-1}\, e^{-t^2/2}\mcom \nonumber
\end{equation}
and complete the inductive proof. 
	
\end{proof}

\section{Decomposition of random overcomplete 3-tensors}\label{sec:random_tensor}
\newcommand{\pisqrt}[1]{(#1^{+})^{1/2}}
\newcommand{\covar}[1]{#1#1^{\trans}}
\newcommand{\covarp}[1]{#1(#1)^{\trans}}

\newcommand*{\Idsym}{\Id_{\mathrm{sym}}}
\newcommand*{\Idsymphi}{\Id_{\mathrm{sym}'}}
\newcommand*{\Qplushalf}{(Q^{+})^{1/2}}

In this section, we assume that we are given a random 3rd order overcomplete symmetric tensor $T$ of the following form
\begin{equation}
T = \sum\nolimits_{i=1}^{n}a_i\tp{3} + E
\mcom
\label{eq:def-tensor}
\end{equation}
where $n\le d^{1.5}/(\log d)^{O(1)}$, the vectors $a_i$ are drawn independently at random from the Euclidean unit sphere, and the error tensor $E$ satisfies $\norm{E}_{\{1\},\{2,3\}}\le \e$.

Let $\Idsym$ be the projection to the symmetric subspace of  $(\R^d)^{\otimes 2}$ (the span of all $x\tp{2}$ for $x\in \R^d$), and let $\Phi = \tfrac 1{\sqrt d} \sum_{i=1}^d e_i^{\otimes 2}\in(\R^{d})^{\otimes 2}$. Let $\Idsymphi$ be the projection to the subspace orthogonal to $\Phi$:
\begin{equation}
\Idsymphi = \Idsym- \Phi\Phi^{\trans} \mper\label{eqn:idsymphi}
\end{equation}

\begin{algorithm}[h]
  \caption{Polynomial-time algorithm for random overcomplete 3-tensor decomposition}  \label{alg:random-overcomplete}
  \textbf{Input:} Number $\epsilon > 0$ and $n\in \N$ and symmetric tensor $T\in (\R^d)^{\otimes 3}$ of the form~\eqref{eq:def-tensor}.
  
  \textbf{Find:} $\hat{a}_1,\dots, \hat{a}_n\in \R^d$.

  \textbf{Algorithm:}
  \vspace{-0.5em}
  \begin{enumerate}
  	\item Call \pref{alg:tensor-general} with \begin{align}\cA &= \Set{\langle T,u\tp{3}\rangle \ge 1-\epsilon, \|u\|^2 =1 }\mcom\label{eqn:eqn41} \\
  	P(u) &= \Idsymphi u\tp{2}\mcom\label{eqn:eqn39}\end{align} where $\Idsymphi$ is defined in~\eqref{eqn:idsymphi}. Suppose the outputs of \pref{alg:tensor-general} are $\hat{b}_1,\dots, \hat{b}_n$.
  \item Let $\hat{a}_i$ be $\tau_i$ the top eigenvector of the matrix reshaping of $\hat{b}_i$, where $\tau_i \in \{1, -1\}$ is chosen so that $T\hat{a}_1 > 0$.
  \end{enumerate}

\end{algorithm}

\begin{theorem*}[Restatement of \pref{thm:random-vectors}]
	With probability $1-d^{-\omega(1)}$ over the choice of random unit vectors $a_1,\ldots,a_n\in\R^d$, when given a symmetric $3$-tensor $T\in (\R^d)^{\otimes 3}$ as input, the output $\hat{a}_1,\ldots,\hat{a}_n\in\R^d$ of \pref{alg:random-overcomplete} satisfies
	\begin{equation}
	\dist_H\Paren{\vbig \Set{\hat{a}_1,\ldots,\hat{a}_{n}},\Set{a_1,\ldots,a_{n}}}^2
	\le O\Paren{\Paren{\frac n{d^{1.5}}}^{\Omega(1)} + \Norm{T - \sum\nolimits_{i=1}^n a_i^{\otimes 3}}_{\{1\},\{2,3\}}}
	\mper
	\end{equation}
\end{theorem*}

\pref{thm:random-vectors} follows immediately from \pref{thm:general-tensor-decomposition} and the following proposition:
\begin{proposition}\label{prop:random_3_tensor}
	With probability $1-d^{-\omega(1)}$ over the choice of random unit vectors $a_1,\dots, a_n$, the parameters $P(\cdot)$ and $\cA$ defined in \pref{alg:random-overcomplete} satisfy the requirements of Theorem~\ref{thm:general-tensor-decomposition}.
	In particular, let $c_i = P(a_i) = \Idsymphi a_i\tp{2}$.
	Then
	\begin{equation}
	\left\|\sum_{i=1}^{n}\dyad{c_i} \right\|\le 1+\delta\mcom\label{eqn:eqn42}
	\end{equation}
	where $\delta  = \tO(\sqrt{n}/d + n/d^{1.5})$, and
	\begin{equation}
	\cA\vdash\sum_{i=1}^{n}\inner{c_i, P(u)}^4 \ge \Big(1- O(\epsilon+1/d)\Big)\|P(u)\|^4\mper \label{eqn:eqn44}
	\end{equation}
\end{proposition}

We first show that by a simple extension of ~\cite[Theorem 4.2 and Lemma 8]{DBLP:conf/approx/GeM15}, $\cA$ implies that the sum of the terms $\inner{a_i,u}^8$ is large.
Note that $c_i\approx a_i\tp{2}$ and $P(u)\approx u\tp{2}$, and therefore this is already fairly close to our target inequality~\eqref{eqn:eqn44}.
\begin{lemma}[{Simple extension of~\cite[Theorem 4.2 and Lemma 8]{DBLP:conf/approx/GeM15}}]
	\label{lem:8th_power}
	With probability $1-d^{-\omega(1)}$ over the choice of random unit vectors $a_i$,
\begin{align}
\cA\vdash \Set{\sum_{i=1}^n \langle  a_i, u \rangle^3 \ge 1-\e, \|u\|^2 = 1}\vdash \Set{\sum_{i=1}^n\inner{a_i,u}^{8} \ge 1-O(\epsilon)-\delta}\mper\label{eqn:eqn10}
\end{align}
	where $\delta = \widetilde{O}(n/d^{3/2})$.
\end{lemma}
\begin{proof}
	Using the proof of~\cite[Theorem 4.2]{DBLP:conf/approx/GeM15} (specifically Lemma 3 and Claim 1), and the proof of Lemma 5 (specifically equation (11) and equation (15))  we have\footnote{Technically~\cite{DBLP:conf/approx/GeM15} only proved the case when the vectors $a_i$ are uniform over $\{\pm 1/\sqrt{d}\}^d$, though the proofs work for the uniform distribution over the unit sphere as well.}
	\begin{equation}
	\cA\vdash \Set{\sum_{i=1}^{n}\inner{u,a_i}^4 \ge 1- O(\epsilon) - \delta, \textup{ and } \sum_{i=1}^n \inner{u,a_i}^6 \ge 1-O(\epsilon) - \delta} \mper\label{eqn:eqn14}
	\end{equation}
	where $\delta = O(n\log^{O(1)}d/d^{3/2})$.
	Then we extend the proof using the same idea to higher powers:
	\begin{align}
	\vdash \left(\sum_{i=1}^{n} \inner{u,a_i}^6\right)^2
	& = \left\langle\left[\sum_{i=1}^{n} \inner{u,a_i}^5a_i\right],\, u\right\rangle^2 \le  \left\|\sum_{i=1}^{n}\inner{u,a_i}^5a_i\right\|^2 \nonumber\\
	& = \sum_{i=1}^{n}\inner{u,a_i}^{10} + \sum_{i\neq j}\inner{u,a_i}^5\inner{u,a_j}^5\inner{a_i,a_j} \nonumber\\
	& \le \sum_{i=1}^{n}\inner{u,a_i}^{10} + \left(\sum_{i=1}^n  \inner{u,a_i}^4\right)  \left(\sum_{i=1}^n \inner{u,a_i}^4\right) \max_{i\neq j}|\inner{a_i,a_j}|\mcom
	\label{eqn:eqn13}
	\end{align}
	where the first line uses the Cauchy-Schwarz inequality (\pref{lem:cauchy-schwarz}) and the last line uses the fact that $D(u)$ satisfies the constraint $-1\le \inner{u,a_i}\le 1$.
	By~\cite[Lemma 2]{DBLP:conf/approx/GeM15}, we have that
	\begin{equation}
	\cA\vdash\sum_{i=1}^n \inner{u,a_i}^4\le 1+\delta \mper
	\end{equation}
	Combining the equation above, equation~\eqref{eqn:eqn13}, equation~\eqref{eqn:eqn14}, and the fact that with high probability $\inner{a_i,a_j} \le \widetilde{O}(1/\sqrt{d})$, we obtain
	\begin{equation}
	\cA \vdash \sum_{i=1}^{n}\inner{u,a_i}^{10} \ge 1-O(\epsilon) - \delta\mper
	\end{equation}
	Therefore, using the fact that $\inner{u,a_i}^2 \le 1$, we complete the proof.
\end{proof}

\begin{lemma}[{Rephrasing of~\cite[Lemma 5.9]{HopkinsSSS16}}]
	\label{lem:subspace_concentration}
	Let $a_1, \dots, a_n \in \R^d$ be independent random vectors drawn uniformly from the Euclidean unit sphere with $1\le n\le d^{1.5}/\log^{O(1)}d$, and
	let $C$ be the matrix with columns $c_i =\Idsymphi a_i\tp{2}$.
	Then
	\begin{equation}
	\|\,\transpose{C}C-\Id_n\|\le \delta\mcom
	\end{equation}
	where $\delta  = \tO(\sqrt{n}/d + n/d^{1.5})$.
\end{lemma}
Though ~\cite[Lemma 5.9]{HopkinsSSS16} assumes $n\ge d$,
its proof can also handle $n\le d$ if the error bound is relaxed to $\tO(\sqrt{n}/d)$.
See specifically the end of the first paragraph of its proof.
Also while \cite[Lemma 5.9]{HopkinsSSS16} assumes Gaussian
random vectors, its proof reduces to the case on the unit sphere. Therefore we omit the proof of Lemma~\ref{lem:subspace_concentration}.

Finally we prove Proposition~\ref{prop:random_3_tensor}.
\begin{proof}[Proof of Proposition~\ref{prop:random_3_tensor}]
	Equation~\eqref{eqn:eqn42} follows from Lemma~\ref{lem:subspace_concentration}. To prove equation~\eqref{lem:subspace_concentration},
	 we essentially just replace $a_i^{\otimes 2}$ in equation~\eqref{eqn:eqn10} by $c_i$ and bound the approximation error. We have
	 \begin{align}
	 \cA \vdash \sum_{i=1}^{n}\inner{\Idsymphi u\tp{2},c_i}^4
	 & = \sum_{i=1}^{n} \inner{u^{\otimes 2},\Idsymphi a_i^{\otimes 2}}^4  =  \sum_{i=1}^{n} \left(\inner{u^{\otimes 2},a_i^{\otimes 2}} -\inner{\Phi, a_i^{\otimes 2}}\right)^4  \nonumber\\
	 & = \sum_{i=1}^{n} \left(\inner{u^{\otimes 2},a_i^{\otimes 2}} -1/d\right)^4  \nonumber\\
	 & \ge (1-1/d) \sum_{i=1}^n \langle u,a_i  \rangle^8 - O(1/d)\nonumber \\
	 & \ge 1 - O(\e) - O(1/d) - \tO(n/d^{3/2})\mper  \label{eqn:eqn12}
	 \end{align}
	 where the second last step uses $\vdash (x-y^2)^4 \ge (1-y^2)x^4 - O(y^2)$, %
	 and the last step uses equation~\eqref{eqn:eqn10}.
\end{proof}

\section{Robust decomposition of overcomplete 4-tensors}
\label{sec:foobi}

\makeatletter
\newcommand*\loweredwidetilde[1]{\mathpalette\loweredwidetildehelper{#1}}
\newcommand*\loweredwidetildehelper[2]{%
	\hbox{\dimen@\accentfontxheight#1%
		\accentfontxheight#11.25\dimen@
		$\m@th#1\widetilde{#2}$%
		\accentfontxheight#1\dimen@
	}%
}
\newcommand*\accentfontxheight[1]{%
	\fontdimen5\ifx#1\displaystyle
	\textfont
	\else\ifx#1\textstyle
	\textfont
	\else\ifx#1\scriptstyle
	\scriptfont
	\else
	\scriptscriptfont
	\fi\fi\fi3
}
\makeatother

\newcommand{\Qish}{\loweredwidetilde{Q}}
\newcommand{\Qhalf}{Q^{1/2}}
\newcommand{\Sish}{\loweredwidetilde{S}}
\newcommand{\Whalf}{\Qish^{1/2}}
\newcommand{\smin}{\sigma_{\min}}
\newcommand{\smax}{\sigma_{\max}}
\newcommand{\sig}{\sigma}
\newcommand{\Wplushalf}{(\Qish^{+})^{1/2}}

In this section we provide a sum-of-squares version of the FOOBI algorithm~\cite{DBLP:journals/tsp/LathauwerCC07}.
FOOBI yields the rank decomposition of a 4th order tensor $T = \sum_{i=1}^{n}a_i\tp{4}$ under the mild condition that the set $\{a_i^{\ot 2} \ot a_j^{\ot 2} - (a_i \ot a_j)^{\ot 2}\}_{i \ne j}$ is linearly independent.
FOOBI has not been formally shown to be robust to noise, though it's believed to tolerate spectral noise with magnitude up to some inverse polynomial of dimension.
In contrast, the noise tolerance of our sum-of-squares version depends only on the condition number of certain matrices, and not directly on the dimension.

In \pref{sec:foobi-condition-number} we additionally show, under a smoothed analysis model where each component $a_i$ of the input tensor is randomly perturbed, that the relevant condition numbers are never smaller than some inverse polynomial of the dimension, with high probability over the random perturbations.

Throughout this section we will work with an input tensor $T$ of the form
\begin{equation}
T = \sum_{i=1}^{n}a_i\tp{4} + E\nonumber
\end{equation}
where $n \le d^2$ and $E$ is a symmetric noise tensor with bounded spectral norm $\|E\|_{\{1,2\},\{3,4\}}$.

For a matrix $M$, we use $\sigma_{\max}(M),\sigma_{\min}(M)$ to denote its largest and smallest singular values respectively, and $\sigma_k(M)$ to denote its $k$th largest singular value.

Let $A\in \R^{d^2\times n}$ be the matrix with columns $a_i\tp{2}$ for $i=1,\dots, n$.
The guarantees of our algorithm will depend on the following 4th order condition number of $A$:
\begin{definition}\label{def:cond_number}
	For a full rank matrix $A\in \R^{d^2\times n}$ with columns $a_i\tp{2}$ for $i=1,\dots, n$, let $\kappa(A)$ defined as
	\begin{equation}
	\kappa(A) =\smax^{1.5}(Q)/\sigma_n^{1.5}(Q)+\smax^{2.5}(Q)/(\smin^2(B)\sigma^{0.5}_n(Q))\mper
	\end{equation}
	where $Q = AA^{\trans}$ and $B\in \R^{d^4\times n(n-1)}$ is the matrix with columns $b_{i,j} = a_i\tp{2}\otimes a_j\tp{2}-(a_i\otimes a_j)\tp{2}$ for every $i \neq j$.
\end{definition}

\begin{theorem}[Restatement of \pref{thm:smoothed-vectors}]\label{thm:sos_foobi}
	Let $\delta > 0$.
	Let $T\in (\R^d)^{\otimes 4}$ be a symmetric $4$-tensor and $\{a_1,\ldots,a_n\}\subseteq \R^d$ be a set of vectors.
	Define $E = T - \sum\nolimits_{i=1}^n a_i^{\otimes 4}$ and define $A$ as the matrix with columns $a_i\tp{2}$.
	If $\Norm{E}_{\{1,2\},\{3,4\}} \le \delta\,\sig_n(\dyad{A})$ and \pref{alg:foobi} outputs $\{\hat{a}_1,\ldots,\hat{a}_n\}$ on input $T$, then there exists a permutation $\pi:[n] \to [n]$ so that for every $i \in [n]$,
	\begin{equation}
	\Norm{a_i-\hat{a}_{\pi(i)}} \le O\big(\delta\,\kappa(A)\big) \Norm{a_i} \mper
	\end{equation}
\end{theorem}
\begin{algorithm}
	\caption{Sum-of-Squares FOOBI for robust overcomplete 4-tensor decomposition}
	\label{alg:foobi}
	\textbf{Input:} Number $\delta > 0$ and symmetric tensor $T\in (\R^d)^{\otimes 4}$.
	
	\textbf{Find:} $\hat{a}_1,\dots, \hat{a}_n\in \R^d$.
	
	\textbf{Algorithm:}
	\vspace{-0.5em}
	\begin{enumerate}
		\item
		Compute the best \Jnote{} \TMnote{Add footnote. Though I think this is rather standard terminology} rank-$n$ approximation\footnote{Throughout the paper, best rank-$r$ approximation refers to the closest rank-$r$ matrix in either Frobenius norm or spectral norm distance. } $\Qish$ of the $d^2\times d^2$ matrix reshaping of $T$.
		Let $\Sish$ be the column span of $\Qish$.
		\item
		Run \pref{alg:tensor-general} with inputs $P(\cdot)$ and $\cA$ set to
		\begin{align}
		P(x) & = (\Qish^+)^{1/2}x\tp{2}\mcom\label{eqn:P}\\
		\cA & = \Set{\|\Id_{\Sish} x^{\otimes 2}\|^2 \ge (1 - 3\delta)\,\|x\|^4}_{x} \mper \label{eqn:foobi_main_constraint}
		\end{align}
        Suppose the algorithm outputs $\hat{c}_1,\dots,\hat{c}_n$.
		\item
		Output $\hat{a}_1,\dots, \hat{a}_n$ such that for each $i \in [n]$, the matrix $\dyad{\hat{a}_i}$ is the best \Jnote{} rank-1 approximation of the matrix reshaping of $\Qish^{1/2}\hat{c}_i$.
	\end{enumerate}
\end{algorithm}

Let $\Qish$ be the best rank-$n$ approximation of the $d^2\times d^2$ matrix reshaping of $T$, and let $\Sish$ be the column space of $\Qish$.
These two objects serve as our initial best-guess approximations of $Q = \dyad{A}$ and the subspace $S$ spanned by $\{a_1\tp{2},\dots, a_n\tp{2}\}$ (also the column space of $Q$), which we do not have access to.
We define $B\in \R^{d^4\times n(n-1)}$ as the matrix with columns $b_{i,j} = a_i\tp{2}\otimes a_j\tp{2}-(a_i\otimes a_j)\tp{2}$ for $i \neq j$.

One of the core techniques in the analysis will be to use (the following rephrased version of) Davis and Kahan's ``$\sin \theta$'' Theorem, which bounds the principle angle between the column spaces of two matrices that are spectrally close to each other.
\begin{theorem}[Direct consequence of Davis-Kahan Theorem~\cite{doi:10.1137/0707001}]
	\label{thm:Davis-Kahan}
	Suppose symmetric PSD matrices $Q\in \R^{D\times D}$ and $\Qish \in \R^{D\times D}$ of rank $n\le D$ satisfy $\|Q-\Qish\|\le \delta\,\sigma_n(Q)$.
	Let $S$ and $\Sish$ be the column spaces of $Q$ and $\Qish$ respectively, and assume $\delta \le \frac{1}{2}$.
	Then we have
	\begin{equation}
	\sin(S,\Sish)\defeq\|\Id_S-\Id_{\Sish}\Id_S\| = \|\Id_{\Sish}-\Id_S\Id_{\Sish}\|\le \delta/(1-\delta)\mper\label{eqn:davis_kahan}
	\end{equation}
	Consequently,
	\begin{equation}
	\|\Id_S-\Id_{\Sish}\|\le O(\delta)\mper\label{eqn:eqn22}
	\end{equation}
\end{theorem}

\pref{thm:sos_foobi} follows from the analysis of \pref{alg:tensor-general} as given in \pref{thm:general-tensor-decomposition}, as long as we can verify its two conditions, which we restate in the following two propositions.
While \pref{prop:foobi_reduction_easy} follows quickly from \pref{thm:Davis-Kahan}, we prove \pref{prop:foobi_reduction} over the next three subsections.

\begin{proposition}\label{prop:foobi_reduction_easy}
	Let $P(x)$ and $\cA$ be as defined in \pref{alg:foobi}.
	Then each vector $a_1,\dots, a_n$ satisfies $\cA$.
\end{proposition}
\begin{proof}
	By \pref{thm:Davis-Kahan}, $\|\Id_{\Sish} a_i\tp{2} \| \ge \|\Id_Sa_i\tp{2}\|-\|(\Id_S-\Id_{\Sish})\,a_i\tp{2}\| \ge (1-2\delta)\,\|a_i\|^4$.
\end{proof}

\begin{proposition}\label{prop:foobi_reduction}
	Let $P(x)$ and $\cA$ be as defined in \pref{alg:foobi}. Then
	\begin{equation}
		\cA \;\vdash_8\; \sum_{i=1}^n \inner{P(a_i), P(x)}^4\ge (1-\tau)\,\|P(x)\|^4 \mcom\nonumber
	\end{equation}
	where $\tau \le O(\delta\,\smax^2(Q)/\smin^2(B))+O(\delta\,\smax(Q)/\sigma_n(Q))$.
\end{proposition}

\begin{proof}[Proof of Theorem~\ref{thm:sos_foobi}]
	By \pref{thm:general-tensor-decomposition} along with \pref{prop:foobi_reduction_easy} and \pref{prop:foobi_reduction}, step 2 in \pref{alg:foobi} must yield vectors $\hat{c}_1,\dots,\hat{c}_n$ that are respectively $O(\tau)$-close to $P(a_1),\dots, P(a_n)$, where $\tau \le  O(\delta\,\smax^2(Q)/\smin^2(B))+O(\delta\,\smax(Q)/\sigma_n(Q))$.
	Then
	\begin{align*} \|a_i\tp{2}- \Qish^{1/2}\hat{c}_i\| & \le  \|a_i\tp{2}- \Qish^{1/2}P(a_i)\|+ \|\Qish^{1/2}(P(a_i)-\hat{c}_i)\| \\
		& \le \|a_i\tp{2}- \Id_{\Sish}a_i^{\otimes 2}\| + \smax^{1/2}(Q) \cdot O(\tau)\\
		&\le \|(\Id_S-\Id_{\Sish})a_i^{\otimes 2}\| + \smax^{1/2}(Q) \cdot O(\tau)\\%
		& \le O(\delta\,\smax^{1/2}(Q)) + \smax^{1/2}(Q)\cdot O(\tau) \\
		& \le \smax^{1/2}(Q)\cdot O(\tau) \\
		& \le \smax^{1/2}(Q)/\sigma_n^{1/2}(Q)\cdot O(\tau)\cdot \|a_i\tp{2}\| \\
		& = O(\delta\,\kappa(A))\, \|a_i\tp{2}\| \mper
	\end{align*}
	Therefore \Jnote{}\TMnote{Wedin's Theorem?} taking the best \Jnote{} rank-1 approximation of the matrix reshaping of $\Qish^{1/2}\hat{c}_i$ gives an $O(\delta\,\kappa(A))$-approximation of $a_i$.
\end{proof}

\subsection{Noiseless case}
We first prove \pref{prop:foobi_reduction} in the noiseless case, when $T = \sum a_i\tp{4}$ and $\Qish = Q$ and $\Sish = S$.
In this scenario, we find that the left-hand side of the conclusion of \pref{prop:foobi_reduction} becomes
\begin{align*}
\sum_{i=1}^n \inner{P(a_i), P(x)}^4
&= \sum_{i=1}^{n}\left\langle(Q^+)^{1/2}a_i\tp{2},\, (Q^+)^{1/2}x\tp{2}\right\rangle^4
\\&= \sum_{i=1}^n \left[\transpose{(a_i\tp{2})} Q^{+} x\tp{2}\right]^4
\\&= \left\|\transpose A Q^{+} x\tp{2}\right\|_4^4 \mper
\end{align*}
The term $\|P(x)\|_2^4$ on the right becomes $\|(Q^+)^{1/2}x\tp{2}\|_2^4$.
Thus \pref{prop:foobi_reduction} becomes
\begin{proposition}[Noiseless \pref{prop:foobi_reduction}]\label{prop:foobi_reduction_noiseless}
	Let
	\begin{equation}
	\cA'  = \Set{\|\Id_{S} x^{\otimes 2}\|_2^2 \ge (1 - c\,\delta)\,\|x\|_2^4}_{x} \mcom \label{eqn:foobi_main_constraint_noiseless}
	\end{equation}
	for some constant $c \ge 0$. Then
	\begin{equation}
	\cA' \;\vdash_8\; \left\|\transpose A Q^{+} x\tp{2}\right\|_4^4 \ge (1-\tau)\left\|\Qplushalf x\tp{2}\right\|_2^4 \mcom\nonumber
	\end{equation}
	where $\tau \le O(\delta\,\smax^2(Q)/\smin^2(B))$, where $c$ is treated as a constant in the big-$O$ notation.
\end{proposition}
\begin{proof}
We write $x\tp{2}$ as a linear combination of the vectors $a_i\tp{2}$ plus some term orthogonal to $S$.
\begin{align}
  \vdash\; x\tp{2}  & = \Id_S x\tp{2} + \Id_{S^{\perp}} x\tp{2}\nonumber \\
           & = \left[\sum_{i=1}^n \alpha_i\, a_i\tp{2}\right] +\Id_{S^{\perp}} x\tp{2}\mcom\nonumber
\end{align}
where $\alpha = A^{+}x\tp{2}$ is a $n$-dimensional vector with polynomial entries.
Since $Q = \dyad{A}$, it follows that $\transpose A Q^{+} x\tp{2} = \alpha$ and $\|\Qplushalf x\tp{2}\|_2^4 = \|\alpha\|_2^4$, so that it will suffice to show $\|\alpha\|_4^4 \ge (1 - \tau)\,\|\alpha\|_2^4$.

We consider $x\tp{4}$:
\begin{align}
  \vdash\; x\tp{4} = x\tp{2}\otimes x\tp{2} &= \left[\sum_{i=1}^{n}\alpha_i\,a_i\tp{2}\right] \otimes \left[\sum_{i=1}^{n}\alpha_i\,a_i\tp{2}\right] + \zeta \nonumber\\
          & = \left[\sum_{i, j\in [n]} \alpha_i\alpha_j\,a_i\tp{2}\otimes a_j\tp{2}\right] + \zeta\mcom\label{eqn:xtensor4}
\end{align}
with the error term $\zeta = (\Id_S x\tp{2} + \Id_{S^{\perp}} x\tp{2})\tp{2} - (\Id_S x\tp{2})\tp{2}$, so that $\cA' \vdash \|\zeta\|_2^2 \le O(\delta)\,\|x\|_2^8$, since $\cA' \vdash \|\Id_{S^{\perp}} x\tp{2}\|_2^2 \le O(\delta)\,\|x\|_2^4$ from the definition of $\cA'$.

Since $x^{\otimes 4}$ is invariant with respect to permutation of its tensor modes, we can also write it as
\begin{equation}
\vdash\; x\tp{4} =  \left[\sum_{i=1}^{n}\alpha_i\alpha_j\,(a_i\otimes a_j)\tp{2}\right]+ \zeta'\mcom\label{eqn:xtensor4_alternate}
\end{equation}
where similarly $\cA' \vdash \|\zeta'\|_2^2 \le O(\delta)\,\|x\|_2^8$.

Therefore, taking the difference of constraints~\eqref{eqn:xtensor4} and~\eqref{eqn:xtensor4_alternate} and recalling the definition of $B$ being the matrix with columns $b_{i,j} = a_i\tp{2}\otimes a_j\tp{2}-(a_i\otimes a_j)\tp{2}$, we obtain
\begin{align}
  \cA'\;\vdash\;\left\|\sum_{i\neq j} \alpha_i\alpha_j\, b_{i,j} \right\|_2^2 = \|\zeta' - \zeta\|_2^2 \le O(\delta)\,\|x\|_2^8 \mcom\nonumber
\end{align}
so that therefore since $\|Bv\|_2^2 \ge \smin^2(B)\,\|v\|_2^2$ for all vectors $v$,
\begin{align}
  \cA'\;\vdash\; \sum_{i\neq j}&\alpha_i^2\alpha_j^2 \cdot \smin^2(B)
  \le \left\|\sum_{i\neq j} \alpha_i\alpha_j\, b_{ij} \right\|_2^2
  \le O(\delta)\,\|x\|_2^8 \nonumber
  \\
  &\le O(\delta)\left[\vphantom{\bigg|}\smax(Q)\,x\tp{2}Q^+x\tp{2}\right]^2
  \le O(\delta)\,\smax^2(Q)\,\|\alpha\|_2^4\mper\nonumber
\end{align}
Hence, substituting in the above inequality,
\[
  \cA'\;\vdash\;  \|\alpha\|_4^4 = \|\alpha\|_2^4 - \sum_{i\neq j}\alpha_i^2\alpha_j^2 \ge (1-O(\delta\,\smax^2(Q)/\smin^2(B)))\,\|\alpha\|_2^4\mper \qedhere
\]
\end{proof}

\subsection{Noisy case}

At this point we've proved a version of \pref{prop:foobi_reduction} in the special case where there is no noise.
In order to handle noise we need to show two things: first that the noisy set of polynomial constraints $\cA$ used in \pref{alg:foobi} implies the noiseless version $\cA'$ from \pref{prop:foobi_reduction_noiseless}, and second that the desired conclusion of $\cA$ in \pref{prop:foobi_reduction} follows from its noiseless counterpart in \pref{prop:foobi_reduction_noiseless}.

The first step follows immediately from \pref{thm:Davis-Kahan}:
\begin{lemma}\label{lem:foobi_constraints_reduction}
	Let $\cA$ be defined as in \pref{alg:foobi} and $\cA'$ be defined as in \pref{prop:foobi_reduction_noiseless}. Then $\cA \vdash \cA'$.
\end{lemma}
\begin{proof}
	By \pref{thm:Davis-Kahan}, $\|\Id_Sx\tp{2}\|_2^2 = \|\Id_{\Sish}x\tp{2}\|_2^2- (x\tp{2})^{\trans}(\Id_{\Sish}-\Id_S)\,x\tp{2} \ge (1-O(\delta))\,\|x\|_2^4$.
\end{proof}
For the second step, we have by \pref{prop:foobi_reduction_noiseless} and \pref{lem:foobi_constraints_reduction} a statement of the form
\[ \cA \;\vdash_8\; \left\|\transpose{A}Q^+ x\tp{2}\right\|_4^4 \ge (1-\tau')\,\left\|\Qplushalf x\tp{2}\right\|_2^4 \mcom \]
but to prove \pref{prop:foobi_reduction} we need a statement of the form (after expanding out the function $P$)
\[ \cA \;\vdash_{\ell}\; \left\|\transpose{A}\Qish^+ x\tp{2}\right\|_4^4 \ge (1-\tau)\,\left\|\Wplushalf x\tp{2}\right\|_2^4 \mper \]
Thus what remains is to show that not too much is lost when we approximate $Q^+$ with $\Qish^+$.
\begin{lemma}
	\label{lem:Qplus_approx}
	Suppose symmetric PSD matrices $Q\in \R^{D\times D}$ and $\Qish \in \R^{D\times D}$ both of rank $n\le D$ satisfy $\|Q-\Qish\|\le \delta\,\sigma_n(Q)$.
	Then $\|Q(Q^+ - \Qish^+)\| \le O(\delta)$
	and similarly $\big\|\Qhalf\big(\Qplushalf - \Wplushalf\big)\big\| \le O(\delta)$.
\end{lemma}
\begin{proof}
	By \pref{thm:Davis-Kahan}, $\|QQ^+ - \Qish\Qish^+\| \le \|\Id_S - \Id_{\Sish}\| \le O(\delta)$,
	where $S$ and $\Sish$ are the column spaces of $Q$ and $\Qish$ respectively.
	Then by adding and subtracting a term of $Q\Qish^+$,
	\[ \left\|Q(Q^+ - \Qish^+) + (Q - \Qish)\,\Qish^+\right\| \le O(\delta)\mper \]
	By triangle inequality,
	\begin{align*}
		\left\|Q(Q^+ - \Qish^+)\right\|
		\;\le\; O(\delta) + \left\|(Q - \Qish)\,\Qish^+\right\|
		\;\le\; O(\delta) + \|Q - \Qish\|\cdot\|\Qish^+\|
		\;\le\; O(\delta) \mper
	\end{align*}
	The analogous result for $\big\|\Qhalf\big(\Qplushalf - \Wplushalf\big)\big\|$ is obtained by substituting $\Qplushalf$ for $Q^+$ and $\Wplushalf$ for $\Qish^+$ in the above argument.
\end{proof}

We also show that not too much is lost when approximating a vector with high $\ell_4/\ell_2$ ratio.
\begin{lemma}\label{lem:4power_reduction}
	For $\gamma, \beta, \tau\le 1/2$, let $\cB$ be the set of polynomial inequalities $$\cB = \Set{-\beta\|v\|_2^2 \le \|u\|_2^2 -\|v\|_2^2\le \beta\|v\|_2^2,\;\; \|u-v\|_2^2 \le \gamma \|v\|_2^2 }\;\cup\; \Set{\|v\|_4^4 \ge (1-\tau)\|v\|_2^4}\mper$$
	Then we have
	\begin{equation}
	\cB\vdash_4 \Set{\|u\|_4^4 \ge (1-\tau - O(\sqrt{\gamma} + \beta))\,\|u\|_2^4}\mper\nonumber
	\end{equation}
\end{lemma}

\begin{proof}[Proof of Lemma~\ref{lem:4power_reduction}]
	We have that $ \Set{\|u-v\|_2^2 \le \gamma \|v\|_2^2}\vdash(u_i-v_i)^2 \le \gamma \|v\|^2$. Moreover, since $\cB\vdash \|u\|_2^2 \le (1+\gamma)\|v\|_2^2$ we have that $\cB\vdash (u_i+v_i)^2\le \|u+v\|_2^2 \le O(\|v\|_2^2)$.  Therefore it follows that
	\begin{align}
		\cB\;\vdash\; v_i^2 - u_i^2  \;&=\; (v_i-u_i)(u_i+v_i)\;\le\; \sqrt{\gamma}/2 \cdot (u_i+v_i)^2 + 1/\sqrt{\gamma}\cdot (u_i-v_i)^2 \nonumber\\
		& \le\; O(\sqrt{\gamma}\,\|v\|_2^2)\mcom
		\nonumber
	\end{align}
	where we used the AM-GM inequality.
	It follows that for every $i$,
	\begin{equation}
	\cB \;\vdash\;\sum_{j\neq i} u_j^2 - v_j^2 \;\le\; \|u\|^2 - \|v\|^2 + v_i^2-u_i^2 \;\le\; O((\sqrt{\gamma}+\beta)\,\|v\|_2^2)\nonumber\mper
	\end{equation}
	Therefore by two rounds of adding-and-subtracting,
	\begin{align}
		\cB \;\vdash\;\sum_{i\neq j}u_i^2u_j^2
		& = \sum_{i}u_i^2 \left(\sum_{j\neq i}u_j^2 - \sum_{j\neq i}v_j^2\right) + \sum_{i}v_i^2 \left(\sum_{j\neq i}u_j^2 - \sum_{j\neq i}v_j^2\right) + \sum_{i\neq j}v_i^2 v_j^2 \nonumber\\
		& \le\sum_{i}u_i^2 \cdot O((\sqrt{\gamma}+\beta)\,\|v\|_2^2) + \sum_{i}v_i^2  \cdot O((\sqrt{\gamma}+\beta)\,\|v\|_2^2)+  \sum_{i\neq j}v_i^2 v_j^2  \nonumber\\
		& = O(\sqrt{\gamma}+\beta)\,(\|u\|_2^2\|v\|_2^2+\|v\|_2^4) + \|v\|_4^4 - \|v\|_2^4 \nonumber\\
		& \le (\tau+O(\sqrt{\gamma}+\beta))\,\|v\|_2^4\nonumber\\
		& \le  (\tau+O(\sqrt{\gamma}+\beta))\,\|u\|_2^4\mper\nonumber
	\end{align}
	Here in the second line we used the axiom that $\|u\|_2^2 - \|v\|_2^2 \le \beta\|v\|_2^2$, the second-to-last line uses the axiom $\|v\|_4^4 \ge (1-\tau)\|v\|_2^4$, and the last one uses $\|u\|_2^2 - \|v\|_2^2 \ge -\beta\|v\|_2^2$ so that $\|v\|_2^2 \le (1-\beta)^{-1}\|u\|_2^2$. Rearranging the final inequality above we obtain the desired result.
\end{proof}

\begin{proof}[Proof of \pref{prop:foobi_reduction}]
	We know by \pref{prop:foobi_reduction_noiseless} and \pref{lem:foobi_constraints_reduction}
	\begin{equation}
	\label{eq:halfway_there}
	\cA \;\vdash_8\; \left\|\transpose{A}Q^+ x\tp{2}\right\|_4^4 \ge (1-\tau')\,\left\|\Qplushalf x\tp{2}\right\|_2^4 \mcom
	\end{equation}
	where $\tau' \le O(\delta\,\smax^2(Q)/\smin^2(B))$.
	
	By \pref{lem:Qplus_approx},
	\begin{align*}
		\vdash\; \left\|\transpose{A}Q^+ x\tp{2} - \transpose{A}\Qish^+ x\tp{2}\right\|_2^2
		&= \left\|\transpose{A}(Q^+ - \Qish^+)\,x\tp{2}\right\|_2^2
		\\&= \left\|A^+\,Q(Q^+ - \Qish^+)\,x\tp{2}\right\|_2^2
		\\&\le \|A^+\|_2^2 \cdot O(\delta^2)\cdot\|x\|_2^4
		\\&\le O(\delta^2)\, \sigma_n^{-1}(Q)\,\|x\|_2^4
		\mper
	\end{align*}
	Also, using \pref{lem:Qplus_approx} after adding and subtracting a term of $Q^+Q\Qish^+$,
	\begin{align*}
		\vdash\; \left\|\transpose{A}Q^+ x\tp{2}\|_2^2 - \|\transpose{A}\Qish^+ x\tp{2}\right\|_2^2
		&= \transpose{(x\tp{2})}Q^+\dyad{A}Q^+x\tp{2} - \transpose{(x\tp{2})}\Qish^+\dyad{A}\Qish^+x\tp{2}
		\\&= \transpose{(x\tp{2})}[Q^+QQ^+ - \Qish^+Q\Qish^+]\,x\tp{2}
		\\&= \transpose{(x\tp{2})}[Q^+Q(Q^+ - \Qish^+) + (Q^+ - \Qish^+)Q\Qish^+]\,x\tp{2}
		\\&\le \transpose{(x\tp{2})}[O(\delta)\,Q^+ + O(\delta)\,\Qish^+]\,x\tp{2}
		\\&\le O(\delta)\,\sigma_n^{-1}(Q)\,\|x\|_2^4
		\mcom
	\end{align*}
	and similarly in the other direction. Furthermore,
	\begin{align*}
		\vdash\;  \left\|\transpose{A}Q^+ x\tp{2}\right\|_2^2 = \left\|A^+ x\tp{2}\right\|_2^2
		\ge \smax^{-1}(Q)\,\|x\|_2^4
		\mper
	\end{align*}
	Combining the above three inequalities with \pref{lem:4power_reduction} and \pref{eq:halfway_there},
	\[ \cA \;\vdash_{8}\; \left\|\transpose{A}\Qish^+ x\tp{2}\right\|_4^4 \ge (1 - \tau)\left\|\transpose{A}\Qish^+ x\tp{2}\right\|_2^4\mcom \]
	where $\tau \le O\big(\delta\,\smax^2(Q)/\smin^2(B) + \delta\,\smax(Q)/\sigma_n(Q)\big)$.
	
	Finally, using the fact that $A^+\Qhalf$ is a whitened matrix and therefore has orthonormal rows and then using triangle inequality with \pref{lem:Qplus_approx},
	\begin{align*}
		\vdash\;  \left\|\transpose{A}\Qish^+ x\tp{2}\right\|_2^2
		&= \left\|A^+\Qhalf\cdot\Qhalf\Wplushalf\cdot\Wplushalf x\tp{2}\right\|_2^2
		\\&= \left\|\Qhalf\Wplushalf\cdot\Wplushalf x\tp{2}\right\|_2^2
		\\&\ge \left\|\Qhalf\Qplushalf\cdot\Wplushalf x\tp{2}\right\|_2^2 - \left\|\Qhalf[\Qplushalf - \Wplushalf]\cdot\Wplushalf x\tp{2}\right\|_2^2
		\\&\ge \left\|\Id_S \cdot\, \Wplushalf x\tp{2}\right\|_2^2 - O(\delta) \cdot\left\|\Wplushalf x\tp{2}\right\|_2^2
		\\&\ge (1 - O(\delta))\,\left\|\Wplushalf x\tp{2}\right\|_2^2
		\mper
	\end{align*}
	Combining the above two inequalities we obtain the theorem.
\end{proof}

\subsection{Condition number under smooth analysis}
\label{sec:foobi-condition-number}
\newcommand{\tila}{\tilde{a}}
\let\wtilde\loweredwidetilde

In this section we prove that the condition number $\kappa(A)$ is at least inverse polynomial under the smooth analysis framework~\cite{Spielman:2004:SAA:990308.990310}.
We work with the same $\rho$-perturbation model as introduced by~\cite{ DBLP:conf/stoc/BhaskaraCMV14}: Each $\tilde{a}_i$ is generated by adding a Gaussian random variable with covariance matrix $\frac{\rho}{d}\Id_d$ to $a_i$.
We are given a symmetric 4th order tensor $\sum_{i=1}^n \tila_i\tp{4}$ (with noise).
Let $\wtilde{A}$ be the corresponding matrix with columns $\tila_i\tp{2}$. We will give an upper bound on $\kappa(\wtilde{A})$.
Suppose the vectors $a_i$ have bounded norm; then $\smax(Q)$ is bounded, and therefore an upper bound on $\kappa(\wtilde{A})$ follows from establishing lower bounds on $\smin(\wtilde{B})$ and $\smin(\dyad{\wtilde{A}})$.
The lower bound on the latter follows from~\cite{ DBLP:conf/stoc/BhaskaraCMV14} and therefore we focus on the former.
\begin{theorem}\label{thm:smooth}
	Let $n \le \frac{d^2}{10}$ and $\tilde{a}_1,\dots, \tilde{a}_n$ be independent $\rho$-perturbations of $a_1,\dots, a_n$.
	Let $\wtilde{B}\in \R^{d^4\times n(n-1)}$ be the matrix with columns $\tilde{b}_{ij} =  \tilde{a}_i\tp{2}\otimes \tila_j\tp{2}-(\tila_i\otimes \tila_j)\tp{2}$ for $i\neq j$.
	Then with probability $1-\exp(-d^{\Omega(1)})$, we have $\smin(\wtilde{B})\ge \poly(1/d, \rho)$.
\end{theorem}

We will bound the smallest singular value using the leave-one-out distance defined by~\cite{rudelson2009smallest}.
\begin{lemma}\cite{rudelson2009smallest}\label{lem:dA}
	For matrix $A\in \R^{d\times n}$ with columns $A_i, i\in [n]$, let $S_{-i}$ be the span of the columns without $A_i$, and
	$d(A) = \min_{i\in [n]} \|(\Id - \Id_{S_{-i}})\,A_i\|$. Then $\smin(A)\ge \frac{1}{\sqrt{n}}d(A)$.
\end{lemma}

To bound $d(\wtilde{B})$ from below, we use~\cite[Theorem 3.9]{ DBLP:conf/stoc/BhaskaraCMV14} as our main tool.

\begin{theorem}~\cite[Theorem 3.9]{ DBLP:conf/stoc/BhaskaraCMV14}\label{thm:bcmv}
	Let $\delta \in (0,1)$ be a constant and $W$ be an operator from $\R^{n^{\ell}}$ to $\R^{m}$ such that $\sigma_{\delta n^{\ell}}(W)\ge \eta$.
Then for  any $a_1,\dots, a_{\ell}\in \R^d$ and  their $\rho$-perturbations $\tilde{a}_1,\dots, \tilde{a}_{\ell}$,
	\begin{equation}
	\Pr\left[\left\|\vphantom{\big|}W (\tilde{a}_1\otimes \dots \otimes \tilde{a}_{\ell}) \right\| \ge \eta \rho^{\ell} d^{-O(3^{\ell})}\right]\le 1-\exp(-\delta d^{1/3^{\ell}})
	\end{equation}
\end{theorem}

Towards bounding the least singular value of $\wtilde{B}$ using Theorem~\ref{thm:bcmv}, we need to address two issues that don't exist in~\cite{DBLP:conf/stoc/BhaskaraCMV14}.
The first one is that
Theorem~\ref{thm:bcmv} requires $\tilde{a}_1,\dots, \tilde{a}_{\ell}$ to be independent perturbations of $a_1,\dots, a_{\ell}$.
However, we need to deal with $\tila_i\otimes \tila_i\otimes\tila_j\otimes\tila_j$ which is a correlated perturbation of $a_i \ot a_i \ot a_j \ot a_j$.
We will use (a simpler version of) the decoupling technique of~\cite{de1999decoupling} and focus on a sub-matrix of $\wtilde{B}$ where the noise is un-correlated.

The second difficulty is that the columns of $\wtilde{B}$ are also correlated since each $\tilde{a}_i$ is used in $n$ columns.
Therefore when the leave-one-out distance of $\wtilde{B}$ is under consideration, the column $\wtilde{B}_{ij}$ and the subspace of the rest of the columns have correlated randomness, which prevents us from using Theorem~\ref{thm:bcmv} directly.
We will address this issue by projecting $\wtilde{B}_{ij}$ into a smaller subspace which is un-correlated with $\wtilde{B}_{ij}$ and then apply Theorem~\ref{thm:bcmv}. %

\begin{proof}[Proof of Theorem~\ref{thm:smooth}]
	We partition $[d]$ into 4 disjoint subsets $L_1,L_2, L_3, L_4$ of size $d/4$.
	Let $\wtilde{B}'$ be the set of rows of $\wtilde{B}$ indexed by $L_1\times L_2\times L_3\times L_4$.
	That is, the columns of $\wtilde{B}'$ are $\tila_{i,L_1}\otimes (\tila_{i,L_2}\otimes \tila_{j,L_3} -\tila_{j,L_2}\otimes \tila_{i,L_3})\otimes \tila_{j,L_4}$, for $i\neq j$, where $\tila_{i,L}$ denotes the restriction of vector $a_i$ to the subset $L$.

	We fix a column $\wtilde{B}_{ij}'$ with $i\neq j$. Let $V = \textup{span}\{\wtilde{B}_{k\ell}' : (k,\ell) \neq (i,j) \}$. Clearly $V$ is correlated with $\wtilde{B}_{ij}'$. We define the following subspace that contains $V$,  %
	\begin{align}
          \hat{V}
          & = \mathrm{span}\biggl\{
            \begin{aligned}[t]
              & \tilde{a}_{j,L_1}\otimes x\otimes y\otimes \tilde{a}_{i,L_4}\mcom\\
              &\tilde{a}_{k,L_1}\otimes \tilde{a}_{k,L_2}\otimes x\otimes y\mcom \\
              &\tilde{a}_{k,L_1}\otimes x\otimes \tilde{a}_{k,L_3}\otimes y\mcom \\
              &x\otimes y\otimes \tilde{a}_{k,L_3}\otimes \tilde{a}_{k,L_4}\mcom\\
              & x\otimes \tilde{a}_{k,L_2}\otimes y\otimes \tilde{a}_{k,L_4} \biggm \vert x,y\otimes \R^{d/4}, k \not \in \{i,j\} \biggr\}
            \end{aligned}
	\end{align}
	Therefore by definition $\hat{V}\supset V$, and thus $\hat{V}^{\perp} \subset V^{\perp}$ where $V^{\perp}$ denotes the subspace orthogonal to $V$.
        Observe that by the definition of $\hat{V}$, we have $\tila_{i,L_1}\otimes \tila_{i,L_2}\otimes \tila_{j,L_3} \otimes \tila_{j,L_4}$ is independent from $\hat{V}$.
        Moreover, $\hat{V}$ has dimension at most $d^2+d^2 + 4nd^2 < d^4/2$. Then by Theorem~\ref{thm:bcmv} we obtain that with probability at least $1-\exp(-d^{\Omega(1)})$,
	\begin{equation}
		\left\|\Id_{\hat{V}^{\perp}} \tila_{i,L_1}\otimes \tila_{i,L_2}\otimes \tila_{j,L_3} \otimes \tila_{j,L_4}\right\|\ge \poly(1/d, \rho)\mper\nonumber
	\end{equation}

	Consequently,
	\begin{align}
		\left\|\Id_{V^{\perp}} \wtilde{B}'_{ij}\right\| \ge 		\left\|\Id_{\hat{V}^{\perp}} \wtilde{B}_{ij}'\right\| = \left\|\Id_{\hat{V}^{\perp}} \tila_{i,L_1}\otimes \tila_{i,L_2}\otimes \tila_{j,L_3} \otimes \tila_{j,L_4}\right\|\ge \poly(1/d, \rho)\mcom\nonumber
	\end{align}
	where the first inequality follows from $V\subset \hat{V}$ and second one follows from the fact that $\tila_{i,L_1}\otimes \tila_{j,L_2}\otimes \tila_{i,L_3}\otimes \tila_{j,L_4}$ is orthogonal to the subspace $\hat{V}^{\perp}$.

	Then taking union bound over all $i\neq j$, we obtain that $d(\wtilde{B}')\ge \poly(1/d,\rho)$ occurs with probability $1-\exp(-d^{\Omega(1)})$.
	Therefore $\smin(\wtilde{B}')\ge \poly(1/d,\rho)$ which in turn implies that $\smin(\wtilde{B})\ge \smin(\wtilde{B}')\ge \poly(1/d,\rho)$.

\end{proof}

\section{Tensor decomposition with general components}\label{sec:rounding_general_components}

In this section we prove \pref{thm:general-tensor-decomposition-general-components} (tensor decomposition with general components).
The key ingredient is a scheme for rounding pseudo-distributions (see \pref{thm:general_components} below) that improves over our previous scheme (\pref{thm:reproducing_bks}):
The improved rounding scheme only requires moments of degree logarithmic in the overcompleteness parameter $\sigma$. 

\subsection{Improved rounding of pseudo-distributions}

\newcommand{\kk}{2}

\begin{theorem}\label{thm:general_components}
  Let $s,\sigma \ge 1$ and $\epsilon\in (0,1)$.
  Let $D$ be a degree-$s$ pseudo-distribution over $\R^d$ that satisfies the constraint $\{\|u\|^2 \le 1\}$, and let $a$ be a unit vector in $\R^d$.
  Suppose that $s \ge O(1/\epsilon)\cdot\log(\sigma/\epsilon)$ and,

		\begin{equation}
    \pE_{D(u)}\inner{a,u}^{2s+2}
    \ge \Omega(1/\sigma)\cdot \Norm{\pE\left[uu^{\trans}\right]}
    \ge d^{-O(1)} \mper
  \end{equation}
  \TMnote{Now it is better right? I just replaced degree $O(s)$ to degree-$s$}

  Then, with probability at least $1/d^{O(s^3)}$ over the choice of independent random variables $g_1,\dots, g_{s} \sim \cN(0,\Id_{d^2})$, the top eigenvector $u^{\star}$ of the following matrix satisfies that $\inner{a,u^{\star}}^2 \ge 1-O(\epsilon)$,
	\begin{equation}
  M_g = \pE_{D(u)} \inner{g_1,u^{\otimes \kk}}\cdots \inner{g_s,u^{\otimes \kk}}\cdot uu^{\trans}
	\end{equation}
\end{theorem}

We start by defining some notations for convenience. Let $$p_g(u) = \inner{g_1,u^{\otimes \kk}}\cdots \inner{g_s,u^{\otimes \kk}}\mper$$ Moreover, Let $$\alpha_j = \inner{g_j, a\tp{\kk}}, \quad~~ g_j' = g_j - \alpha_j a\tp{\kk}, ~~~~\textup{and }~~~~\beta_j = \inner{g_j', u}\mper$$ Therefore we have that $\inner{g_j, u\tp{\kk}} =  \inner{\alpha_ja\tp{\kk}, u\tp{\kk}} + \inner{g_j',u\tp{\kk}} = \alpha_j \inner{a,u}^{\kk} + \beta_j$, and it follows that $p_g(u) = \prod_{1\le j\le s} (\alpha_j \inner{a,u}^{\kk} + \beta_j)$.

Theorem~\ref{thm:general_components} follows from the following proposition  and a variant of Wedin's Theorem (see Lemma~\ref{lem:improved_wedin}).

\begin{proposition}In the setting of Theorem~\ref{thm:general_components}, let $\Id_{-1} = \Id - aa^{\trans}$. Then,  with at least $(\Omega(1/n) - 1/d^{O(1)})\cdot 1/d^{O(s^3)}$ probability over randomness of $g$, we have
	\begin{align}
	\max\left\{\Norm{\pE\left[p_g(u)\Id_{-1}uu^{\trans}\Id_{-1}\right]}, 	\Norm{\pE\left[p_g(u)\Id_{-1}uu^{\trans}\Id_{1}\right]}	\right\}	&\le \epsilon  \pE\left[p_g(u)\inner{u,a}^2\right]\mper\label{eqn:91}
	\end{align}
\end{proposition}
\TMnote{added the cross order term in the statement. }
Proposition above follows from the following two propositions, one of which lowerbounds the RHS of~\eqref{eqn:91} and the other upperbounds the LHS of~\eqref{eqn:91}. 

\begin{proposition}\label{prop:good_part}
		In the setting of Theorem~\ref{thm:general_components}, let $\tau = s\sqrt{s\log d}$. Conditioned on the event that $\alpha_1,\dots, \alpha_s \ge \tau$, we have with at least $\Omega(1/n)$ probability over the choice of $g'$
	\begin{align}
		\pE\left[p_g(u)\inner{u,a}^2\right] \ge  0.9 \alpha_1\dots \alpha_s \pE[\inner{a,u}^{2s+2}]\nonumber
		\end{align}
\end{proposition}

\begin{proposition}\label{prop:noise_part}In the setting of Theorem~\ref{thm:general_components}, let $\tau = s\sqrt{s\log d}$ and  $\Id_{-1} = \Id - aa^{\trans}$. Conditioned on the event that $\alpha_1,\dots, \alpha_s \ge \tau$, we have with at least $1-d^{-\Omega(1)}$ probability over the choice of $g'$,  
	\begin{align}
		\max\left\{\Norm{\pE\left[p_g(u)\Id_{-1}uu^{\trans}\Id_{-1}\right]}, 	\Norm{\pE\left[p_g(u)\Id_{-1}uu^{\trans}\Id_{1}\right]}	\right\}\le O(\epsilon\alpha_1\dots \alpha_s)\cdot 	\pE\left[\inner{a,u}^{2s+2}\right]\mper
	\end{align}
\end{proposition}
We first prove Proposition~\ref{prop:good_part}. We need the following three lemmas. 
\begin{lemma}\label{lem:spectralnorm_sos_proof1}
	Let $\epsilon \in (0,1/3)$ and $0\le \delta \le \epsilon$. Suppose $0\le \kappa \le \delta \alpha_j$ for every $j\in [s]$, then there exists a SoS proof
	\begin{equation}
	\vdash_{x} x^2 \prod_{j\in [s]} (\alpha_jx^2+ \kappa) \le  \alpha_1\dots \alpha_s\left((1-\epsilon)^s x^2  +  (1+O(\delta))^s x^{2s+2}\right)\nonumber
	\end{equation}

\end{lemma}

\begin{proof}
	
	Since this is a univariate polynomial inequality, it suffice to show that it's true for every $x$, which will imply that there is also a SoS proof. 
	For $x\in \R$ such that $x^2\le 1-2\epsilon$, we have that 
	\begin{align}
	x^2 \prod_{j\in [s]} (\alpha_jx^2+ \kappa) &\le x^2 \prod_{j\in [s]}  (1-\epsilon)\alpha_j \tag{by $\epsilon\ge \delta$ and $\delta \alpha_j \ge \kappa$} \\
	& \le (1-\epsilon)^s \alpha_1\dots \alpha_sx^2 \nonumber
	\end{align}
	For $x\in \R$ such that $x^2 \ge 1-2\epsilon\ge 1/3$, we have that 
	\begin{align}
	x^2 \prod_{j\in [s]} (\alpha_jx^2+ \kappa) &\le x^2 \prod_{j\in [s]}  \alpha_j (1 + O(\delta))x^2\tag{by $x^2 \ge 1/3$ and $\delta \alpha_j \ge \kappa$} \\
	& \le \alpha_1\dots \alpha_s (1+O(\delta))^s x^{2s+2}\nonumber
	\end{align}
	Hence we obtain a proof for the nonnegativity of the target polynomial. It is known that every nonnegative univariate polynomial admits a sum-of-squares proof. %
	Therefore the inequality above has a sum-of-squares proof. 
\end{proof}

\begin{lemma}\label{lem:general_components_2}
	In the setting of Theorem~\ref{thm:general_components}, let $\tau = s\sqrt{s\log d}$, $\kappa = O(\sqrt{s\log d})$. Conditioned on the event that $\alpha_1,\dots, \alpha_j \ge \tau$, we have %
	\begin{align}
	\pE\left[\inner{a,u}^2\prod_{j\in [s]} (\alpha_j\inner{a,u}^2+\kappa)\right] %
	& \le \alpha_1\dots \alpha_s \cdot O(\pE\left[\inner{a,u}^{2s+2}\right] )\mper \nonumber%
	\end{align}
\end{lemma}

\begin{proof}
	By Lemma~\ref{lem:spectralnorm_sos_proof1}, we have, 
	\begin{align}
	\pE\left[\inner{a,u}^2\prod_{j\in [s]} (\alpha_j\inner{a,u}^2+\kappa)\right] &\le \alpha_1\dots \alpha_s \left((1-\epsilon)^s\pE[\inner{a,u}^2] + (1+O(1/s))^s \pE\left[\inner{a,u}^{2s+2}\right] \right)\nonumber \\
	& \le \alpha_1\dots \alpha_s \cdot O(\pE\left[\inner{a,u}^{2s+2}\right] )\tag{by $(1-\epsilon)^s \le \epsilon/\sigma$ and $\pE\left[\inner{u,a}^{2s+2}\right]\ge \frac{1}{\sigma}\Norm{\pE[uu^{\trans}}$}
	\end{align}
\end{proof}

\begin{proof}[Proof of Proposition~\ref{prop:good_part}]
	We have 
	\begin{align}
	\E_{g'}\left[\pE\left[p_g(u)\inner{u,a}^2\right]\mid \alpha_1,\dots, \alpha_s\right] & = 	\E_{g'}\left[\pE\left[\prod_{1\le j\le s} (\alpha_j \inner{a,u}^{\kk} + \beta_j)\inner{u,a}^2\right]\mid \alpha_1,\dots, \alpha_s\right] \nonumber\\
	& = \pE\left[\prod_{j\in [s]}(\alpha_j\inner{a,u}^2 + \E[\beta_j])\cdot \inner{u,a}^2\right] \tag{by  linearity of pseudo-expectation and independence of $g_1',\dots, g_s'$} \\
	& = \alpha_1\dots \alpha_s \pE[\inner{a,u}^{2s+2}]
	\end{align}
	Moreover, we bound the variance,
	\begin{align}
	 \E_{g'}\left[\left(\pE\left[p_g(u)\inner{u,a}^2\right]\right)^2\mid \alpha_1,\dots, \alpha_s\right]  & \le 		\E_{g'}\left[\pE\left[p_g(u)^2\inner{u,a}^4\right]\mid \alpha_1,\dots, \alpha_s\right]\nonumber \\
	& = \pE\left[\prod_{j\in [s]}\E_{g_j}\left[\alpha_j\inner{a,u}^2 + \beta_j\right]^2\cdot \inner{u,a}^4\right] \nonumber\\ %
	& = \pE\left[\prod_{j\in [s]}(\alpha_j^2\inner{a,u}^4 + 1)\cdot \inner{u,a}^4\right] \nonumber\\
		& \le  \pE\left[\prod_{j\in [s]}(\alpha_j^2\inner{a,u}^2 + 1)\cdot \inner{u,a}^2\right] \tag{by $\inner{u,a}^2\le 1$}\\
	& = \alpha_1\dots \alpha_s \pE[\inner{a,u}^{2s+2}] \tag{By Lemma~\ref{lem:general_components_2}}
	\end{align}
	Therefore, by Paley-Zygmund inequality, we have that with probability 
	\begin{equation}
	\Pr_{g'}\left[\pE\left[p_g(u)\inner{u,a}^2\right] \ge 0.9\alpha_1\dots \alpha_s \pE[\inner{a,u}^{2s+2}]\mid \alpha_1, \dots, \alpha_s\right]\ge \frac{1}{100}\alpha_1\dots \alpha_s \pE[\inner{a,u}^{2s+2}] \ge \Omega(1/n)\mcom\nonumber
	\end{equation}
	which completes the proof. 
\end{proof}
Towards proving Proposition~\ref{prop:noise_part}, we start with the following Lemma. 
\begin{lemma}\label{lem:general_components_matrix_SoS}
	Let $\e>0$, $k\in \N$, $a\in \R^d$ with $\norm{a}=1$ and $\cA=\{\norm{u}^2\le 1\}$.
	Then, there exists a matrix sum-of-squares proof,
	\begin{displaymath}
	\cA \vdash_{u,1/\e} \Paren{\langle  a,u \rangle^{2k} - (1-\e)^k} \cdot \dyad{(\Id_{-1}u)}
	\preceq O(\e) \cdot \langle  a,u \rangle^{2k+2}\cdot \Id\mper
	\end{displaymath}
\end{lemma}

\begin{proof}
	Let $r=1/\e$.
	We may assume $r$ is an integer and that $\epsilon>0$ is small enough such that $(1-\e)^r\ge 1/3$.
	Then, the univariate polynomial inequality $x^{2k}-(1-\epsilon)^k\le 3 x^{2k} \cdot x^{2r}$ holds for all $x\in \R$.
	(For $x^2< 1-\e$, the left-hand side is negative.
	For $x^2 \ge 1-\e$, the right-hand side is at least $x^{2k}$ because $3 x^{2r}\ge x^{2r}/(1-\e)^r\ge 1$.)
	It follows that there exists a sum-of-squares proof
	\begin{equation}
	\label{eq:2}
	\vdash_x x^{2k}-(1-\epsilon)^k\le 3 x^{2k} \cdot x^{2r}.
	\end{equation}
	Similarly, there exists a sum-of-squares proof (see the texts below equation~\eqref{eq:8} as well)
	\begin{equation}
	\label{eq:3}
	\vdash _x x^{2r}(1-x^2)\le O(1/r)\cdot x^2=O(\e)\cdot x^2.
	\end{equation}
	Therefore,
	\begin{align}
	\cA \vdash_{u,1/\e} & \Paren{\langle  a,u \rangle^{2k} - (1-\e)^k} \cdot \dyad{(\Id_{-1}u)} \notag\\
	& \preceq 3 \langle  a,u \rangle^{2k+2r} \cdot \dyad{(\Id_{-1}u)} \tag {by \pref{eq:2}}\\
	& \preceq 3 \langle  a,u \rangle^{2k+2r} \cdot (1-\langle  a,u \rangle^2) \Id \tag {because  $\vdash \dyad v\preceq \norm{v}^2 \Id$ by \pref{lem:matrix-sos}}\\
	& \preceq O(\e) \cdot \langle  a,u \rangle^{2k+2}\cdot \Id\mper \tag{by \pref{eq:3}}
	\end{align}
\end{proof}

\begin{proof}[Proof of Proposition~\ref{prop:noise_part}] We only bound $\Norm{\pE\left[p_g(u)\Id_{-1}uu^{\trans}\Id_{-1}\right]}$. The other term can be controlled similarly and the detailed proof are left to the readers. 
	Let $\alpha_S = \prod_{j\in S}\alpha_j$ and $\beta_S = \prod_{j\in S} \beta_j$. By the fact that $\inner{g_j, u\tp{\kk}} = \alpha_j \inner{a,u}^{\kk} + \beta_j$, we have, %
	\begin{align}
	p_g(u)\Id_{-1}uu^{\trans}\Id_{-1}& = \sum_{S\subset [s], L= S^c}\underbrace{\alpha_S\beta_L\inner{a,u}^{\kk |S|}\Id_{-1}uu^{\trans}\Id_{-1}}_{W_S(u)}\mcom\label{eqn:eqn48}
	\end{align}
	where each summand is denoted by $W_S(u)$. 
	Observe that $W_S$ can be written as 
	\begin{align}
		W_S(u) & =\left(\Id\otimes \Id \otimes \underbrace{g_{j_1}^{'\trans}\otimes \dots \otimes g_{j_r}^{'\trans}}_{\{j_1,\dots, j_r\} = T}\right) \cdot \left( \inner{a,u}^{\kk |S|} (\Id_{-1}u)\otimes (\Id_{-1}u) \otimes u\tp{|L|}\right)\nonumber	\end{align}
	Then by Theorem~\ref{thm:matrix_concentration_multi_mode}, with probability at least $1-2^{-s}d^{-\Omega(1)}$ over the choice of $g_1',\dots,g_s'$ we have , 
	\begin{align}
		\Norm{\pE\left[W_S\right]} & \le  \alpha_S O(s\log d)^{|L|/2}\cdot \max_{J\in [|L|+2]: 1\in J, 2\not\in J}\Norm{\pE\left[\inner{a,u}^{\kk |S|} (\Id_{-1}u)\otimes (\Id_{-1}u) \otimes u\tp{|L|}\right]}_{J,J^c}\nonumber\\
		& \le \alpha_S O(s\log d)^{|L|/2}\cdot \Norm{\pE\left[\inner{a,u}^{\kk |S|} (\Id_{-1}u)\otimes (\Id_{-1}u)\right]}\tag{Lemma~\ref{lem:spectral1} and $\|u\|^2\le 1$.}		%
	\end{align}
	By Lemma~\ref{lem:general_components_matrix_SoS} we have that 
	\begin{displaymath}
	\{\|u\|^2 \le 1\} \vdash_{u,1/\e} \Paren{\langle  a,u \rangle^{2|S|} - (1-\e)^{|S|}} \cdot \dyad{(\Id_{-1}u)}
	\preceq O(\e) \cdot \langle  a,u \rangle^{2|S|+2}\cdot \Id\mper
	\end{displaymath}
	Therefore taking pseudo-expectation, we obtain that 
		\begin{align}
			\pE\left[\inner{a,u}^{\kk |S|} (\Id_{-1}u)\otimes (\Id_{-1}u)\right]\preceq \pE\left[ (1-\epsilon)^{|S|} (\Id_{-1}u)\otimes (\Id_{-1}u)\right] + O(\epsilon)\pE\left[\inner{a,u}^{\kk |S|+2}\Id\right]\label{eqn:eqn47}
		\end{align}

	Then using the fact that
	$$\Norm{\pE\left[ (1-\epsilon)^{|S|} (\Id_{-1}u)\otimes (\Id_{-1}u)\right]} = \Norm{\Id_{-1}\pE\left[ (1-\epsilon)^{|S|} uu^{\trans}\right]\Id_{-1}} \le (1-\epsilon)^{|S|} \Norm{\pE\left[ uu^{\trans}\right]}\mcom$$
	and equation~\eqref{eqn:eqn47}, we have
	\begin{align}
	\Norm{\pE\left[W_S\right]}
	& \le \alpha_S O(s\log d)^{|L|/2}\cdot  \left((1-\epsilon)^{|S|} \Norm{\pE\left[ uu^{\trans}\right]} + O(\epsilon)\pE\left[\inner{a,u}^{\kk |S|+2}\Id\right]\right) \label{eqn:eqn49} %
	\end{align}
	
	Taking union bound over all subset $S$, with probability at least $1-d^{-\Omega(1)}$, we have equation~\eqref{eqn:eqn49} holds for every $S\subset [s]$. Taking the sum of equation~\eqref{eqn:eqn49} over $S$, we conclude that %
	\begin{align}
	& \Norm{\pE\left[p_g(u)\Id_{-1}uu^{\trans}\Id_{-1}\right]}\le  \sum_S\Norm{\pE\left[W_S\right]} \tag{by equation~\eqref{eqn:eqn48}}\\ 
	& \le\Norm{\pE\left[ uu^{\trans}\right]} \cdot  \prod_{j\in [s]} ((1-\epsilon)\alpha_j + O(\sqrt{s\log d})) + O(\epsilon)\pE\left[\inner{a,u}^2\prod_{j\in [s]} (\alpha_j\inner{a,u}^2+O(\sqrt{s\log d}))\right] \label{eqn:eqn51}%
	\end{align}
	
	When $\alpha_j \ge \tau$, where $\tau = s\sqrt{s\log d}$ and $s = \frac{c}{\epsilon}\log (\sigma/\epsilon)$  where $c$ is a sufficiently large absolute constant, using the fact that $\Norm{\pE\left[ uu^{\trans}\right]} \le \sigma/n$, we have $$\Norm{\pE\left[ uu^{\trans}\right]} \prod_{j\in [s]} ((1-\epsilon)\alpha_j + O(\sqrt{s\log d}))  \le (1-\epsilon/2)^s\alpha_1\dots \alpha_s\le \frac{\epsilon}{n}\alpha_1\dots \alpha_s\mper$$ Regarding the second term on the RHS of~\eqref{eqn:eqn51}, by Lemma~\ref{lem:general_components_2}, we have that when $\alpha_j \ge \tau$, 
	\begin{align}
	\pE\left[\inner{a,u}^2\prod_{j\in [s]} (\alpha_j\inner{a,u}^2+O(\sqrt{s\log d}))\right] %
	& \le \alpha_1\dots \alpha_s \cdot O(\pE\left[\inner{a,u}^{2s+2}\right] ) \nonumber%
	\end{align}
	Therefore, plugging in the two bounds above into equation~\eqref{eqn:eqn51}, we obtain that 
	\begin{align}
	\Norm{\pE\left[p_g(u)\Id_{-1}uu^{\trans}\Id_{-1}\right]}\le O(\epsilon\alpha_1\dots \alpha_s)\cdot 	\pE\left[\inner{a,u}^{2s+2}\right]\mper\nonumber\end{align}

\end{proof}

\subsection{Finding all components}

In this section we prove Theorem~\ref{thm:general-tensor-decomposition-general-components} (restated below) using iteratively the rounding scheme that is developed in the subsection before. 
\restatetheorem{thm:general-tensor-decomposition-general-components}

\begin{algorithm}
	\caption{Tensor decomposition with general components}
	\label{alg:tensor-general-general-components}
	
	\textbf{Parameters:}
	numbers $\e>0,\, n\in \N$.
	
	\textbf{Given: } $2k$-th order tensor $T$

	\textbf{Find:}
	Set of vectors $S =\Set{\hat{a}_1,\ldots,\hat{a}_{n'}}\subset \R^d$ with $n'\le n$.	
	
	\begin{itemize}
		\item Let $s = k-1$, $\ell = O(s)$, and $\eta = O(\epsilon)^{1/2}$. 
			\begin{equation}
			\cA =
			\Set{\|u\|^2 = 1}	\cup 
			\Set{\vbig
				\inner{T, u^{2s+2}} \ge 2/3 
			}\mper
			\end{equation}
		\item For $i$ from $1$ to $n$, do the following:
		\begin{enumerate}
			\item Compute a $\ell$-degree pseudo-distribution $D(u)$ over $\R^d$ that
			satisfies the constraints
			\begin{align}
				& \cA %
				\textup{ and } \Norm{\pE\nolimits_{D(u)} \dyad{u}} \le \frac{\sigma}{n-i+1}\mper
				\label{eq:47}
			\end{align}
			\item Repeat $T =  d^{O(s^3)}$ rounds of the following: 
			\begin{itemize}
				\item Choose standard Gaussian vectors $g_1,\ldots,g_s\sim \cN(0,\Id_{d^2})$ and compute the top eigenvectors $a^{\star}$ of the following matrix, 
				\begin{equation}
				\pE_{D(u)} \inner{g_s,u^{\otimes 2}}\dots \inner{g_s,u^{\otimes 2}}\cdot \dyad{u} \in \R^{d\times d}\mper
				\end{equation}
				\item Check if $a^{\star}$ satisfies $\cA$. If yes, let $\hat{a}_i = \hat{a}^\star$ and $S\leftarrow S\cup \{\hat{a}_i\}$, add to $\cA$ the constraint $\{\langle  u, \hat{a}_i \rangle^2 \le 1-5\eta\}$, and break the (inner) loop. %
			\end{itemize}
			\item If no new $\hat{a}_i$ is found in the previous step, stop the algorithm. 
		\end{enumerate}
		
	\end{itemize}
\end{algorithm}

\begin{proof}[Proof of \pref{thm:general-tensor-decomposition-general-components}]
	We analyze \pref{alg:tensor-general-general-components}.
	Let $\eta = c_0\eps^{1/2}$ where $c_0$ is a large enough absolute constant. Let $\mathcal{A}'$ be the constraint that $\sum_i {a_i,u}^{2s+2}\ge 1/3$. Then we have $\mathcal{A}\vdash \mathcal{A}'$. 
	We first observe that as long as a vector $a$ satisfies $\cA$, then $a$ has to be $O(\epsilon^{1/2})$-close to one of the $a_i$'s up to sign flip. This is because 
	\begin{eqnarray}
	1/3 \le \inner{T,u^{\otimes 2s+2}}-1/3\le \inner{u,a_i}^{2s+2}\le \max_{i} \inner{a_i,u}^{2s} \left(\sum_i \inner{a_i,u}^2\right)\le \sigma \max_{i} \inner{a_i,u}^{2s}\mper \nonumber
	\end{eqnarray}
	That is, we can always check whether $a^{\star}$ is what we wanted as in the second bullet of step 2. 
	Therefore, it remains to show that as long as there exists $a_j$ that is $\eta^{1/2}$-far away (up to sign flip) to the set $S$, we will find a new vector after Step 2 in the next iteration. %

	We assume that after iteration $i_0$,  %
the set $W = \{a_j: \forall i\in [i_0], \inner{a_j, \hat{a}_{i}}^2 \le 1-\eta\}$ is not empty. We will show that after iteration $i_0+1$, we will find a new vector in $W$ up to $O(\epsilon)^{1/2}$ error. %
	We claim first that in the $i_0+1$ iteration there exists a pseudo-distribution $D(u)$ that satisfies ~\eqref{eq:47}. Indeed, this is because the actual uniform distribution over the finite set $W$ satisfies constraint~\eqref{eq:47}. Here we used the fact that for every $j\in [n]$ we have $\inner{T, a_j^{\otimes 2k}}\ge \inner{\sum_{i} a_i^{\otimes 2k}, a_j^{\otimes 2k}} - 1/3 \ge \inner{a_j,a_j}^{k} - 1/3 = 2/3$. 
	
	Since constraints~\eqref{eq:47} enforce that for every $i\le i_0$ pseudo-distribution $D(u)$ satisfies that $\inner{u,\hat{a}_i}^2 \le 1- \eta$, and moreover, we have $\|a_i-\tau_i\hat{a}_i\|^2 \le O(\epsilon)$ for some $\tau_i\in \{-1,+1\}$, by Lemma~\ref{lem:helper_lemma1}, we conclude that $D(u)$ also satisfies the constraint that $\inner{u,a_i}^2\le 1-\eta/2$ (here we use the fact that $\eta = c_0\epsilon$ with large enough constant $c_0$). These implies that $D(u)$ satisfies that $\sum_{i=1}^{i_0} \inner{u,a_i}^{2s+2}\le (1-\eta)^{2s}\sum_{i=1}^{i_0}\inner{u,a_i}^2$. Therefore, we have $\pE\left[\sum_{i=1}^{i_0} \inner{u,a_i}^{2s+2}\right]\le (1-\eta)^{2s}\pE\left[ \inner{u,a_i}^2\right]\le \sigma^{-1}\cdot \sigma/(n-i_0+1)\le 1/3$. Thus by constraint~\eqref{eq:47} we have $\pE\left[\sum_{i > i_0} \inner{u,a_i}^{2s+2}\right]\ge 1/3$. Therefore, there exists $i^{\star} > i_0$ such that $\pE\left[\inner{u,a_{i^{\star}}}^{2s+2}\right]\ge \frac{1}{3(n-i_0+1)}$. Then by Theorem~\ref{thm:general_components} we obtain that with $1/d^{O(s^3)}$ probability, in each step of the inner loop we can find $\hat{a}_i$ that is $O(\epsilon^{1/2})$-close to $a_{i^{\star}}$, and therefore at the end of the inner loop with high probability we found a new vector $\hat{a}_{i_0+1}$ which is close $O(\epsilon^{1/2})$-close to $a_{i^{\star}}$.

\end{proof}

\section{Fast orthogonal tensor decomposition without sum-of-squares}

\label{sec:nosos-orthogonal}

In this section, we give an algorithm (see ~\pref{thm:orthogonal-alg}) with quasi-linear running time (in the size of the input) that finds a component of an orthogonal $3$-tensor in the presence of spectral norm error at most $1/\log d$ . 
The previous best known algorithm for orthogonal $3$-tensor is by~\cite[Theorem 5.1]{DBLP:journals/jmlr/AnandkumarGHKT14} which takes similar runtime and tolerates $1/d$ error in injective norm. It is known that for any symmetric tensor $E$ the spectral norm can be bounded by injective norm with multiplicative factor $\sqrt{d}$, that is, $\norm{E}_{\{1\}\{2,3\}}\le \sqrt{d}\cdot \norm{E}_{\{1\}\{2\}\{3\}}$. Therefore, our robustness guarantee is at least $\sqrt{d}$ factor better than tensor power method. 

The key step of Algorithm is the following simple Theorem that finds a single component. It is in fact an analog of  Theorem~\ref{thm:reproducing_bks} without sum-of-squares. Here we analyze the success probability much more carefully for achieving quasi-linear time. 

\begin{theorem}
\label{thm:orthogonal-spectral}
  Let $a_1,\ldots,a_n\in \R^d$ be orthonormal vectors.
  Let $T\in (\R^d)^{\otimes 3}$ be a symmetric $3$-tensor such that $\lVert T-\sum_i a_i ^{\otimes 3}\rVert_{\{1\},\{2,3\}} \le \tau$.
  Let $g$ be a standard $d$-dimensional Gaussian vector.
  Let $\delta\in [0,1]$.
  Then, with probability $1/(d^{1+\delta}(\log d)^{O(1)})$ over the choice of $g$, the top eigenvector of the following matrix is $O(\tau/\delta)$-close to $a_1$,
  \begin{displaymath}
    M_g := (\Id \otimes \Id \otimes \transpose g) T\,.
  \end{displaymath}
  At the same time, the ratio between the top eigenvalue and the second largest eigenvalue in absolute value is at least $1+\delta/3-O(\tau)$.
\end{theorem}

\begin{proof}
  Let $E=T-\sum_i a_i^{\otimes 3}$.
  Then,
  \begin{equation}
    M_g = (\Id \otimes \Id \otimes \transpose {g})E + \sum _{i=1}^n\langle  g,a_i \rangle \cdot a_i^{\otimes 2}\mcom
  \end{equation}
  Since $E$ is symmetric and $\Norm{E}_{\{1\},\{2,3\}}\le \tau$, \pref{thm:matrix_concentration} implies that with probability at least $1-1/d^2$ over the choice of $g$, %
  \begin{equation}
    \label{eq:10}
    \Norm{\vbig (\Id \otimes \Id \otimes \transpose {g}) E }_{\{1\},\{2\}} \le 2(\log d)^{1/2} \tau\mper
  \end{equation}
  Let $t=\sqrt{ 2 \log d}$.
  By the fact that $\langle  g,a_1 \rangle,\ldots,\langle  g,a_n \rangle$ are independent standard Gaussian variables and standard estimates on their cumulative density function, the following event happens with probability at least  $\frac 1{d^{(1+\delta)} \cdot (\log d)^{O(1)}}$
  \begin{equation}
    \label{eq:11}
\iprod{g,a_1} \ge (1+\delta/3)\cdot t ~~\textup{ and }~~\max_{i\in \{2,\ldots,n\}}\abs{\iprod{g,a_i}} \le t
  \end{equation}
  Conditioned on the events in \pref{eq:10} and \pref{eq:11}, we have the following bound on the spectral norms of $M_g$ and $M_g - \delta/3 \cdot  t \cdot a_1^{\otimes 2}$, which implies that the top eigenvector of $M_g$ is $O(\tau / \delta)$-close to $a_1$ (by~\cite[Lemma A.1]{HopkinsSSS16}),
  \begin{align}
    & \Norm{M_g}_{\{1\},\{2\}} - \Norm{M_g - \frac 1 3\delta t \cdot a_1^{\otimes 2}}_{\{1\},\{2\}} \\
    & \ge \Norm{ \sum\nolimits_{i=1}^n \langle  g,a_i \rangle\cdot \dyad {a_i}} - \Norm{ (\langle  g,a_1 \rangle-\frac 1 3\delta t )\cdot \dyad {a_1} + \sum\nolimits_{i=2}^n \langle  g,a_i \rangle\cdot \dyad {a_i}} - 2 (\log n)^{1/2} \tau
      \tag{conditioned on event in \pref{eq:10}}\\
    & \ge \frac 1 3\delta t - 2(\log d)^{1/2} \tau
      \tag{conditioned on event in \pref{eq:11}}\\
    & \ge (1 - O(\tau/\delta)) \cdot \frac 1 3 \delta t\mper
  \end{align}
  The probability that the events of \pref{eq:10} and \pref{eq:11} happen simultaneously is at least
  \begin{displaymath}
    \frac 1 {d^{1+\delta}(\log d)^{O(1)}}-\frac 1{d^2} \ge  \frac 1 {d^{1+\delta}(\log d)^{O(1)}}\mper
  \end{displaymath}
  This bound implies the first part of the theorem. To see the eigengap bound, we first observe that the largest eigenvalue of $M_g$ is at least $\inner{g,a_1} - \Norm{\vbig (\Id \otimes \Id \otimes \transpose {g}) E }_{\{1\},\{2\}} \ge (1+\delta/3)t - 2(\log d)^{1/2} \tau$. On the other hand, by eigenvalue interlacing, the second larges eigenvalue of $M_g$ is bounded by the top eigenvalue of $M_g - \inner{g,a_1}a_1a_1^{\trans} = (\Id \otimes \Id \otimes \transpose {g}) E + \sum _{i=2}^n\langle  g,a_i \rangle \cdot a_i^{\otimes 2}$, which in turn is bounded above by $2(\log d)^{1/2} \tau + t$. Therefore the eigenvalue gap statement follows by recalling $t = \sqrt{2\log d}$. 
\end{proof}

We remark that we can amplify the success probability of the algorithm by running it repeatedly with independent randomness.

\begin{theorem}
\label{thm:orthogonal-alg}
  There exists a randomized algorithm  with running time $d^3\cdot (\log d)^{O(1)}$ that given a symmetric $3$-tensor $T\in (\R^d)^{\otimes 3}$ such that $\norm{T-\sum_{i=1}^n \cramped{a_i^{\otimes 3}}}_{\{1\},\{2,3\}}\le 1/\log d$ for some set of orthonormal vector $\{a_1,\ldots,a_n\}\subseteq \R^d$ outputs with probability $\Omega(1)$ a vector unit $v$ such that
  \begin{displaymath}
    \min_{i\in[n]} \Norm{v-a_i}^2 \le \frac 1 {2^{d}} + \Norm{T-\sum\nolimits_{i=1}^n a_i^{\otimes 3}}_{\{1\},\{2\},\{3\}}\mper
  \end{displaymath}
  Furthermore, there exists a randomized algorithm with running time $d^{1+\omega}\cdot (\log d)^{O(1)}\le O(d^{3.33})$ that given $T$ as before, with probability at least $\Omega(1)$ outputs a set of vectors $\{a_1',\ldots,a_n'\}$ with Hausdorff distance at most $\frac 1 {2^{d}} + \Norm{T-\sum\nolimits_{i=1}^n a_i^{\otimes 3}}_{\{1\},\{2\},\{3\}}$ from $\{a_1,\ldots,a_n\}$.
  Here $\omega$ is the matrix multiplication exponent.
\end{theorem}

\Dnote{} \TMnote{I guess we could be a bit lazy about this and just leave it as it is.. Hopefully people who have bear with us for the first 9 sections should be fine with this prof... }

\begin{proof}
  We may assume that $d$ is larger than some constant. 
 We run $d^{1+\omega}\cdot O(\log d)^{O(1)}$ iterations of the following procedure which can be carried out in $\tilde{O}(d^3)$ time. We will discuss how to speed the algorithm at the end. 
  	\begin{enumerate}
  		\item Choose a standard Gaussian vector $g$ and compute $M_g=(\Id \otimes \Id \otimes \transpose g) T$.
  		\item Run $O(\log d)^2$ iterations of the matrix power method on $M_g$ viewed as a $d$-by-$d$ matrix (with random initialization) and set $u$ to be the top
  		eigenvector calculated in this way.
  		\item Check that $\abs{\iprod{u^{\otimes 3},T}} \ge 0.9$.
  		\item Run $O(\log \log d)$ iterations of the tensor power method on $T$ starting from $u$. Output the final iterate $v$ of the method.
  	\end{enumerate}
  The analysis of the tensor power method \cite[Lemma 5.1]{DBLP:journals/jmlr/AnandkumarGHKT14} shows that whenever the check in step 3 succeeds then the final output $v$ satisfies the desired accuracy guarantee of the theorem.
  It remains to show that the check in step 3 succeeds with probability at least $1/(\log d)^{O(1)}$ over the randomness of the algorithm.
  (We obtain success probability $\Omega(1)$ by repeating the algorithm $(\log d)^{O(1)}$ times.)
  We apply \pref{thm:orthogonal-spectral} for $\delta=O(1/\log d)$ and $\tau=1/\log d$ such that for every $i\in[n]$, the distance guarantee for the top eigenvector of $M_g$ is at most $0.001$ and the ratio between first and second eigenvalue is at least $1+1/\log d$. %
  By symmetry, for every index $i\in[n]$, the probability that the top eigenvector of $M_{g}$ is $0.001$-close to $a_i$ is at least $1/d(\log d)^{O(1)}$.
  Since the vectors $a_1,\ldots,a_n$ are orthonormal these events are disjoint.
  Therefore, with probability at least $1/(\log d)^{O(1)}$ over the choice of $g$, the top eigenvector of $M_g$ is $0.001$-close to one of the vectors $a_1,\ldots,a_n$.
  Since the multiplicative gap between the top eigenvalue and the remaining eigenvalues of $M_g$ is at least $1+1/\log d$ (by \pref{thm:orthogonal-spectral} for our choice of $\delta$ and $\tau$), it follows that with constant probability over the choice of the random initialization of the matrix power method, the second step of the algorithm recovers a vector that is $0.001$-close to the top eigenvector of $M_g$.
  In this case, the resulting vector $u$ satisfies the check $\abs{\iprod{u^{\otimes 3}, T}}\ge 0.9$.

  In order to find all components in time $d^{1+\omega}\cdot O(\log d)^{O(1)}$ we run $d\cdot (\log d)^{O(1)}$ independent evaluations of the above algorithm.
 Note that each run involves multiplication of a $d^2\times d$ matrix with a $d$ dimensional vector and therefore in total we are to multiply a $d^2\times d$ matrix with $d\times d$ matrix. Therefore, 
  using fast matrix multiplication, we can ``parallelize'' all of the required linear algebra operations and speedup the running time from $d^4 (\log d)^{O(1)}$ to the desired $O(d^{1+\omega})\cdot (\log d)^{O(1)}$.
  \TMnote{Hopefully this explanation is enough}
\end{proof}

\begin{remark}[Extension to other settings]\label{remark:extension}
	The same rounding idea in~\pref{thm:orthogonal-spectral} can be extended to the setting when the components $a_1,\dots, a_n$ are close to isotropic in the sense that $\Norm{\sum_i a_i\transpose{a_i} - \Id_d} \le\sigma$. The success probability will decrease to roughly $1/d^{1+\poly(\sigma)}$, and therefore when $\sigma$ is at most a constant, the overall runtime will remain polynomial in $d$. 
	
	Suppose $a_1,\dots, a_n$ are separate vectors as in the setting of~\pref{thm:separated-unit-vectors}, we can apply the idea in paragraph above to the 3-tensor $\sum_i b_i^{\otimes 3}$ where $b_i = a_i^{\otimes k/3}$ and $k \ge O\left(\frac{1+\log \sigma}{\log \rho}\right) \cdot \log(1/\eta)$ is a multiple of 3. By \pref{lem:separated-spectral-norm} and the condition on $k$, we have that $b_i$ are in nearly isotropic position with $\Norm{\sum_i b_i\transpose{b_i}}\le 1+\eta$. Hence, using idea above we have a spectral algorithm without sum-of-squares for this setting. As noted before (below~\pref{thm:separated-unit-vectors}), the error tolerance of this algorithm is in terms of an \emph{unbalanced} spectral norm: $\norm{T-\sum_{i=1}^n a_i^{\otimes 2k}}_{\{1,\ldots,2k/3\},\{2k/3+1,\ldots,2k\}}$, which limits its application, for example, to dictionary learning.  

\end{remark}

\addreferencesection
\bibliographystyle{amsalpha}
\bibliography{bib/mr,bib/dblp,bib/ads,bib/scholar,bib/zblatt,bib/custom}

\appendix
\section{Toolbox}

\begin{lemma}[sums-of-squares proof for Cauchy-Schwarz inequality]
	\label{lem:cauchy-schwarz}
	Let $x_1, \dots, x_n$ and $y_1, \dots, y_n$ be polynomials in some indeterminates. Then
	\[ \vdash\, \left\{ \left(\sum_{i=1}^n x_iy_i\right)^2 \le \left(\sum_{i=1}^n x_i^2\right)\left(\sum_{i=1}^n y_i^2\right) \right\} \mper \]
\end{lemma}
\begin{proof}
	The difference between the RHS and the LHS is a sum of squares.
	\[ \left(\sum_{i=1}^n x_i^2\right)\left(\sum_{i=1}^n y_i^2\right) - \left(\sum_{i=1}^n x_iy_i\right)^2
	= \sum_{i,j} (x_iy_j - x_jy_i)^2 \mper\qedhere \]
\end{proof}

 \begin{lemma}\label{lem:helper_lemma1}
 	Let $u$ be indeterminate and $a$ be a unit vector.  Let $\cA = \{\|u\|^2 =1, \inner{u,a}^2\le \tau\}$. Then for any unit vector $b$ such that $\|a-b\|\le 2\delta$, we have that 
 	\[\cA\vdash \{\inner{u,b}^2 \le \left(\sqrt{\tau}+\sqrt{\delta}\right)^2\}\mper\]
 \end{lemma}
 \begin{proof}
 	First of all, by Lemma~\ref{lem:square-root}, $\cA \vdash \inner{u,a} \le \sqrt{\tau}$ and $\cA \vdash \inner{u,a} \ge -\sqrt{\tau}$. To bound $\inner{u,b}$, we
 	decompose $u$ and $b$ into their components parallel to and perpendicular to $a$.
 	
 	Let $u' = u - \inner{u,a}a$, so that $u'$ is a vector polynomial in $u$ that satisfies $u = \inner{u,a}a + u'$ and $\inner{u',a} = 0$.
 	Then $\|u\|^2 = \|u'\|^2 + \inner{u,a}^2 \|a\|^2$ and therefore $\cA \vdash \|u'\|^2 \le 1$.

 	Similarly, let $b' = b - \inner{b,a}a$, so that $b = \inner{b,a}a + b'$ and $\inner{b',a} = 0$.
 	Since $\|b-a\|^2 \le 2\delta$, we have $\inner{b,a} \ge 1 - \delta$,
 	meaning that $\|b'\|^2 = \|b\|^2 - \inner{b,a}\|a\|^2 \le 1 - (1 - \delta) = \delta$.

 	Then we are ready to bound $\inner{u,b}$:
 	\begin{align*}
 		\inner{u,b} &= \Big\langle\inner{u,a}a + u',\, \inner{b,a}a + b'\Big\rangle \\
 		&= \inner{u,a}\inner{b,a} + \inner{u',b'} \mper
 		\end{align*}
 	Since $\cA \vdash \inner{u,a}\inner{b,a} \le \sqrt{\tau}$ and also $\inner{u',b}^2 \le \|u'\|^2 \|b'\|^2 \le \delta\|u'\|^2$ which in turn implies (by $\cA \vdash \|u'\|^2 \le 1$ and Lemma~\ref{lem:square-root}) that $\cA \vdash \inner{u',b'}\le \sqrt{\delta}$, we conclude that $\cA \vdash \inner{u,b} \le \sqrt{\tau} + \sqrt{\delta}$.

 	Similarly, $\cA \vdash \inner{u,b}\ge -\sqrt{\tau} -\sqrt{\delta}$.
 	Hence
 	\[ \inner{u,b}^2 - \left(\sqrt{\tau}+\sqrt{\delta}\right)^2 = \left(\inner{u,b}-
 	\left(\sqrt{\tau}+\sqrt{\delta}\right)\right)\left(\inner{u,b}+
 	\left(\sqrt{\tau}+\sqrt{\delta}\right)\right)\le 0\mcom\]
 	as desired.
 \end{proof}

\begin{lemma}\label{lem:square-root}
	For a positive real number $a$, and $x$ be an indeterminate, then we have that 
	\begin{equation}
	\Set{x^2\le a^2} \vdash \Set{x\le a, x\ge -a}
	\end{equation}
\end{lemma}

\begin{proof}
  The first statement simply follows from the following two polynomial identities, 
  \begin{displaymath}
    a-x = \tfrac 1{2a} \Paren{a^2-x^2 +  (a-x)^2}\mcom
  \end{displaymath}
  and similarly,
  \begin{displaymath}
  x+a = \tfrac 1{2a} \Paren{a^2-x^2 +  (a+x)^2}\mper
  \end{displaymath}

\end{proof}

\begin{theorem}[Consequence of Davis-Kahan Theorem~\cite{doi:10.1137/0707001}. c.f~\cite{yu2015useful}]\label{thn:davis-kahan-rank-1}
	Let $\Sigma, \hat{\Sigma}$ be symmetric matrices in $\R^{d\times d}$. Let $v_1,\tilde{v}_1$ be their top eigenvector respectively and let $\lambda_1\ge \lambda_2\dots$ and $\hat{\lambda}_1\ge \hat{\lambda}_2\dots $ be their eigenvalues, respectively. Then, 
	\begin{align}
		\Norm{v_1-\hat{v}_1}\le \frac{\sqrt{2}\norm{\Sigma-\hat{\Sigma}}}{|\lambda_1-\hat{\lambda}_2|} \mper\nonumber
	\end{align} 
\end{theorem}

\begin{lemma}[Consequence of Theorem~\ref{thn:davis-kahan-rank-1}]\label{lem:improved_wedin}
	Let $a$ be unit vector and $\Id_1 = aa^{\top}$, $\Id_{-1} = \Id - aa^{\top}$. Suppose symmetric matrix $M$ satisfies that 
	\begin{align}
	\max\Set{\Norm{\Id_{-1}M\Id_{-1}}, \Norm{\Id_1M\Id_{-1}}}\le \epsilon a^{\top}Ma \nonumber
	\end{align}
	Then, $a$ is $3\sqrt{2}\epsilon$-close to the top eigenvector of $M$ in Euclidean distance. 
\end{lemma}

\begin{proof}
	Let $t = a^{\top}Ma$ and $\widehat{M} = taa^{\top} = \Id_1M\Id_1$. Then we have that $M   =  (\Id_1+\Id_{-1})M(\Id_1+\Id_{-1}) = \widehat{M} + \Id_1M\Id_{-1} + \Id_{-1}M\Id_1 + \Id_{-1}M\Id_{-1} $. Therefore by the assumption we have $\norm{\widehat{M} - M}\le 3\epsilon t$. Therefore using Theorem~\ref{thn:davis-kahan-rank-1} with $\Sigma = M$ and $\hat{\Sigma}= \widehat{M}$ we obtain that the top eigenvector of $M$ is $3\sqrt{2}\epsilon$-close to $M$ in Euclidean distance. 
\end{proof}

\section{Missing proofs in Section~\ref{sec:preliminaries}}
\label{sec:prelim_appendix}
\label{app:prelim_appendix}

\begin{proof}[Proof of Theorem~\ref{thm:inequality_sos}]
	Suppose $\cA=\{f_1\ge 0,\ldots,f_m\ge 0\}$ with $f_1,\ldots,f_m\in \R[x]$ and $\deg(f)\le d$. Let $m = |\cA|$ and let $\ell = \max_{i \in [m]} \deg(f_i)$. We use an optimization algorithm to find such pseudo-distribution $D$. 
	Here the variables are all the moments $\pE_{D(x)}\left[\prod_{i\in S} x_i\right]$ for all $S$ with $|S|\le d$.  The constraints are linear constraints over these variables
	
	\begin{displaymath}
	\Set{\pE_D \left(\prod_{i\in S}f_i\right) h^2 \ge 0 \Mid
		S\subseteq [m],\,
		h\in \R[x],\,
		|S|\ell + \deg\left(h^2\right) \le d\}}\mper
	\end{displaymath}
	We can separate over these constraints in time $(n+m)^{O(d)}$.
	Indeed, for every fixed choice of $S$, the set of constraints of the form $\pE_D \left(\prod_{i\in S}f_i\right) h^2 \ge 0$ may be written as a single matrix constraint $\pE_D  \left(\prod_{i\in S}f_i\right) \dyad{\left[(1,x)^{\otimes d-|S|\ell}\right]} \succeq 0$, with the equivalence established by mapping $h$ to a vector of coefficients.
	Therefore, by \pref{thm:sos} and the equivalence of optimization and separation \cite{DBLP:journals/combinatorica/GrotschelLS81}, we can find moments $\pE_{D(x)}(1,x)^{\otimes d}$ of a degree-$d$ pseudo-distribution in time $(n+m)^{O(d)}$.
	A standard multivariate polynomial interpolation argument allows us to recover the underlying pseudo-distribution $D$ \cite{zbMATH06125965}.
\end{proof}

\begin{proof}[Proof of Lemma~\ref{lem:soundness_sos}]
	Suppose $D$ is a degree-$d$ pseudo-distribution.
	Let $\cA = \{f_1 \ge 0, \dots, f_n \ge 0\}$ and let $\cB = \{g_1 \ge 0, \dots, g_m \ge 0\}$.
	Moreover, $\cA \vdash_{\ell'} \cB$ means that for every constraint $\{g_j \ge 0\}$ in $\cB$, there are sums-of-squares polynomials $p_{j,S}$ for each $S \subseteq [n]$ such that $g_j = \sum_{S \subseteq [n]} p_{j,S}\prod_{i \in S}f_i$ where each summand $p_{j,S}\prod_{i \in S}f_i$ has degree at most $\ell'$.
	
	Consider some set $T \subset [m]$ and some sum-of-squares polynomial $h'$ such that $|T|\ell\ell' + \deg h' \le d$.
	We would like to show that \begin{equation}\pE_D \left(\prod_{j \in T}g_j\right)h' \ge 0 \label{eqn:76}\mper\end{equation}
	$\cA \vdash_{\ell'} \cB$ means that for every constraint $\{g_j \ge 0\}$ in $\cB$, there are sums-of-squares polynomials $p_{j,S}$ for each $S \subseteq [n]$ such that $g_j = \sum_{S \subseteq [n]} p_{j,S}\prod_{i \in S}f_i$ where each summand $p_{j,S}\prod_{i \in S}f_i$ has degree at most $\ell'$. Substituting $g_j$ in equation~\eqref{eqn:76}, it suffices to show that 
	\begin{equation}
	\pE_D \left(\prod_{j \in T}\left(\sum_{S \subseteq [n]} p_{j,S}\prod_{i \in S}f_i\right)\right)h' \ge 0 \mper \label{eqn:77}
	\end{equation}
	We expand the outer product over $T$ and see that the polynomial inside the pseudo-expectation is in fact a sum of many polynomials of the form $q_1\dots q_{|T|} \left(\prod_{i\in W} f_i\right)h'$, where each of the $q_i$ is equal to $p_{j,S}$ for some $j\in T$ and some $S\subset[n]$, and where $W\subset S$ is a multi-set, with $\deg\left(q_1\dots q_{|T|} \left(\prod_{i\in W} f_i\right)\right) \le |T|\ell$.
	Moreover, we note that since each $q_i$ is a sum of squares, $q_1\dots q_{|T|} \left(\prod_{i\in W} f_i\right)$ can be written as $q \prod_{i\in W'} f_i$ where $q$ is a sum of squares and $W'$ is the set of elements that appear in $W$ an odd number of times.
	We calculate
	\begin{align*}
	  \deg(q)
	  &= \deg(q_1\dots q_{|T|}) + \deg\bigg(\prod_{i \in W \setminus W'} f_i\bigg)
	\\&\le (|T|\ell - |W|) + (|W| - |W'|)\ell'
	\\&\le |T|\ell\ell' - |W'|\ell' \mcom
	\end{align*}
	where we used $\deg\left(q_1\dots q_{|T|} \left(\prod_{i\in W} f_i\right)\right)\le |T|\ell$ in combination with $\deg\left(\prod_{i \in W} f_i\right) \ge |W|$, along with the fact that therefore $|W| \le |T|\ell$.
	Therefore since $qh'$ is a sum of squares and $|W'|\ell' + \deg(qh') \le d$, by the definition of $D \models_{\ell} \cA$, we have $\pE\left[q_1\dots q_{|T|} \left(\prod_{i\in W} f_i\right)h'\right] = \pE\left[qh'\left(\prod_{i\in W'} f_i\right)\right]\ge 0$. Then by linearity of pseudo-expectation we prove equation~\eqref{eqn:77}, which completes the proof. 
\end{proof}

\begin{proof}[Proof of Lemma~\ref{lem:completeness_sos}]
	We prove the contrapositive.
	Let $\cA = \{ f_1\ge 0,\ldots,f_m \ge 0\}$.
   Assume that $\cA\not \vdash_{d} \{g\ge -\epsilon\}$ for some $\epsilon > 0$. %
	
	A polynomial $h$ satisfies $\cA \vdash_{d} \{h \ge 0\}$ precisely when $h = \sum_{S \subset [m]} p_S \prod_{i \in S} f_i$ for some sum-of-squares polynomials $p_S$ where the degree of each summand is at most $d$.
	We observe that $\cH= \{h \mid \cA \vdash_{d} \{h \ge 0\}\}$ is a  convex cone. Let $\bar{\cH}$ be its closure. We argue that $g\not\in \bar{\cH}$. 
	
	Indeed, if there exists a sequence of polynomial $g_k$ that converges to $g$ (in coefficients), then there exists a sufficiently large $K$ such that for $k\in K$, $\{\|x\|^2 \le B\}\vdash g_k(x)-g(x)\le \epsilon/2$.  Therefore $\cA \vdash g + \epsilon = g_k + (g-g_k+\epsilon) \ge 0$. This contradicts our assumption. %
	Then by the hyperplane separation theorem, there exists a linear functional $L$ over the space of all degree-$d$ polynomials such that $L [g] < 0$ and $L \left[h\right] >   0$ for all $h\in \cH$. %
	Since $1\in \cH$, we have $L(1) \ge 0$. We can scale $L$ properly so that $L(1) = 1$ and therefore $L$ defines a pseudo-distribution $D$. 
	In particular, $D$ is a pseudo-distribution such that $D\models_\ell \cA$ because $(\prod_{i\in S} f_i)h\in \cH$ holds whenever $|S|\ell + \deg(h)\le d$ and thus $\pE_D\left[(\prod_{i\in S} f_i)h\right]\ge 0$. However, we also have $D\not\models_{\ell'} \cB$ since $L(g) < 0$. 
\end{proof}

\begin{proof}[Proof of Lemma~\ref{lem:matrix_sos_soundness}]\TMnote{maybe expand this a bit}Let $\cA = \{f_1\ge 0,\dots, f_k\ge 0\}$. For any vector $z\in \R^p$, we prove $\inner{z,\pE\left[M\right]z} = \pE\left[z^{\trans}Mz\right]\ge 0$. Indeed, $\cA\cal\vdash_{\ell} M$ implies the existence of $q_i,v_i$'s that satisfy equation~\eqref{eqn:def:matrix_sos}, where $q_i$ can be written as $q_i = \sum_{S}p_{i,S}\prod_{j\in S} f_j$. Therefore, $\pE\left[z^{\trans}Mz\right] = \pE\left[\sum_i\sum_{S}\inner{z,v_i(x)}^2p_{i,S}\prod_{j\in S} f_j\right]$. For fixed $i,S$, we have that $\deg(\inner{z,v_i(x)}^2) = 2\deg(v_i)$, and $|S|\le \ell'-2\deg(v_i)$. Therefore, we have $|S|\ell + 2\deg(v_i)\le \ell \ell'\le d$, and by $D\models_{\ell}\cA$ we obtain that  $\pE\left[\sum_i\sum_{S}\inner{z,v_i(x)}^2p_{i,S}\prod_{j\in S} f_j\right]\ge 0$, which completes the proof. %
\end{proof}

\end{document}

